\documentclass[11pt]{article}

\usepackage[T1]{fontenc} %
\usepackage{lmodern} %
\usepackage[latin1]{inputenc} %
\usepackage[margin=1in]{geometry} %
\usepackage[pdftex]{graphicx} %
\usepackage{caption} %
\usepackage{subcaption} %
\usepackage{amsmath,amssymb,amsthm} %
\usepackage{paralist} %
\usepackage{xspace} %
\usepackage[absolute,overlay]{textpos} %
\usepackage{tikz} %
\usepackage[english]{babel} %
\usepackage{textcomp} %
\usepackage{refcount} %
\usepackage{hyperref} %

\usepackage[backgroundcolor=white,linecolor=black,textsize=scriptsize]{todonotes} %
\theoremstyle{plain} %
\newtheorem{theorem}{Theorem} %
\newcounter{lemmacounter} %
\newtheorem{lemma}[lemmacounter]{Lemma} %
\newtheorem{cor}{Corollary}
\theoremstyle{definition} %

\newenvironment{relemma}[1]{ %
  \newcounter{tmpcounter} %
  \setcounter{tmpcounter}{\thelemmacounter} %
  \setcounterref{lemmacounter}{#1} %
  \addtocounter{lemmacounter}{-1}
  \begin{lemma}} { %
  \end{lemma} %
  \setcounter{lemmacounter}{\thetmpcounter}} %

\newcommand*\circled[1]{ %
  \protect\tikz[baseline=(char.base)]{ %
    \protect\node[shape=circle,draw,inner sep=0.2pt] (char) {#1};}} %

\setdefaultitem{\textopenbullet}{}{}{}

\newcommand{\1}[1]{{\normalfont \ensuremath{#1^{\tiny\circled{1}}}}} %
\newcommand{\2}[1]{{\normalfont \ensuremath{#1^{\tiny\circled{2}}}}} %
\newcommand{\proj}[2]{\ensuremath{\left.#1\right|_{#2}}} %
\newcommand{\eps}{\varepsilon}
\DeclareMathOperator{\skel}{skel} %
\DeclareMathOperator{\rep}{rep} %
\DeclareMathOperator{\ex}{exp} %
\DeclareMathOperator{\fixed}{fixed} %
\DeclareMathOperator{\id}{id} %
\DeclareMathOperator{\depth}{depth} %
\DeclareMathOperator{\level}{level} %
\DeclareMathOperator{\resp}{resp} %

\begin{document}

\author{Thomas Bläsius, Ignaz Rutter}
\title{Simultaneous PQ-Ordering \\with Applications to Constrained Embedding
  Problems\thanks{Research was partially supported by
    EUROGIGA project GraDR 10-EuroGIGA-OP-003.}}

\date{Karlsruhe Institute of Technology (KIT)\\%
\texttt{firstname.lastname@kit.edu}}

\maketitle

\begin{abstract}
  In this paper, we define and study the new problem {\sc Simultaneous
    PQ-Ordering}.  Its input consists of a set of PQ-trees, which
  represent sets of circular orders of their leaves, together with a
  set of child-parent relations between these PQ-trees, such that the
  leaves of the child form a subset of the leaves of the parent.  {\sc
    Simultaneous PQ-Ordering} asks whether orders of the leaves of
  each of the trees can be chosen \emph{simultaneously}, that is, for
  every child-parent relation the order chosen for the parent is an
  extension of the order chosen for the child.
  We show that {\sc Simultaneous PQ-Ordering} is $\mathcal
  {NP}$-complete in general and that it is efficiently solvable for a
  special subset of instances, the \emph{\mbox{2-fixed} instances}.
  We then show that several constrained embedding problems can be
  formulated as such 2-fixed instances.

  In particular, we obtain a linear-time algorithm for {\sc Partially
    PQ-Constrained Planarity} for biconnected graphs, a common
  generalization of two recently considered embedding
  problems~\cite{Testingplanarityof-Angelini.etal(10),gkm-ptoei-08},
  and a quadratic-time algorithm for {\sc Simultaneous Embedding with
    Fixed Edges} for biconnected graphs with a connected intersection;
  formerly only the much more restricted case that the intersection is
  biconnected was known to be efficiently
  solvable~\cite{TestingSimultaneousEmbeddability-Angelini.etal(11),TestingSimultaneousPlanarity-Haeupler.etal(10)}.
  Both results can be extended to the case where the input graphs are
  not necessarily biconnected but have the property that each
  cutvertex is contained in at most two non-trivial blocks.  This
  includes for example the case where both graphs have maximum
  degree~5.  Moreover, we give an optimal linear-time algorithm for
  recognition of simultaneous interval graphs, improving upon a
  recent~$O(n^2 \log n)$-time algorithm due to Jampani and
  Lubiw~\cite{SimultaneousIntervalGraphs-Jampani.Lubiw(10)} and show
  that this can be used to also solve the problem of extending partial
  interval representations of graphs with~$n$ vertices and~$m$ edges
  in time~$O(n+m)$, improving a recent result of Klav{\'i}k et
  al.~\cite{ExtendingPartialRepresentations-Klavik.etal(11)}.
\end{abstract}

\section{Introduction}
\label{sec:introduction}

Many types of data can be formulated as graphs, such as UML-diagrams
in software engineering, evolutionary trees in biology, communication
networks or relationships in social networks, to name a few.  For a
human it is nearly impossible to extract useful information out of
such a graph by looking at the pure data.  The way a graph is
interpreted crucially relies on its visualization.  Besides that,
creating a small chip with a large number of transistors on it is
closely related to the problem of drawing a graph with few bends, few
crossings and high resolution (ratio between smallest and largest
distances).  Thus, drawing graphs and particularly drawing planar
graphs is an important field of research.  However, in many
applications one needs not only to find a drawing of a given graph but
also satisfy additional conditions that are for example specified by a
user in an interactive graph drawing system or stem from technical
restrictions, such as port-constraints, requiring wires to attach to a
specific side of a component~\cite{cgmsw-cmldh-11, sfhm-pchld-10,
  gkm-ptoei-08}.  Other examples are
drawings that must be compatible with a previously fixed drawing of a
subpart of the system\cite{Testingplanarityof-Angelini.etal(10),
  kuratowski-typetheoremplanarity-Jelinek.etal(11)}.  Closely related
to constrained embeddings are simultaneous embeddings, which ask for a
given set of graphs sharing some vertices and edges whether they can
be drawn simultaneously, such that the common parts are drawn the same
in all drawings.  This is for example important to compare different
snapshots of a graph that changes dynamically over time.

The constrained embedding problem we consider involves PQ-trees.  In a
PQ-tree every inner node is either a P- or a Q-node and the order of
edges around a P-node can be chosen arbitrarily, whereas the order of
edges around a Q-node is fixed up to reversal;
Figure~\ref{fig:introduction-pq} depicts an example.  Such a PQ-tree
represents a set of possible orders of its leaves.  We consider the
problem {\sc Partially PQ-Constrained Planarity} having as input a
graph $G$ together with a PQ-tree $T(v)$ for every vertex $v$ with a
subset of edges incident to $v$ as leaves.  Thus $T(v)$ restricts the
possible orders of these edges to the orders that are represented by
it.  The question is, whether $G$ has a planar drawing respecting
these restrictions.  Additionally, we consider the problem {\sc
  Simultaneous Embedding with Fixed Edges} ({\sc SEFE}) having two
planar graphs $\1G$ and $\2G$ with a common subgraph $G$ as input,
asking whether planar drawings of $\1G$ and $\2G$ exist, such that the
drawing of $G$ is the same in both; see
Figure~\ref{fig:introduction-sefe} for an example.  Jünger and
Schulz~\cite{IntersectionGraphsin-Juenger.Schulz(09)} show that if $G$
is connected, this amounts to deciding whether $\1G$ and~$\2G$ admit
embeddings whose restrictions to~$G$ coincide.  If $\1G$ and~$\2G$ are
biconnected, then the possible circular orders of edges around each
vertex can be represented by a PQ-tree, and if additionally~$G$ is
connected, we seek orderings for each of the trees that together form
embeddings of $\1G$ and~$\2G$, respectively, and coincide on the
common graph.

As a common abstraction of these problems, we introduce the auxiliary
problem {\sc Simultaneous PQ-Ordering}, which we sketch in the
following.  Given a PQ-tree $T$ with the leaves $L$ and another
PQ-tree $T'$ with leaves $L' \subseteq L$, called a \emph{child} of
$T$, are there orders $O$ and $O'$ of the leaves $L$ and $L'$
represented by the PQ-trees $T$ and $T'$, respectively, such that the
order $O$ \emph{extends} $O'$?  In this case, we say that the relation
is \emph{satisfied}.  This question is fairly easy to answer, but what
happens if $T$ has more than one child or if $T'$ has additional
parents?  The problem {\sc Simultaneous PQ-Ordering} asks for a given
collection of PQ-trees with child-parent relations specified by a DAG,
whether orders for the PQ-trees can be chosen \emph{simultaneously},
such that each child-parent relation is satisfied; see
Section~\ref{sec:simult-pq-order} for a precise definition.

Applications of PQ-trees include also the recognition of
\emph{interval graphs}, i.e., graph that admit a representation where
each vertex corresponds to an interval in the reals, such that two
vertices are adjacent if and only if the corresponding intervals
intersect.  It is thus not surprising that also constrained interval
representation problems can be formulated within the framework of {\sc
  Simultaneous PQ-Ordering}.

\begin{figure}[tb]
  \centering
  \subcaptionbox{ \label{fig:introduction-pq}}
  {\includegraphics[page=1]{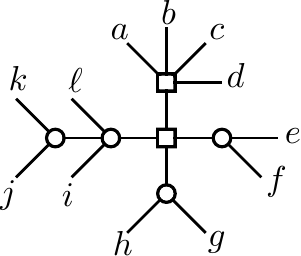}}\hspace{3em}
  \subcaptionbox{ \label{fig:introduction-sefe}}
  {\includegraphics[page=2]{fig/introduction}}
  \caption{(\subref{fig:introduction-pq}) A PQ-tree with leaves $\{a, \dots,
    \ell\}$ where P- and Q-nodes are depicted as circles and boxes,
    respectively.  For example the degree-5 Q-node at the top enforces the
    leaves $a, b, c, h$ to occur in this or its reversed order.  Furthermore,
    the two P-nodes on the left enforce the leaves $i, j, k, \ell$ to appear
    consecutively.  (\subref{fig:introduction-sefe}) Drawings of two graphs
    $\1G$ and $\2G$ on the common node set $\{1, \dots, 8\}$.  Although some of
    the vertices are drawn to similar positions in both drawings, it is hard to
    identify the differences and similarities between the two graphs.  This is
    much easier in the {\sc SEFE} on the right.}
  \label{fig:introduction}
\end{figure}

\subsection{Related Work}
\label{sec:related-work}

Since this work touches several different topics we consider the related work
about constrained embedding problems, simultaneous embedding problems,
PQ-trees and interval graphs separately.

\paragraph{Constrained Embedding.}

Constrained embedding problems in general ask for a given planar graph
whether it can be drawn without crossings in the plane, satisfying
some additional constraints.  Pach and Wenger show that every planar
graph can be draw crossing free even if the vertex positions are
prespecified by the
application~\cite{EmbeddingPlanarGraphs-Pach.Wenger(98)}.
Unfortunately, such a drawing can require linearly many bends per
edge.  Kaufmann and Wiese prove that two bends per edge are sufficient
if only the set of points in the plane is given whereas the mapping of
the vertices to these points can be
chosen~\cite{EmbeddingVerticesat-Kaufmann.Wiese(02)}.  Another
constrained embedding problem is {\sc Partially Embedded Planarity}
asking whether a planar drawing of a subgraph can be extended to a
planar drawing of the whole graph.  Angelini et al. give a linear-time
algorithm for testing {\sc Partially Embedded
  Planarity}~\cite{Testingplanarityof-Angelini.etal(10)} and Jelínek
et al. give a characterization by forbidden substructures similar to
Kuratowski's
theorem~\cite{kuratowski-typetheoremplanarity-Jelinek.etal(11)}.  The
problem {\sc PQ-Constrained Planarity} has as input a planar graph $G$
and a PQ-tree $T(v)$ for every vertex $v$ of $G$, such that the leaves
of $T(v)$ are exactly the edges incident to~$v$.  {\sc PQ-Constrained
  Planarity} asks whether $G$ has a planar drawing such that the order
of incident edges around every vertex $v$ is represented by the
PQ-tree $T(v)$.  Gutwenger et al. show that {\sc PQ-Constrained
  Planarity} can be solved in linear time by simply replacing every
vertex by a gadget and testing planarity of the resulting
graph~\cite{gkm-ptoei-08} (their main result
is a solution for {\sc Optimal Edge Insertion} with these
constraints).  Furthermore, they show how to deal with {\sc
  PQ-Constrained Planarity} if additionally the orientations of some
Q-nodes are fixed.

\paragraph{Simultaneous Embedding.}

Besides {\sc Simultaneous Embedding with Fixed Edges} ({\sc SEFE}) there are
other simultaneous embedding problems such as {\sc Simultaneous Embedding}
only requiring the common vertices to be drawn at the same
position~\cite{SimultaneousEmbeddingof-Erten.Kobourov(05),
  SimultaneousEmbeddingof-DiGiacomo.Liotta(07)}, and {\sc Simultaneous
  Geometric Embedding} requiring the edges to be straight-line
segments~\cite{SimultaneousEmbeddingof-Erten.Kobourov(05),
  simultaneousplanargraph-Brass.etal(07),
  SimultaneousGeometricGraph-Estrella-Balderrama.etal(08),
  Twotreeswhich-Geyer.etal(09), TreeandPath-Angelini.etal(11)}.  We only
consider {\sc SEFE} and there are essentially three types of results dealing
with this problem.  First, graph classes always having a
SEFE~\cite{SimultaneousEmbeddingof-Erten.Kobourov(05),
  SimultaneousEmbeddingof-DiGiacomo.Liotta(07),
  EmbeddingGraphsSimultaneously-Frati(07),
  SPQR-TreeApproachto-Fowler.etal(09),
  IntersectionGraphsin-Juenger.Schulz(09),
  Characterizationsofrestricted-Fowler.etal(11)}.  Second, graph classes
containing counter examples~\cite{simultaneousplanargraph-Brass.etal(07),
  EmbeddingGraphsSimultaneously-Frati(07),
  IntersectionGraphsin-Juenger.Schulz(09),
  Characterizationsofrestricted-Fowler.etal(11)}.  Third, results on the
complexity of the decision problem {\sc SEFE}.  Gassner et al. show that it is
$\mathcal {NP}$-complete to decide whether three or more graphs have a {\sc
  SEFE}~\cite{SimultaneousGraphEmbeddings-Gassner.etal(06)}.  Fowler at
al. show how to solve {\sc SEFE} efficiently, if $\1G$ and $G$ have at most
two and one cycles, respectively~\cite{SPQR-TreeApproachto-Fowler.etal(09)}
and Fowler et al. give an algorithm testing {\sc SEFE} if both graphs are
outerplanar~\cite{Characterizationsofrestricted-Fowler.etal(11)}.  Haeupler et
al. solve {\sc SEFE} in linear time for the case that the common graph is
biconnected~\cite{TestingSimultaneousPlanarity-Haeupler.etal(10)}.  Angelini
et al. obtain the same result with a completely different
approach~\cite{TestingSimultaneousEmbeddability-Angelini.etal(11)a}.  They
additionally solve the case where the common graph is a star.

\paragraph{PQ-Trees.}

PQ-Trees were originally introduced by Booth and
Lueker~\cite{TestingConsecutiveOnes-Booth.Lueker(76)}.  They were
designed to decide whether a set $L$ has the {\sc Consecutive Ones}
property with respect to a family $\mathcal S = \{S_1, \dots, S_k\}$
of subsets $S_i \subseteq L$.  The set $L$ has this property if a
linear order of its elements can be found, such that the elements in
each subset $S_i \in \mathcal S$ appear consecutively.  Booth and
Lueker showed how to solve {\sc Consecutive Ones} in linear time.
Furthermore, they showed that all linear orders of the elements in $L$
in which each subset $S_i \in \mathcal S$ appears consecutively can be
represented by a PQ-tree having the elements in $L$ as leaves.
Besides testing planarity in linear time they were able to decide in
linear time if a given graph is an interval graph.  In the original
approach by Booth and Lueker, the PQ-trees were rooted, representing
linear orders of their leaves.  However, they can also considered to
be unrooted representing circular
orders~\cite{PlanarityAlgorithmsvia-Haeupler.Tarjan(08)}.  Unrooted
PQ-trees are sometimes also called
PC-trees~\cite{PC-Treesvs.PQ-Trees-Hsu(01),
  PQTreesPC-Hsu.McConnell(01), PCtreesand-Hsu.McConnell(03)}.  In most
cases we will use unrooted PQ-trees representing circular orders while
the same results can be achieved for rooted PQ-trees representing
linear orders by simply adding a single leaf (see
Section~\ref{sec:pq-trees} for further details).

\paragraph{Interval Graphs.}

Fulkerson and Gross gave a characterization of interval graphs in
terms of the {\sc Consecutive Ones}
property~\cite{IncidenceMatricesand-Fulkerson.Gross(65)}, enabling
Booth and Lueker to recognize them in linear time using
PQ-trees~\cite{TestingConsecutiveOnes-Booth.Lueker(76)}.  More
recently, Klavík et al. give an $\mathcal O(nm)$ time algorithm
testing whether a given interval representation of a subgraph can be
extended to an interval representation of the whole
graph~\cite{ExtendingPartialRepresentations-Klavik.etal(11)}.  Jampani
and Lubiw show that simultaneous interval graphs can be recognized in
$\mathcal O(n^2 \log n)$
time~\cite{SimultaneousIntervalGraphs-Jampani.Lubiw(10)}.  Two graphs
with common vertices are simultaneous interval graphs, if they have
interval representations representing the common vertices by the same
intervals.

\subsection{Contribution and Outline}
\label{sec:our-contribution}

We first define basic notation and present known results, which we use
throughout this paper, in Section~\ref{sec:preliminaries}.  In
Section~\ref{sec:simult-pq-order}, we first give a precise problem
definition for {\sc Simultaneous PQ-Ordering} and show that it is
$\mathcal{NP}$-complete in general; see
Section~\ref{sec:np-hardn-simult}.  In the remainder of that section,
which forms the main part of this paper, we characterize a subset of
``simple'' instances, the so-called \emph{\mbox{2-fixed} instances}, for
which a solution can be computed efficiently, namely in quadratic
time.  We present several applications in
Section~\ref{sec:applications}, where we show how to formulate various
problems as 2-fixed instances within the framework of {\sc
  Simultaneous PQ-Ordering}, thus yielding efficient algorithms to
solve them.  The algorithms obtained in this way either solve problems
that were not known to be efficiently solvable or significantly
improve over the previously best running times.

In particular, we show that {\sc Partially PQ-Constrained Planarity}
can be solved in linear time for biconnected graphs; see
Section~\ref{sec:partially-pq-constr-planarity}.  Note that this
problem can be seen as a common generalization of the constrained
embedding problems {\sc Partially Embedded
  Planarity}~\cite{Testingplanarityof-Angelini.etal(10),
  kuratowski-typetheoremplanarity-Jelinek.etal(11)} and {\sc
  PQ-Constrained
  Planarity}~\cite{gkm-ptoei-08}.  The
former completely fixes the order of some edges around a vertex, the
latter partially fixes the order of all edges around a vertex.  {\sc
  Partially PQ-Constrained Planarity} partially fixes the order of
some edges.  Similar to the work of Gutwenger et al., we can also
handle the case where some Q-nodes have a fixed orientation.  In
addition to that, {\sc SEFE} can be formulated as a 2-fixed instance
of {\sc Simultaneous PQ-Ordering}, if both graphs are biconnected and
the common graph is connected, thus yielding a quadratic-time
algorithm for this case; see
Section~\ref{sec:simult-embedd-biconn-conn}.  This significantly
extends the results requiring that the common graph is
biconnected~\cite{TestingSimultaneousPlanarity-Haeupler.etal(10),
  TestingSimultaneousEmbeddability-Angelini.etal(11)a} for the
following reason.  If the intersection $G$ of two graphs $\1G$ and
$\2G$ is biconnected, it is completely contained in a single maximal
biconnected component of $\1G$ and $\2G$, respectively.  Thus, testing
{\sc SEFE} for $\1G$ and $\2G$ is equivalent to testing it for these
two biconnected components, since all remaining biconnected components
can be attached if and only if they are planar.  Moreover, we improve
the previously best algorithms for recognizing simultaneous interval
graphs~\cite{SimultaneousIntervalGraphs-Jampani.Lubiw(10)}
from~$O(n^2 \log n)$ to linear(Section~\ref{sec:simult-interv-graphs})
and for extending partial interval
representations~\cite{ExtendingPartialRepresentations-Klavik.etal(11)}
from~$O(nm)$ to $O(n+m)$ (Section~\ref{sec:extend-part-interv}).  We
show that the results for {\sc Partially PQ-constrained Planarity} and
{\sc SEFE} still hold if the input graphs have the property that each
cutvertex is contained in at most two nontrivial blocks in
Section~\ref{sec:thoughts-gener}.  We conclude with some prospects for
future work and some open question in Section~\ref{cha:conclusion}.

Note that all applications follow easily from the main results in
Section~\ref{sec:simult-pq-order}.  The formulations as instances of
{\sc Simultaneous PQ-Ordering} we use are straightforward and can
easily be verified to be 2-fixed, at which point the machinery
developed in the main part of this paper takes over.

\section{Preliminaries}
\label{sec:preliminaries}

In this chapter we define the notation and provide some basic tools we
use in this work.  Section~\ref{sec:graphs-planar-graphs} deals with
graphs and their connectivity, planar graphs and embeddings of planar
graphs, directed acyclic graphs and trees.  Linear and circular orders
and how permutations act on them are considered in
Section~\ref{sec:circ-line-orders}.  PQ-trees are defined in
Section~\ref{sec:pq-trees}.  Furthermore, the relation between rooted
and unrooted PQ-trees is described and operations that can be applied
to them are defined.  In Section~\ref{sec:spqr-trees-basics} we give a
short introduction to SPQR-trees, which are used to represent all
embeddings of a planar graph.  In Section~\ref{sec:spqr-trees} we show
how PQ- and SPQR-trees are related.

\subsection{Graphs, Planar Graphs, DAGs and Trees}
\label{sec:graphs-planar-graphs}

A graph $G=(V, E)$ is \emph{connected} if there is a path from $u$ to
$v$ for every pair of vertices $u, v \in V$.  A \emph{separating
  $k$-set} is a set of $k$ vertices whose removal disconnects $G$.
Separating 1-sets and 2-sets are called \emph{cutvertices} and
\emph{separation pairs}, respectively.  A graph is \emph{biconnected}
if it is connected and does not have a cutvertex and it is
\emph{triconnected} if it additionally does not have a separation
pair.  The maximal connected subgraphs (with respect to inclusion) of
$G$ are called \emph{connected components} and the maximal biconnected
subgraphs are called \emph{blocks}.  A complete subgraph of $G$ is
called a \emph{clique}.  A clique is \emph{maximal} if it is not
contained in a larger clique.  Sometimes we also use the term
\emph{node} instead of vertex to emphasize that it represents a larger
object.

A \emph{drawing} of a graph $G$ is a mapping of every vertex $v$ to a
point $(x_v, y_v)$ in the plane and a mapping of every edge $\{u, v\}$
to a Jordan curve having $(x_u, y_u)$ and $(x_v, y_v)$ as endpoints.
A drawing of $G$ is \emph{planar} if edges do not intersect except at
common endpoints.  The graph $G$ is \emph{planar} if a planar drawing
of $G$ exists.  Consider $G$ to be a connected planar graph.  Every
planar drawing of $G$ splits the plane into several connected regions,
called the \emph{faces} of the drawing.  Exactly one of these faces,
called the \emph{outer face}, is unbounded.  The boundary of each face
is a directed cycle in $G$ and two faces in different drawings are
said to be the same if they have the same boundary.  Additionally,
every planar drawing of $G$ induces for every vertex an order of
incident edges around it and two drawings inducing the same order for
every vertex are called \emph{combinatorially equivalent}.  It is
clear that two combinatorially equivalent drawings have the same
faces, which implies that they have the same topology since $G$ is
connected.  Note that being combinatorially equivalent is an
equivalence relation and the equivalence classes are called
\emph{combinatorial embeddings} of~$G$.  A combinatorial embedding
together with the choice of an outer face is a \emph{planar
  embedding}.  In most cases we do not care about which face is the
outer face, thus we mean a combinatorial embedding by simply saying
embedding.

In a directed graph we call the edges \emph{arcs} and an arc from the
\emph{source}~$u$ to the \emph{target}~$v$ is denoted by $(u, v)$.  A
directed graph $G$ without directed cycles is called \emph{directed
  acyclic graph (DAG)}.  Let $u$ and $v$ be vertices of a DAG $G$ such
that there exists a directed path from $u$ to $v$.  Then $u$ is called
an \emph{ancestor} of~$v$ and~$v$ a \emph{descendant} of $u$.  If the
arc $(u,v)$ is contained in $G$, then $u$ is a \emph{parent} of $v$
and~$v$ is a \emph{child} of $u$.  A vertex $v$ in a DAG $G$ is called
\emph{source} (\emph{sink}) if it does not have parents (children).
Note that this overloads the term source, but it will be clear from
the context which meaning is intended.  A \emph{topological ordering}
of a DAG $G$ is an ordering of its vertices such that $u$ occurs
before $v$ if $G$ contains the arc $(u, v)$.  By saying that a DAG is
processed \emph{top-down} (\emph{bottom-up}) we mean a traversal of
its vertices according to a (reversed) topological ordering.  Let $G$
be a DAG and let $v$ be a vertex.  The \emph{level} of $v$, denoted by
$\level(v)$, is the length of the shortest directed path from a source
to $v$.  The \emph{depth} of $v$, denoted by $\depth(v)$, is the
length of the longest directed path from a source to $v$.  Note that
the level and the depth have in a sense contrary properties.  Let $v$
be a vertex in $G$ and let $u$ be a parent of $v$.  Then the depth of
$u$ is strictly smaller than the depth of $v$ whereas the level
decreases by at most one: $\depth(u) < \depth(v)$; $\level(u) \ge
\level(v) - 1$.  By the level and the depth of the DAG $G$ itself we
mean the largest level and the largest depth any vertex in $G$,
respectively.

An \emph{(unrooted) tree} $T$ is a connected graph without cycles.
The degree-1 vertices of $T$ are called \emph{leaves} and the other
are \emph{inner vertices}.  A tree $T$ together with a special vertex
$r$, called the \emph{root} of $T$, is a \emph{rooted tree}.  A rooted
tree can be seen as DAG by directing all edges towards the leaves of
the tree.  Then the terms top-down, bottom-up, ancestor, descendant,
child, and parent can be defined as for DAGs.  Note that a tree with
$n$ vertices has $m = n - 1$ edges.  However, in general, the ratio
between the number of vertices (or edges) and the number of leaves is
unbound (consider a tree consisting of a single path).  We will use
the following lemma, which for trees that do not contain degree-2
vertices bounds the tree size in terms of the number of leaves.

\begin{lemma}
  \label{lem:tree-without-deg-2}
  A tree with $n_1$ leaves and without degree-2 vertices has at most
  $n_1 - 2$ inner vertices and at most $2n_1 - 3$ edges.
\end{lemma}
\begin{proof}
  Let $T$ be a tree with $n_1$ leaves and the maximum number of edges
  possible.  Then every inner vertex in $T$ has degree~3, because a
  vertex with four incident edges $e_1,\dots,e_4$ could be split into
  two vertices with incident edges $e_1, e_2$ and $e_3, e_4$
  respectively, plus an additional edge connecting them.  Clearly, $T$
  has also the maximum number of inner vertices for the fixed number
  of leaves $n_1$.  Let now~$n$ and~$m$ denote the total number of
  vertices and edges of~$T$, respectively, and let~$n_3$ denote the
  number of vertices of degree~3 in~$T$.  Since every vertex of~$T$
  has either degree~3 or is a leaf, we have $n = n_1 + n_3$.  Since
  $T$ is a tree we have $m = n - 1$ and, by counting the edge
  incidences, we get $2m = n_1 + 3n_3$.  Together these three
  equations imply $n_3 = n_1 - 2$, and therefore $m = 2n_1 - 3$.
\end{proof}

\subsection{Linear and Circular Orders and Permutations}
\label{sec:circ-line-orders}

Let $L$ be a finite set (all sets we consider are finite).  A sequence
$O$ of all elements in $L$ specifies a relation ``$\le$'' on $L$ in
the way that $\ell_1 \le \ell_2$ for $\ell_1 \not= \ell_2 \in L$ if
and only if $\ell_2$ occurs behind $\ell_1$ in~$O$.  Such a relation
is called \emph{linear order} (or also \emph{total order}) and is
identified with the sequence $O$ specifying it.  Let $O_1$ and $O_2$
be two linear orders on $L$ and let $\ell \in L$ be an arbitrary
element.  Let further $O_i'$ (for $i = 1, 2$) be the order that is
obtained from $O_i$ by concatenating the smallest suffix containing
$\ell$ with the largest prefix not containing~$\ell$.  We call $O_1$
and $O_2$ \emph{circularly equivalent} if $O_1'$ and $O_2'$ are the
same linear order.  Not that this is a equivalence relation not
depending on the chosen element $\ell$.  The equivalence classes are
called \emph{circular orders}.  For example for $L = \{a, \dots, e\}$
the orders $O_1 = baedc$ and $O_2 = dcbae$ are circularly equivalent
and thus define the same circular order since $O_1' = O_2' = aedcb$,
if we choose $\ell = a$.  In most cases we consider circular orders.
Unless stated otherwise, we refer to circular orders by simply writing
orders.  Note that a linear order can be seen as a graph with vertex
set $L$ consisting of a simple directed path, whereas a circular order
corresponds to a graph consisting of a simple directed cycle
containing $L$ as vertices, see Figure~\ref{fig:preliminaries-orders}
for an example.  Let $L$ be a set and let $O$ be a circular order of
its elements.  Let further $S \subseteq L$ be a subset and let $O'$ be
the circular order on $S$ that is induced by $O$.  Then $O'$ is a
\emph{suborder} of $O$ and $O$ is an \emph{extension} of $O'$.  Note
that $S$ does not really need to be a subset of $L$.  Instead it can
also be an arbitrary set together with an injective map $\varphi : S
\rightarrow L$.  We overload the terms suborder and extension for this
case by calling an order $O'$ of $S$ a \emph{suborder} of $O$ and $O$
an \emph{extension} of $O'$ if $\varphi(O')$ is a suborder of $O$,
where $\varphi(O')$ denotes the order obtained from $O'$ by applying
$\varphi$ to each element.

\begin{figure}[tb]
  \centering \subcaptionbox{ \label{fig:preliminaries-orders}}
  {\includegraphics[page=1]{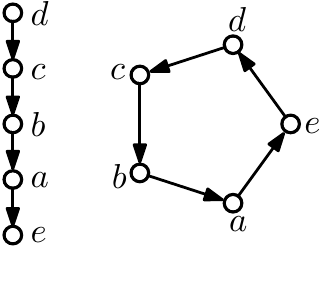}}\hspace{3em}
  \subcaptionbox{ \label{fig:preliminaries-permutations}}
  {\includegraphics[page=2]{fig/preliminaries-orders-permutations}}
  \caption{(\subref{fig:preliminaries-orders}) The interpretation of
    the linear and circular order $dcbae$ as simple path and simple
    cycle, respectively.  (\subref{fig:preliminaries-permutations})
    The permutation $\varphi = (aec) \circ (bfd)$ on the left can be
    seen as as a clockwise rotation by~2 of the circular order
    $abcdef$, and thus is order preserving, whereas the permutation
    $\varphi = (af) \circ (be) \circ (cd)$ in the middle is order
    reversing.  However, $\varphi = (af) \circ (be) \circ (cd)$ is not
    only order reversing but also order preserving (rotation by~3)
    with respect to the order $abcfed$ as shown on the right.  The
    permutations $\varphi$ are depicted as thin arrows with empty
    arrowheads and the different permutation cycles are distinct by
    solid, dashed and dotted lines.}
  \label{fig:preliminaries-orders-permutations}
\end{figure}

In the following, we consider permutations on the set $L$ and provide
some basic properties on how these permutations act on circular orders
of $L$.  Let $L$ be a set and let $\varphi : L \rightarrow L$ be a
permutation.  The permutation $\varphi$ can be decomposed into $r$
disjoint \emph{permutation cycles} $\varphi = (\ell_1 \varphi(\ell_1)
\dots \varphi^{k_1}(\ell_1)) \circ \dots \circ (\ell_r \varphi(\ell_r)
\dots \varphi^{k_r}(\ell_r))$.  We call $k_i$ the \emph{length} of the
cycle $(\ell_i \varphi(\ell_i) \dots \varphi^{k_i}(\ell_i))$.
Fixpoints for example form a permutation cycle of length~1.  We can
compute this decomposition by starting with an arbitrary element
$\ell$ and applying $\varphi$ iteratively until we reach $\ell$ again.
Then we continue with an element not contained in any permutation
cycle so far to obtain the next cycle.  Now consider a circular order
$O$ of the elements in $L$.  The permutation $\varphi$ is called
\emph{order preserving} with respect to $O$ if $\varphi(O) = O$.  It
is called \emph{order reversing} with respect to $O$ if $\varphi(O)$
is obtained by reversing~$O$.  Note that for a fixed order $O$ the
order preserving and order reversing permutations are exactly the
rotations and reflections of the dihedral group, respectively (the
dihedral group is the group of rotations and reflections on a regular
$k$-gon).  If we interpret $O$ as a graph as mentioned above, that is,
a graph with vertex set $L$ consisting of a simple directed cycle, we
obtain that $\varphi$ is order preserving with respect to $O$ if it is
a graph isomorphism on this cycle, whereas the cycle is reversed if
$\varphi$ is order reversing with respect to $O$;
Figure~\ref{fig:preliminaries-permutations} depicts this
interpretation for an example.  We say that $\varphi$ is \emph{order
  preserving} or \emph{order reversing} if it is order preserving or
order reversing with respect to at least one order $O$.  In this
setting the order is not fixed and we want to characterize for a given
permutation if it is order preserving or order reversing and
additionally we want to find an order that is preserved or reversed,
respectively.  Note that not fixing the order has for example the
effect, that the same permutation $\varphi$ can be a rotation with
respect to one order and a reflection with respect to another, which
means that it can be order preserving and order reversing at the same
time.

\begin{lemma}
  \label{lem:permutation-order-preserving}
  A permutation $\varphi$ on the set $L$ is order preserving if and
  only if all its permutation cycles have the same length.
\end{lemma}
\begin{proof}
  Assume $\varphi$ consists of $r$ permutation cycles of length $k$,
  let $\ell_i$ be an element in the $i$th permutation cycle.  Then
  $\varphi$ is order preserving with respect to the following circular
  order.
  \[\ell_1 \dots
  \ell_r \quad \varphi(\ell_1) \dots \varphi(\ell_r) \quad \dots \quad
  \varphi^k(\ell_1) \dots \varphi^k(\ell_k)\]

  Assume we have a circular order $O = \ell_1 \dots \ell_n$ such that
  $\varphi(O) = O$.  We show that the permutation cycles of two
  consecutive elements $\ell_i$ and $\ell_{i+1}$ have the same size.
  This claim holds if $\ell_i$ and $\ell_{i+1}$ are contained in the
  same permutation cycle.  Assume they are in different permutation
  cycles with lengths $k_{i}$ and $k_{i+1}$, respectively, such that
  $k_i < k_{i+1}$.  Then $\varphi^{k_i}(\ell_{i+1}) \not= \ell_{i+1}$
  is not the successor of $\varphi^{k_i}(\ell_i) = \ell_i$ in $O$.
  Thus, $\varphi^{k_i}(O)$ cannot be the same circular order $O$ and
  hence $\varphi$ is not order preserving, which is a contradiction.
\end{proof}

\begin{lemma}
  \label{lem:permutation-order-reverting}
  A permutation $\varphi$ on the set $L$ is order reversing if and
  only if all its permutation cycles have length~2, except for at most
  two cycles with length~1.
\end{lemma}
\begin{proof}
  Assume we have $\varphi = (\ell_1 \ell_1') \circ \dots \circ (\ell_r
  \ell_r')$, $\varphi = (\ell) \circ (\ell_1 \ell_1') \circ \dots
  \circ (\ell_r \ell_r')$ or $\varphi = (\ell) \circ (\ell') \circ
  (\ell_1 \ell_1') \circ \dots \circ (\ell_r \ell_r')$.  Then
  $\varphi$ reverses the orders, $\ell_1 \dots \ell_r \ell_r' \dots
  \ell_1'$, $\ell_1 \dots \ell_r \ell \ell_r' \dots \ell_1'$ and
  $\ell_1 \dots \ell_r \ell \ell_r' \dots \ell_1' \ell'$,
  respectively.

  Now assume we have an order $O$ such that $\varphi$ is order
  reversing with respect to $O$, that is it is a reflection in the
  dihedral group defined by $O$.  Thus, $\varphi^2$ is the identity
  yielding that $\varphi$ cannot contain a permutation cycle of length
  greater than~2.  Furthermore, a reflection has at most two
  fixpoints.
\end{proof}

It is clear that the characterizations given in
Lemma~\ref{lem:permutation-order-preserving} and
Lemma~\ref{lem:permutation-order-reverting} can be easily checked in
linear time.  Additionally, from the proofs of both lemmas it is clear
how to construct an order that is preserved or reversed if the given
permutation is order preserving or order reversing, respectively.

\subsection{PQ-Trees}
\label{sec:pq-trees}

Given an unrooted tree $T$ with leaves $L$ having a fixed circular
order of edges around every vertex, that is, having a fixed
combinatorial embedding, the circular order of the leaves (as they
occur along the outer face of the embedding) is also fixed.  In an
\emph{unrooted PQ-tree} for some inner nodes, the \emph{Q-nodes}, the
circular order of incident edges is fixed up to reversal, for the
other nodes, the \emph{P-nodes}, this order can be chosen arbitrarily.
Hence, an unrooted PQ-tree represents a set of circular orders of its
leaves.  Given a set $L$, a set of circular orders $\mathcal L$ of $L$
is called \emph{PQ-representable}, if there is an unrooted PQ-tree
with leaves $L$ representing it.  Formally, the empty set, saying that
no order is possible, is represented by the \emph{null tree}, whereas
the \emph{empty tree} has the empty set as leaves and represents the
set containing only the empty order.  A simple example for an unrooted
PQ-tree is shown in Figure~\ref{fig:PQ-trees-unrooted-rooted-a}.  Note
that not every set of orders is PQ-representable; for example every
PQ-representable set of orderings must be closed under reversal.

In the same way, we can define a \emph{rooted PQ-tree} representing
sets of linear orders by replacing \emph{circular} by \emph{linear}
and additionally choosing an inner node of the PQ-tree as root.  There
is an equivalence between unrooted and rooted PQ-trees in the
following sense.  Let $T$ be an unrooted PQ-tree with leaves $L$,
representing the set of circular orders~$\mathcal L$.  If we choose
one leaf $\ell\in L$ to be the \emph{special leaf}, every circular
order in $\mathcal L$ can be seen as a linear order of $L' := L -
\ell$ by breaking the cycle at $\ell$.  Since every circular order in
$\mathcal L$ yields a different linear order, we obtain a bijection to
a set of linear orders $\mathcal L'$.  We can construct a rooted
PQ-tree $T'$ with the leaves $L'$ representing $\mathcal L'$ as
follows.  First, we choose the special leaf $\ell$ to be the root of
$T$.  Then, for every Q-node we obtain a linear order from the given
circular order by breaking the cycle at the (unique) parent.  Finally,
we remove $\ell$ and choose its (unique) child as the new root.
Hence, given an unrooted PQ-tree, we can work with its rooted
equivalent instead, by choosing one leaf to be the special leaf; see
Figure~\ref{fig:PQ-trees-unrooted-rooted} for an example.  Conversely,
rooted PQ-trees can be represented by unrooted ones by inserting a
single leaf.  In most cases we will work with unrooted PQ-trees
representing sets of circular orders.  Unless stated otherwise, we
thus refer to circular orders and unrooted PQ-trees if we write orders
and PQ-trees, respectively.

\begin{figure}[tb]
  \centering \subcaptionbox{ \label{fig:PQ-trees-unrooted-rooted-a}}
  {\includegraphics[page=1]{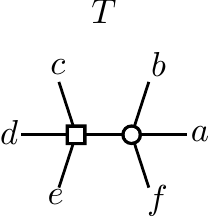}}\hspace{3em}
  \subcaptionbox{ \label{fig:PQ-trees-unrooted-rooted-b}}
  {\includegraphics[page=2]{fig/PQ-trees-unrooted-rooted}}
  \caption{(\subref{fig:PQ-trees-unrooted-rooted-a}) An unrooted
    PQ-tree $T$ with leaves $L=\{a,\dots,f\}$, where P- and Q-nodes
    are drawn as circles and boxes, respectively.  By choosing an
    order for $a, b, f$ and concatenating it with $cde$ or $edc$, we
    obtain all circular orders in~$\mathcal L$.
    (\subref{fig:PQ-trees-unrooted-rooted-b}) Choosing $a$ as the
    special leaf yields the rooted PQ-tree $T'$ with leaves
    $L'=\{b,\cdots,f\}$.  By choosing an arbitrary order for
    $b,\square,f$ where $\square$ stands for $cde$ or $edc$, we obtain
    all orders in $\mathcal L'$.  Note that this simply means to break
    the cyclic orders in $\mathcal L$ at the special leaf~$a$.}
  \label{fig:PQ-trees-unrooted-rooted}
\end{figure}

PQ-trees were introduced by Booth and
Lueker~\cite{TestingConsecutiveOnes-Booth.Lueker(76)} in the rooted
version.  Let $L$ be a finite set and let $\mathcal S = \{S_1, \dots,
S_k\}$ be a family of subsets $S_i \subseteq L$.  Booth and Lueker
showed that the set $\mathcal L$ containing all linear orders in which
the elements in each set $S_i$ appear consecutively is
PQ-representable.  Note that $\mathcal L$ could be the empty set,
since in no order all subsets $S_i$ appear consecutively, then
$\mathcal L$ is represented by the null tree.  This result can be
easily extended to unrooted PQ-trees and circular orders in which the
subsets~$\mathcal S$ appear consecutively, which will become clearer
in a moment.

As mentioned above, not every set of orders $\mathcal L$ is
PQ-representable, but we will see three operations on sets of orders
that preserve the property of being PQ-representable.  Given a subset
$S\subseteq L$, the \emph{projection} of $\mathcal L$ to $S$ is the
set of orders of $S$ achieved by restricting every order in $\mathcal
L$ to $S$.  The \emph{reduction} with $S$ is the subset of $\mathcal
L$ containing the orders where the elements of $S$ appear
consecutively.  Given two sets of orders $\mathcal L_1$ and $\mathcal
L_2$ on the same set $L$, their \emph{intersection} is simply
$\mathcal L_1 \cap \mathcal L_2$.  That projection, reduction and
intersection preserve the property of being PQ-representable can be
shown constructively.  But first we introduce the following notation,
making our life a bit easier.  Let $T$ be a PQ-tree with leaf set $L$,
representing $\mathcal L$, and let $\mu$ be an inner node with
incident edges $\varepsilon_1,\dots,\varepsilon_k$.  Removing
$\varepsilon_i$ splits $T$ into two components.  We say that the
leaves contained in the component not containing $\mu$ \emph{belong}
to $\varepsilon_i$ with respect to $\mu$, and we denote the set of
these leaves by $L_{\varepsilon_i,\mu}$.  In most cases it is clear
which node $\mu$ we refer to, so we simply write $L_{\varepsilon_i}$.
Note that the sets $L_{\varepsilon_i}$ form a partition of $L$.

\begin{compactdesc}
\item[Projection] Let $T$ be a PQ-Tree with leaves $L$, representing
  the set of orders $\mathcal L$.  The projection to $S\subseteq L$ is
  represented by the PQ-tree $T'$ that is obtained form $T$ by
  removing all leaves not contained in $S$ and simplifying the result,
  where simplifying means, that former inner nodes now having degree~1
  are removed iteratively and that degree-2 nodes together with both
  incident edges are iteratively replaced by single edges.  We denote
  the tree resulting from the projection of $T$ to $S$ by $\proj T S$
  and we often call $\proj T S$ itself the projection of $T$ to $S$.
\item[Reduction] Recall that the reduction with a set $S$ reduces a
  set of orders to these orders in which all elements in $S$ appear
  consecutively.  The reduction can be seen as the operation PQ-trees
  were designed for by Booth and
  Lueker~\cite{TestingConsecutiveOnes-Booth.Lueker(76)}.  They showed
  for a rooted PQ-tree $T$ representing the linear orders $\mathcal L$
  that the reduction to $S$ is again PQ-representable and the PQ-tree
  representing it can be computed in $\mathcal O(|L|)$ time.  For an
  unrooted PQ-tree $T$ we can consider the rooted PQ-tree $T'$ instead
  by choosing $\ell \in L$ as special leaf.  Since the reduction
  with~$L$ and~$S \setminus L$ are equivalent, we may assume without
  loss of generality that~$\ell \notin S$, and we obtain the reduction
  of~$T$ by reducing~$T'$ with~$S$, reinserting~$\ell$ and
  unrooting~$T'$ again.  This shows for a family of subsets $\mathcal
  S = \{S_1, \dots, S_k\}$ that the set containing all circular orders
  in which each subset $S_i \subseteq L$ appears consecutively can be
  represented by an unrooted PQ-tree $T$.  Thus, applying a reduction
  with $S$ to a given PQ-tree $T$ can be seen as adding the subset $S$
  to $\mathcal S$.  Therefore, we denote the result of the reduction
  of $T$ with $S$ by $T + S$ and we often call $T + S$ itself the
  reduction of $T$ with $S$.
\item[Intersection] For an inner node $\mu$, all leaves
  $L_{\varepsilon}$ belonging to an incident edge $\varepsilon$ appear
  consecutively in every order contained in~$\mathcal L$.
  Furthermore, if $\mu$ is a Q-node with two consecutive incident
  edges $\varepsilon$ and $\varepsilon'$, all leaves in
  $L_{\varepsilon}\cup L_{\varepsilon'}$ need to appear consecutively.
  On the other hand, if we have an order of $L$ satisfying these
  conditions for every inner node, it is contained in~$\mathcal L$.
  Hence, $T$ can be seen as a sequence of reductions applied to the
  set of all orders, which is represented by the star with a P-node as
  center.  Now, given two unrooted PQ-trees $T_1$ and $T_2$ with the
  same leaves, we obtain their intersection by applying the sequence
  of reductions given by $T_1$ to~$T_2$.  Note that the size of all
  these reductions can be quadratic in the size of~$T_1$.  However
  Booth showed how they can be applied consuming time linear in the
  size of $T_1$ and $T_2$~\cite{PQ-TreeAlgorithms-Booth(75)}.  We
  denote the intersection of $T_1$ and $T_2$ by $T_1 \cap T_2$.
\end{compactdesc}

Let $\proj T S$ be the projection of $T$ to $S\subseteq L$.  The
extension of an order of $S$ represented by $\proj T S$ to an order of
$L$ represented by $T$ is straightforward.  An inner node in $T$ is
either contained in $\proj T S$ or it was removed in the
simplification step.  If a Q-node in $T$ is also contained in $\proj T
S$, its orientation is determined by the orientation chosen in $\proj
T S$ and we call it \emph{fixed}, otherwise its orientation can be
chosen arbitrarily and we call it \emph{free}.  For a P-node not
contained in $\proj T S$, the order of incident edges can be chosen
arbitrarily.  If a P-node is contained in $\proj T S$, every incident
edge is either also contained, was removed or replaced (and the
replacement was not removed).  The order of the contained and replaced
edges is fixed, and the removed edges can be inserted arbitrarily.  We
call the removed edges (and the edges incident to removed P-nodes)
\emph{free} and all other edges \emph{fixed}.

Let $T + S$ be the reduction of a PQ-tree $T$ with leaves $L$ with the
subset $S \subseteq L$.  Choosing an order in the reduction $T + S$ of
course determines an order of the whole leaf set~$L$. Hence, it
determines the order of incident edges for every inner node in $T$.
For every Q-node $\mu$ in $T$ there exists exactly one Q-node in $T +
S$ determining its orientation, we call it the \emph{representative}
of $\mu$ with respect to the reduction with $S$ and denote it by
$\rep_S(\mu)$, where the index is omitted, if it is clear from the
context.  Note that one Q-node in $T + S$ can be the representative of
several Q-nodes in $T$.  For a P-node $\mu$ we cannot find such a
representative in $T + S$ since it may depend on several nodes in $T +
S$.  However, if we consider a P-node $\mu'$ in $T + S$ there is
exactly one P-node $\mu$ in $T$ that depends on $\mu'$.  We say that
$\mu'$ \emph{stems} from this P-node $\mu$.

\begin{figure}[bt]
  \centering
  \includegraphics[page=1]{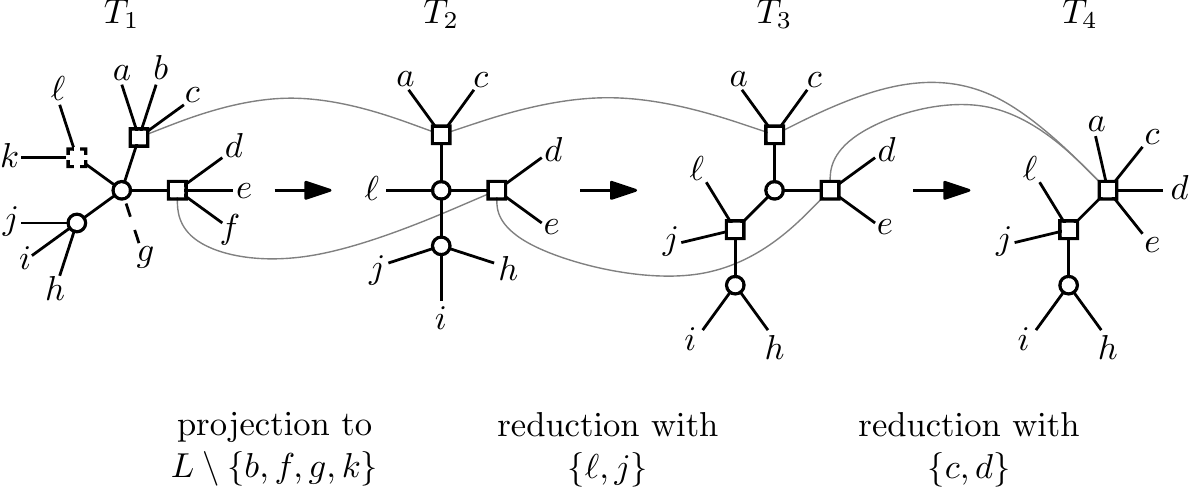}
  \caption{We start with the PQ-tree $T_1$ on the left and project it
    to $L\setminus \{b,f,g,k\}$ yielding $T_2$.  There is one Q-node
    and one edge incident to a P-node, both drawn dashed, that do not
    appear in $T_2$ and hence are free.  The trees $T_3$ and $T_4$ are
    obtained by applying reductions with $\{\ell,j\}$ and $\{c,d\}$ to
    $T_2$.  Note that the arrows (and even their transitive closure)
    can be interpreted as child-parent relation between the PQ-trees.
    Every fixed Q-node has a representative depicted by gray lines,
    whereas it is not so easy to find something similar for the
    P-nodes.}
  \label{fig:PQ-trees-restrictions-fixed-edges}
\end{figure}

The considerations concerning a PQ-tree $T$ with leaves $L$ together
with another PQ-tree $T'$ with leaves $L' \subseteq L$ that is a
projection or a reduction of $T$ can of course be extended to the case
where $T'$ is obtained from $T$ by a projection followed by a sequence
of reductions.  This can be further generalized to the case where $T$
and $T'$ are arbitrary PQ-trees with leaves $L$ and $L'$ with an
injective map $\varphi : L' \rightarrow L$.  Note that the injective
map ensures that $L'$ can be treated as a subset of $L$.  In this
case, we call $T'$ a \emph{child} of $T$ and $T$ a \emph{parent} of
$T'$.  Choosing an order for the leaves $L$ of $T$ induces an order
for the leaves $L'$ of $T'$, whereas an order of $L'$ only partially
determines an order of $L$.  Now we are interested in all the orders
of the leaves $L$ that are represented by $T$ and additionally induce
an order for the leaves $L'$ that is represented by $T'$.  Informally
spoken, we want to find orders represented by $T'$ and $T$
simultaneously, fitting to one another.  It is clear that $T'$ can be
replaced by $T' \cap \proj T {L'}$ without changing the possible
orders, since each possible order of the leaves $L'$ is of course
represented by the projection $\proj T {L'}$ of $T$ to $L'$.  Hence
this general case reduces to the case where $T'$ is obtained from $T$
by applying a projection and a sequence of reductions.  We can extend
the notation of free and fixed nodes to this situation as follows.  An
edge incident to a P-node in the parent $T$ is free with respect to
the child $T'$ if and only if it is free with respect to the
projection $\proj T {L'}$.  If all edges are free, the whole P-node is
called free.  Similarly, a Q-node is free with respect to $T'$ if and
only if it is free with respect to $\proj T {L'}$.  Again, every fixed
Q-node $\mu$ has a representative $\rep(\mu)$ in $T'$ (which is also a
Q-node).  Figure~\ref{fig:PQ-trees-restrictions-fixed-edges} shows an
example PQ-tree together with a projection and a sequence of
reductions applied to it.

\subsection{SPQR-Trees}
\label{sec:spqr-trees-basics}

Consider a biconnected planar graph $G$ and a split pair $\{s,t\}$,
that is, $G - s - t$ is disconnected.  Let $H_1$ and $H_2$ be two
subgraphs of $G$ such that $H_1 \cup H_2 = G$ and $H_1 \cap H_2 = \{s,
t\}$.  Consider the following tree containing the two nodes $\mu_1$
and $\mu_2$ associated with the graphs $H_1 + \{s, t\}$ and $H_2 +
\{s, t\}$, respectively.  These graphs are called \emph{skeletons} of
the nodes $\mu_i$, denoted by $\skel(\mu_i)$ and the special edge
$\{s, t\}$ is said to be a \emph{virtual edge}.  The two nodes $\mu_1$
and $\mu_2$ are connected by an edge, or more precisely, the
occurrence of the virtual edges $\{s, t\}$ in both skeletons are
linked by this edge.  Now a combinatorial embedding of $G$ uniquely
induces a combinatorial embedding of $\skel(\mu_1)$ and
$\skel(\mu_2)$.  Furthermore, arbitrary and independently chosen
embeddings for the two skeletons determine an embedding of $G$, thus
the resulting tree can be used to represent all embeddings of $G$ by
the combination of all embeddings of two smaller planar graphs.  This
replacement can of course be applied iteratively to the skeletons
yielding a tree with more nodes but smaller skeletons associated with
the nodes.  Applying this kind of decomposition in a systematic way
yields the SPQR-tree as introduced by Di Battista and
Tamassia~\cite{On-LineMaintenanceof-DiBattista.Tamassia(96),
  On-LinePlanarityTesting-DiBattista.Tamassia(96)}.  The SPQR-tree
$\mathcal T$ of a biconnected planar graph $G$ contains four types of
nodes.  First, the P-nodes having a bundle of at lest three parallel
edges as skeleton and a combinatorial embedding is given by any order
of these edges.  Second, the skeleton of an R-node is triconnected
having exactly two embeddings, and third, S-nodes have a simple cycle
as skeleton without any choice for the embedding.  Finally, every edge
in a skeleton representing only a single edge in the original graph
$G$ is formally also considered to be a virtual edge linked to a
Q-node in $\mathcal T$ representing this single edge.  Note that all
leaves of the SPQR-tree $\mathcal T$ are Q-nodes.  Besides from being
a nice way to represent all embeddings of a biconnected planar graph,
the SPQR-tree has only linear size and Gutwenger and Mutzel showed how
to compute it in linear
time~\cite{LinearTimeImplementation-Gutwenger.Mutzel(01)}.
Figure~\ref{fig:spqr-tree} shows a biconnected planar graph together
with its SPQR-tree.

\begin{figure}[tb]
  \centering %
  \subcaptionbox{ \label{fig:spqr-tree}}
  {\includegraphics[page=1]{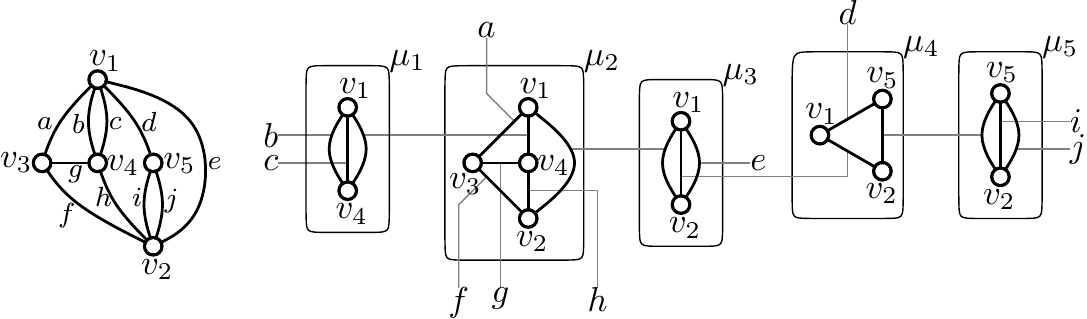}}\vspace{1em}
  \subcaptionbox{ \label{fig:PQ-trees-with-common-nodes-from-SPQR}}
  {\includegraphics[page=1]{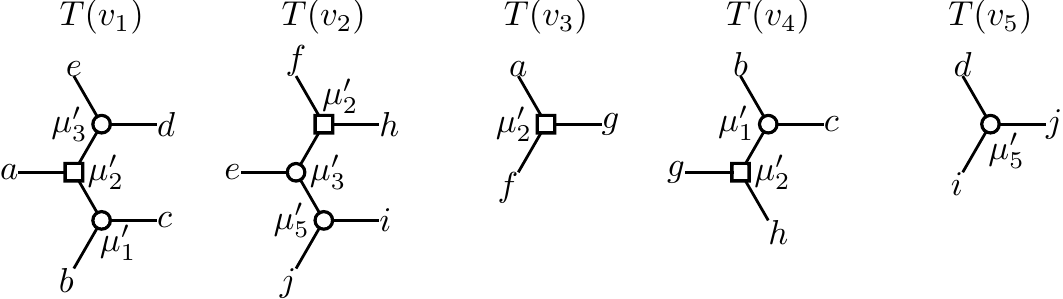}}
  \caption{(\subref{fig:spqr-tree}) A biconnected planar graph on the
    left and its SPQR-tree on the right.  The Q-nodes are depicted as
    single letters, whereas $\mu_1$, $\mu_3$ and $\mu_5$ are P-nodes,
    $\mu_2$ is an R-node and $\mu_4$ is an S-node.  The embeddings
    chosen for the skeletons yield the embedding shown for the graph
    on the left.
    (\subref{fig:PQ-trees-with-common-nodes-from-SPQR})~The embedding
    trees of the five vertices, where the inner nodes are named
    according to the nodes in the SPQR-tree they stem from.}
  \label{fig:spqr-tree-pq-trees}
\end{figure}

\subsection{Relation between PQ- and SPQR-Trees}
\label{sec:spqr-trees}

Given the SPQR-tree of a biconnected graph, it is easy to see that the
set of all possible orders of edges around a vertex is
PQ-representable.  For a vertex $v$ and a P-node in the SPQR-tree
containing $v$ in its skeleton, every virtual edge represents a set of
edges incident to $v$ that need to appear consecutively around $v$;
the order of the sets can be chosen arbitrarily.  For an R-node in the
SPQR-tree containing $v$, again every virtual edge represents a set of
edges that needs to appear consecutively, additionally the order of
the virtual edges is fixed up to reversal in this case.  Hence, there
is a bijection between the P- and R-nodes of the SPQR-tree containing
$v$ and the P- and Q-nodes of the PQ-tree representing the possible
orders of edges around $v$, respectively.  Note that the occurrence of
$v$ in the skeleton of an S-node enforces the edges belonging to one
of the two virtual edges incident to $v$ to appear consecutively
around $v$.  But since this would introduce a degree-2 node yielding
no new constraints, we can ignore the S-nodes.  We call the resulting
PQ-tree representing the possible circular orders of edges around a
vertex $v$ the \emph{embedding tree} of $v$ and denote it by $T(v)$.
Figure~\ref{fig:spqr-tree-pq-trees} depicts a planar graph together
with its SPQR-tree and the resulting embedding trees.

For every planar embedding of $G$, the circular order of edges around
every vertex $v$ is represented by the embedding tree $T(v)$, and for
every order represented by $T(v)$ we can find a planar embedding
realizing this order.  However, we cannot choose orders for the
embedding trees independently.  Consider for example the case that the
order of edges around $v_1$ in
Figure~\ref{fig:PQ-trees-with-common-nodes-from-SPQR} is already
chosen.  Since the embedding tree $T(v_1)$ contains nodes stemming
from the P-nodes $\mu_1$ and $\mu_3$ and the Q-node $\mu_2$ in the
SPQR-tree, the embedding of the skeletons in these nodes is already
fixed.  Since every other embedding tree except for $T(v_5)$ contains
nodes stemming from one of these three nodes the order of the incident
edges around $v_2$, $v_3$ and $v_4$ is at least partially determined.
In general, every P-node $\mu$ contains two vertices $v_1$ and $v_2$
in its skeleton, thus there are two embedding trees $T(v_1)$ and
$T(v_2)$ containing the P-nodes $\mu_1$ and $\mu_2$ stemming from
$\mu$.  The order of virtual edges in $\skel(\mu)$ around $v_1$ is the
opposite of the order of virtual edges around $v_2$ for a fixed
embedding of $\skel(\mu)$.  Hence, in every planar embedding of $G$
the edges around $\mu_1$ in $T(v_1)$ are ordered oppositely to the
order of edges around $\mu_2$ in $T(v_2)$.  Similarly, all Q-nodes in
the embedding trees stemming from the same R-node in the SPQR-tree
need to be oriented the same, if we choose the orders induced by one
of the two embeddings of the skeleton as reference orders of the
Q-nodes.  On the other hand, if every two P-nodes stemming from the
same P-node are ordered oppositely and all Q-nodes stemming from the
same R-node are oriented the same, we can simply use these orders and
orientations to obtain embeddings for the skeleton of every node in
the SPQR-tree, yielding a planar embedding of $G$.  Hence, all planar
embeddings of $G$ can be expressed in terms of the PQ-trees $T(v_1),
\dots, T(v_n)$, if we respect the additional constraints between nodes
stemming from the same node in the SPQR-tree.

\section{Simultaneous PQ-Ordering}
\label{sec:simult-pq-order}

As we have seen, all planar embeddings of a biconnected planar graph
$G$ can be expressed in terms of PQ-trees $T(v_1), \dots, T(v_n)$,
called the embedding trees, describing the orders of incident edges
around every vertex, if we respect some additional constraints between
the nodes of the embedding trees stemming from the same node of the
SPQR-tree.  In this section, we show how to get completely rid of the
SPQR-tree by providing a way to express these additional constraints
also in terms of PQ-trees.

The problem {\sc Simultaneous PQ-Ordering} is defined as follows.  Let $D = (N,
A)$ be a DAG with nodes $N = \{T_1, \dots, T_k\}$, where $T_i$ is a PQ-tree
representing the set of orders $\mathcal L_i$ on its leaves~$L_i$.  Every arc $a
\in A$ consist of a source $T_i$, a target $T_j$ and an injective map $\varphi :
L_j \rightarrow L_i$, and it is denoted by $(T_i, T_j; \varphi)$.  {\sc
  Simultaneous PQ-Ordering} asks whether there are orders $O_1, \dots, O_k$ with
$O_i \in \mathcal L_i$ such that an arc $(T_i, T_j; \varphi) \in A$ implies that
$\varphi(O_j)$ is a suborder of $O_i$.  Normally, we want every arc to represent
a projection followed by a sequence of reductions, which is not ensured by this
definition.  Hence, we say that an instance $D = (N, A)$ of {\sc Simultaneous
  PQ-Ordering} is \emph{normalized}, if an arc $(T_i, T_j; \varphi) \in A$
implies that $\mathcal L_i$ contains an order $O_i$ extending $\varphi(O_j)$ for
every order $O_j \in \mathcal L_j$.  It is easy to see that every instance of
{\sc Simultaneous PQ-Ordering} can be normalized.  If there is an order $O_j \in
\mathcal L_j$ such that $\mathcal L_i$ does not contain an extension of
$\varphi(O_j)$, then $O_j$ cannot be contained in any solution.  Hence, we do
not loose solutions by applying the reductions, given by $T_i$, to $T_j$.
Applying these reductions for every arc in $A$ top-down yields an equivalent
normalized instance.  From now on, all instances of {\sc Simultaneous
  PQ-Ordering} we consider are assumed to be normalized.  In most cases it is
not important to consider the map $\varphi$ explicitly, hence we often simply
write $(T_i, T_j)$ instead of $(T_i, T_j; \varphi)$ and say that $O_i$ is an
extension of $O_j$ instead of $\varphi(O_j)$.

Note that we cannot measure the size of an instance $D$ of {\sc Simultaneous
  PQ-Ordering} by the number of vertices plus the number of arcs, as it is usual
for simple graphs, since the nodes and arcs in $D$ are not of constant size in
our setting.  The size of every node in $D$ consisting of a PQ-tree $T$ is
linear in the number of nodes in $T$ or even linear in the number of leaves by
Lemma~\ref{lem:tree-without-deg-2}.  For every arc $(T_i, T_j; \varphi) \in A$
we need to store the injective map $\varphi$ from the leaves of $T_j$ to the
leaves of~$T_i$.  Thus, the size of this arc is linear in the number of leaves
in $T_j$.  Finally, the size of $D$, denoted by $|D|$, can be measured by the
size of all nodes plus the sizes of all arcs.

To come back to the embedding trees introduced in
Section~\ref{sec:spqr-trees}, we can now create a PQ-tree consisting
of a single Q-node as a common child of all embedding trees containing
a Q-node stemming from the same R-node in the SPQR-tree.  With the
right injective maps this additional PQ-tree ensures that all these
Q-nodes are oriented the same.  Similarly, we can ensure that two
P-nodes stemming from the same P-node of the SPQR-tree are ordered the
same, but what we really want is that the two P-nodes are ordered
oppositely.  Therefore, we also need \emph{reversing arcs} not
ensuring that an order is enforced to be the extension of the order
provided by the child, but requiring that it is an extension of the
reversal of this order.  To improve readability we do not consider
reversing arcs for now.  We will come back to this in
Section~\ref{sec:simult-pq-order-reversed} showing what changes if we
allow reversing arcs.

Since {\sc Simultaneous PQ-Ordering} is ${\mathcal NP}$-hard, which
will be shown in Section~\ref{sec:np-hardn-simult}, we will not solve
it in general, but we will give a class of instances that we can solve
efficiently.  In Section~\ref{sec:crit-tripl-expans} we figure out the
main problems in general instances and provide an approach to solve
{\sc Simultaneous PQ-Ordering} for ``simple'' instances.  In
Section~\ref{sec:1-critical-inst} we make precise which instances we
can solve and show how to solve them.  In Section~\ref{sec:impl-deta}
we give a detailed analysis of the running time and in
Section~\ref{sec:simult-pq-order-reversed} we show that the results on
{\sc Simultaneous PQ-Ordering} can be extended to the case where we
allow reversing arcs, that is, arcs ensuring that the order of the
source is an extension of the reversed order of the target.

\subsection[$\mathcal{NP}$-Completeness of Simultaneous
PQ-Ordering]{$\boldsymbol{\mathcal{NP}}$-Completeness of Simultaneous
  PQ-Ordering}
\label{sec:np-hardn-simult}

Let $L = \{\ell_1, \dots, \ell_n\}$ be a set of elements and let
$\Delta = \{(\ell_1^1, \ell_2^1, \ell_3^1), \dots, (\ell_1^d,
\ell_2^d, \ell_3^d)\}$ be a set of triples such that each triple
$(\ell_1^i, \ell_2^i, \ell_3^i)$ specifies a circular order for these
three elements.  The problem {\sc Cyclic Ordering} is to decide
whether there is a circular order of all elements in $L$ respecting
the circular order specified for every triple in $\Delta$.  Galil and
Megiddo proved that {\sc Cyclic Ordering} is
$\mathcal{NP}$-complete~\cite{Cyclicorderingis-Galil.Megiddo(77)}.

\begin{theorem}
  \label{thm:sim-pq-ord-np-compl}
  {\sc Simultaneous PQ-Ordering} is $\mathcal{NP}$-complete.
\end{theorem}
\begin{proof}
  It is clear that {\sc Simultaneous PQ-Ordering} is in $\mathcal{NP}$
  since it can be tested in polynomial time, if the conditions
  provided by the arcs are satisfied by given circular orders.  We
  show $\mathcal{NP}$-hardness by reducing {\sc Cyclic Ordering} to
  {\sc Simultaneous PQ-Ordering}.  Let $(L, \Delta)$ be an instance of
  {\sc Cyclic Ordering}.  We define the corresponding instance $D(L,
  \Delta)$ of {\sc Simultaneous PQ-Ordering} as follows.  We create
  one PQ-tree $T$ consisting of a single P-node with leaves $L$.  For
  every triple $(\ell_1^i, \ell_2^i, \ell_3^i)$ we create a PQ-tree
  $T(\ell_1^i, \ell_2^i, \ell_3^i)$ consisting of a single node (it
  does not matter if P- or Q-node) with leaves $\{\ell_1^i, \ell_2^i,
  \ell_3^i\}$ with an incoming arc $(T, T(\ell_1^i, \ell_2^i,
  \ell_3^i); \id)$, where $\id$ is the identity map.  With this
  construction it is still possible to choose an arbitrary order for
  each of the triples.  To ensure that they are all ordered the same,
  we introduce an additional PQ-tree $T^\times$ consisting of a single
  node with three leaves $1$, $2$ and $3$ and an incoming arc
  $(T(\ell_1^i, \ell_2^i, \ell_3^i), T^\times; \varphi)$ with
  $\varphi(j) = \ell_j^i$ for every triple $(\ell_1^i, \ell_2^i,
  \ell_3^i)$.  Figure~\ref{fig:SPQO-np-complete} illustrates this
  construction.  It is clear that the size of $D(L, \Delta)$ is linear
  in the size of $(L, \Delta)$.  It remains to show that the instance
  $(L, \Delta)$ of {\sc Cyclic Ordering} and the instance $D(L,
  \Delta)$ of {\sc Simultaneous PQ-Ordering} are equivalent.

  Assume we have a solution of $(L, \Delta)$, that is, we have a
  circular order $O$ of $L$ such that every triple $(\ell_1^i,
  \ell_2^i, \ell_3^i) \in \Delta$ has the circular order $\ell_1^i
  \ell_2^i \ell_3^i$.  The PQ-tree $T$ in $D(L, \Delta)$ has the
  leaves $L$, thus we can choose $O$ as the order of the leaves of
  $T$.  For every triple $(\ell_1^i, \ell_2^i, \ell_3^i)$ there is an
  incoming arc from $T$ to $T(\ell_1^i, \ell_2^i, \ell_3^i)$ inducing
  the circular order $\ell_1^i \ell_2^i \ell_3^i$ on its leaves.
  Furthermore, there is an outgoing arc to $T^\times$ inducing the
  order $1 2 3$.  Since all of these arcs having $T^\times$ as target
  induce the same circular order $1 2 3$, these orders are a solution
  of the instance $D(L, \Delta)$ of {\sc Simultaneous PQ-Ordering}.

  Conversely, assume we have a solution for $D(L, \Delta)$.  If the
  order of leaves in $T^\times$ is $1 3 2$, we obtain another solution
  by reversing all orders.  Thus, we can assume without loss of
  generality that the leaves of $T^\times$ have the order $1 2 3$.
  Hence, the leaves of the tree $T(\ell_1^i, \ell_2^i, \ell_3^i)$ are
  ordered $\ell_1^i \ell_2^i \ell_3^i$ for every triple $(\ell_1^i,
  \ell_2^i, \ell_3^i)$ implying that the order on the leaves $L$ of
  $T$, which is an extension of all these orders, is a solution of the
  instance $(L, \Delta)$ of {\sc Cyclic Ordering}.
\end{proof}

\begin{figure}[tb]
  \centering
  \includegraphics{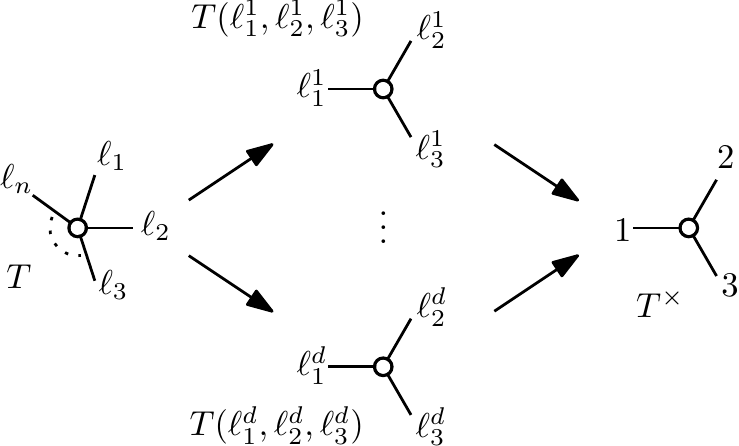}
  \caption{The instance $D(L, \Delta)$ of {\sc Simultaneous
      PQ-Ordering} corresponding to the instance $(L, \Delta)$ of {\sc
      Cyclic Ordering} with $L = \{\ell_1, \dots, \ell_n\}$ and
    $\Delta = \{(\ell_1^1, \ell_2^1, \ell_3^1), \dots, (\ell_1^d,
    \ell_2^d, \ell_3^d)\}$.}
  \label{fig:SPQO-np-complete}
\end{figure}

\subsection{Critical Triples and the Expansion Graph}
\label{sec:crit-tripl-expans}

Although {\sc Simultaneous PQ-Ordering} is $\mathcal {NP}$-complete in
general, we give in this section a strategy how to solve it for
special instances.  Afterwards, in Section~\ref{sec:1-critical-inst},
we show that this strategy really leads to a polynomial-time algorithm
for a certain class of instances.  Let $D = (N, A)$ be an instance of
{\sc Simultaneous PQ-Ordering} and let $(T, T_1) \in A$ be an arc.  By
choosing an order $O_1 \in \mathcal L_1$ and extending $O_1$ to an
order $O \in \mathcal L$, we ensure that the constraint given by the
arc $(T, T_1)$ is satisfied.  Hence, our strategy will be to choose
orders bottom-up, which can always be done for a single arc since our
instances are normalized.  Unfortunately, $T$ can have several
children $T_1, \dots, T_\ell$, and orders $O_i \in \mathcal L_i$
represented by $T_i$ for $i = 1, \dots, \ell$ cannot always be
simultaneously extended to an order $O \in \mathcal L$ represented by
$T$.  We derive necessary and sufficient conditions for the orders
$O_i$ to be simultaneously extendable to an order $O \in \mathcal L$
under the additional assumption that every P-node in $T$ is fixed with
respect to at most two children.  We consider the Q- and P-nodes in
$T$ separately.

Let $\mu$ be a Q-node in $T$.  If $\mu$ is fixed with respect to
$T_i$, there is a unique Q-node $\rep(\mu)$ in $T_i$ determining its
orientation.  By introducing a boolean variable $x_\eta$ for every
Q-node $\eta$, which is {\sc true} if $\eta$ is oriented the same as a
fixed reference orientation and {\sc false} otherwise, we can express
the condition that $\mu$ is oriented as determined by its
representative by $x_\mu = x_{\rep(\mu)}$ or $x_\mu \not=
x_{\rep(\mu)}$.  For every Q-node in $T$ that is fixed with respect to
a child $T_i$ we obtain such an (in)equality and we call the resulting
set of (in)equalities the \emph{Q-constraints}.  It is obvious that
the Q-constraints are necessary.  On the other hand, if the
Q-constraints are satisfied, all children of $T$ fixing the
orientation of $\mu$ fix it in the same way.  Note that the
Q-constraints form an instance of {\sc 2-Sat} that has linear size in
the number of Q-nodes, which can be solved in
polynomial~\cite{DecisionProblemClass-Krom(67)} and even
linear~\cite{ComplexityofTimetable-Even.etal(76),
  linear-timealgorithmtesting-Aspvall.etal(79)} time.  Hence, we only
need to deal with the P-nodes, which is not as simple.

Let $\mu$ be a P-node in $T$.  If $\mu$ is fixed with respect to only
one child $T_i$, we can simply choose the order given by $O_i$.  If
$\mu$ is additionally fixed with respect to $T_j$, it is of course
necessary that the orders $O_i$ and $O_j$ induce the same order for
the edges incident to $\mu$ that are fixed with respect to both, $T_i$
and $T_j$.  We call such a triple $(\mu, T_i, T_j)$, where $\mu$ is a
P-node in $T$ fixed with respect to the children $T_i$ and $T_j$ a
\emph{critical triple}.  We say that the critical triple $(\mu, T_i,
T_j)$ is \emph{satisfied} if the orders $O_i$ and $O_j$ induce the
same order for the edges incident to $\mu$ commonly fixed with respect
to $T_i$ and $T_j$.  If we allow multiple arcs, we can also have a
critical triple $(\mu, T', T')$ for two parallel arcs $(T, T';
\varphi_1)$ and $(T, T'; \varphi_2)$.  Clearly, all critical triples
need to be satisfied by the orders chosen for the children to be able
to extend them simultaneously.  Note that this condition is not
sufficient, if $\mu$ is contained in more than one critical triple,
which is one of the main difficulties of {\sc Simultaneous
  PQ-Ordering} for general instances.  However, the following lemma
shows that satisfying all critical triples is not only necessary but
also sufficient, if every P-node is contained in at most one critical
triple, that is, it is fixed with respect to at most two children of
$T$.  See Figure~\ref{fig:SPQO-extension-of-2-or-3-orders} for two
simple examples, illustrating that satisfying critical triples is
sufficient if every P-node is contained in at most one critical
triple, whereas the general case is not as simple.

\begin{figure}[tb]
  \centering
  \subcaptionbox{ \label{fig:SPQO-extension-of-2-or-3-orders-a}}
  {\includegraphics[page=1]{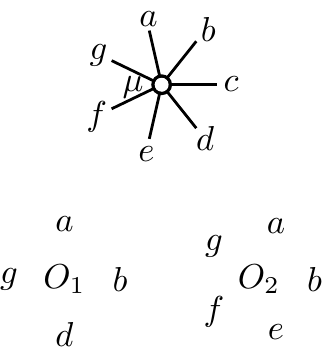}}\hspace{3em}
  \subcaptionbox{ \label{fig:SPQO-extension-of-2-or-3-orders-b}}
  {\includegraphics[page=2]{fig/SPQO-extension-of-2-or-3-orders}}
  \caption{(\subref{fig:SPQO-extension-of-2-or-3-orders-a})~We can
    find an order for the P-node $\mu$ extending the orders $O_1$ and
    $O_2$ if and only if the
    commonly fixed edges $a$, $b$ and $g$ are ordered the same. \\
    (\subref{fig:SPQO-extension-of-2-or-3-orders-b})~Although for
    every pair $\{O_i, O_j\}$ of orders out of the three orders $O_1$,
    $O_2$ and $O_3$ the commonly fixed edges are ordered the same, we
    cannot extend all three orders simultaneously.}
  \label{fig:SPQO-extension-of-2-or-3-orders}
\end{figure}

\newcommand{\lemSimultExtensionOfOrdersText}{ Let $T$ be a PQ-tree
  with children $T_1, \dots, T_\ell$, such that every P-node in $T$ is
  contained in at most one critical triple, and let $O_1, \dots,
  O_\ell$ be orders represented by $T_1, \dots, T_\ell$.  An order $O$
  that is represented by $T$ and simultaneously extends the orders
  $O_1, \dots, O_\ell$ exists if and only if the Q-constraints and all
  critical triples are satisfied.}
\begin{lemma}
  \label{lem:simult-extension-of-orders}
  \lemSimultExtensionOfOrdersText
\end{lemma}
\begin{proof}
  The only if part is clear, since an order $O$ represented by $T$
  extending the orders $O_1, \dots, O_\ell$ yields an assignment of
  {\sc true} and {\sc false} to the variables $x_\eta$ satisfying the
  Q-constraints.  Additionally, for every critical triple $(\mu, T_i,
  T_j)$ the common fixed edges are ordered the same in $O$ as in $O_i$
  and in $O_j$ and hence $(\mu, T_i, T_j)$ is satisfied.

  Now, assume that we have orders $O_1, \dots, O_\ell$ satisfying the
  Q-constraints and every critical triple.  We show how to construct
  an order $O$ represented by $T$, extending all orders $O_1, \dots,
  O_\ell$ simultaneously.  The variable assignments for the variables
  stemming from Q-nodes in each of the children $T_1, \dots, T_\ell$
  imply an assignment of every variable stemming from a fixed Q-node
  in $T$, and hence an orientation of this Q-node.  Since the
  Q-constraints are satisfied, all children fixing a Q-node in $T$
  imply the same orientation.  The orientation of free Q-nodes can be
  chosen arbitrarily.  For a P-node $\mu$ in $T$ that is fixed with
  respect to at most one child of $T$, we can simply choose the order
  of fixed edges incident to $\mu$ as determined by the child and add
  the free edges arbitrarily.  Otherwise, $\mu$ is contained in
  exactly one critical triple $(\mu, T_i, T_j)$.  We first choose the
  order of edges incident to $\mu$ that are fixed with respect to
  $T_i$ as determined by $O_i$.  From the point of view of $T_j$, some
  of the fixed edges incident to $\mu$ are already ordered, but this
  order is consistent with the order induced by $O_j$, since $(\mu,
  T_i, T_j)$ is satisfied.  Additionally, some edges that are free
  with respect to $T_j$ are already ordered.  Of course, the remaining
  edges incident to $\mu$ that are fixed with respect to $T_j$ can be
  added as determined by $O_j$, and the free edges can be added
  arbitrarily.
\end{proof}

Since testing whether the Q-constraints are satisfiable is easy, we
concentrate on satisfying the critical triples.  Let $\mu$ be a P-node
in a PQ-tree $T$ such that $\mu$ is fixed with respect to two children
$T_1$ and $T_2$, that is, $(\mu, T_1, T_2)$ is a critical triple.  By
projecting $T_1$ and $T_2$ to representatives of the common fixed
edges incident to $\mu$ and intersecting the result, we obtain a new
PQ-tree $T(\mu, T_1, T_2)$.  There are natural injective maps from the
leaves of $T(\mu, T_1, T_2)$ to the leaves of $T_1$ and $T_2$, hence
we can add $T(\mu, T_1, T_2)$ together with incoming arcs from $T_1$
and $T_2$ to our instance~$D$ of {\sc Simultaneous PQ-Ordering}.  This
procedure of creating $T(\mu, T_1, T_2)$ is called \emph{expansion
  step} with respect to the critical triple $(\mu, T_1, T_2)$, and the
resulting new PQ-tree $T(\mu, T_1, T_2)$ is called the \emph{expansion
  tree} with respect to that triple; see
Figure~\ref{fig:SPQR-expansion-step} for an example of the expansion
step.  We say that the P-node $\mu$ in $T$ is \emph{responsible} for
the expansion tree $T(\mu, T_1, T_2)$.  Note that every expansion tree
has two incoming and no outgoing arcs at the time it is created.

\begin{figure}[tb]
  \centering
  \includegraphics{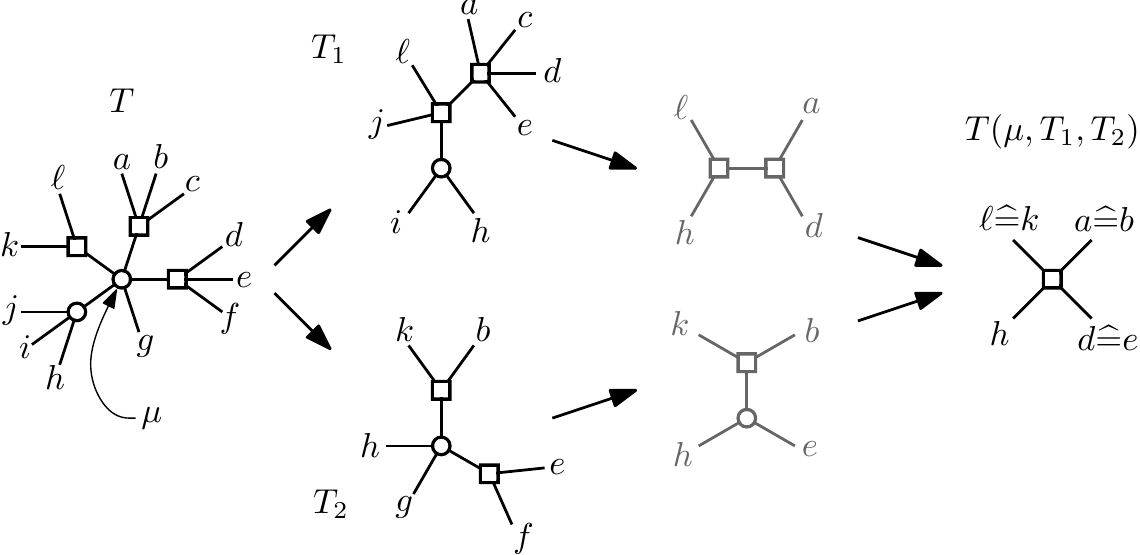}
  \caption{The P-node $\mu$ in the PQ-tree $T$ is fixed with respect
    to the children $T_1$ and $T_2$.  We first project $T_1$ and $T_2$
    to representatives of the common fixed edges incident to $\mu$ and
    intersect the result to obtain $T(\mu, T_1, T_2)$.  Note that the
    gray shaded projections only illustrate an intermediate step, they
    are not inserted.}
  \label{fig:SPQR-expansion-step}
\end{figure}

We introduce the expansion tree for the following reason.  If we find
orders $O_1$ and $O_2$ represented by $T_1$ and $T_2$ that both extend
the same order represented by the expansion tree $T(\mu, T_1, T_2)$,
we ensure that the edges incident to $\mu$ fixed with respect to both,
$T_1$ and $T_2$, are ordered the same in $O_1$ and $O_2$, or in other
words, we ensure that $O_1$ and $O_2$ satisfy the critical triple
$(\mu, T_1, T_2)$.  By Lemma~\ref{lem:simult-extension-of-orders}, we
know that satisfying the critical triple is necessary, thus we do not
loose solutions by adding expansion trees to an instance of {\sc
  Simultaneous PQ-Ordering}.  Furthermore, it is also sufficient, if
every P-node is contained in at most one critical triple (if we forget
about the Q-nodes for a moment).  Hence, given an instance $D$ of {\sc
  Simultaneous PQ-Ordering}, we would like to expand $D$ iteratively
until no unprocessed critical triples are left and find simultaneous
orders bottom-up.  Unfortunately, it can happen that the expansion
does not terminate and thus yields an infinite graph; see
Figure~\ref{fig:SPQO-infinite-expansion} for an example.  Thus, we
need to define a special case where we do not expand further.  Let
$\mu$ be a P-node of $T$ with outgoing arcs $(T, T_1; \varphi_1)$ and
$(T, T_2; \varphi_2)$ such that $(\mu, T_1, T_2)$ is a critical
triple.  Denote the leaves of $T_1$ and $T_2$ by $L_1$ and $L_2$,
respectively.  If $T_i$ (for $i = 1, 2$) consists only of a single
P-node, the image of $\varphi_i$ is a set of representatives of the
edges incident to $\mu$ that are fixed with respect to $T_i$.  Hence
$\varphi_i$ is a bijection between $L_i$ and the fixed edges incident
to $\mu$.  If additionally the fixed edges with respect to both, $T_1$
and $T_2$, are the same, we obtain a bijection $\varphi : L_1
\rightarrow L_2$.  Assume without loss of generality that there is no
directed path from $T_2$ to $T_1$ in the current DAG.  If there is
neither a directed path from $T_1$ to $T_2$ nor form $T_2$ to $T_1$,
we achieve uniqueness by assuming that $T_1$ comes before $T_2$ with
respect to some fixed order of the nodes in~$D$.  Instead of an
expansion step we apply a \emph{finalizing step} by simply creating
the arc $(T_1, T_2; \varphi)$.  This new arc ensures that the critical
triple $(\mu, T_1, T_2)$ is satisfied if we have orders for the leaves
$L_1$ and $L_2$ respecting $(T_1, T_2; \varphi)$.  Since no new node
is inserted, we do not run into the situation where we create the same
PQ-tree over and over again.

For the case that $(\mu, T', T')$ is a critical triple resulting from
two parallel arcs $(T, T'; \varphi_1)$ and $(T, T'; \varphi_2)$, we
can apply the expansion step as described above.  If the conditions
for a finalizing step are given, that is $T'$ consists of a single
P-node and both maps $\varphi_1$ and $\varphi_2$ fix the same edges
incident to $\mu$, a finalizing step would introduce a self loop with
the permutation $\varphi$ associated with it.  In this case, we omit
the loop and mark $(T, T'; \varphi_1)$ and $(T, T'; \varphi_2)$ as a
\emph{critical double arc} with the associated permutation $\varphi$.
When choosing orders bottom-up in the DAG, we have to explicitly
ensure that all critical triples stemming from critical double arcs
are satisfied.  To simplify this, we ensure that all targets of
critical double arcs are sinks in the expansion graph.  This follows
from the construction, except for the case when the critical double
arc is already contained in the input instance.  In this case, we
apply one additional expansion step, which essentially clones the
double arc.  We thus distinguish between the two cases that $T'$ is an
expansion tree and that it was already contained in $D$.  If it is an
expansion tree, we do nothing and mark the critical triple as
processed.  Otherwise, we apply an expansion step having the effect
that the resulting expansion tree again satisfies the conditions to
apply a finalizing step and additionally is an expansion tree.  Since
we want to apply Lemma~\ref{lem:simult-extension-of-orders} by
choosing orders bottom-up, it is a problem that the critical triples
belonging to critical double arcs are not satisfied automatically.
However, if every P-node is contained in at most one critical triple,
our construction ensures that the target $T'$ of a critical double arc
is a sink and no further expansion or finalizing steps can change
that.  Hence, we are free to choose any order for the leaves of $T'$
and we will use Lemma~\ref{lem:permutation-order-preserving} (about
order preserving permutations) to choose it in a way satisfying the
critical triple or decide that this is impossible.

\begin{figure}[tb]
  \centering
  \includegraphics{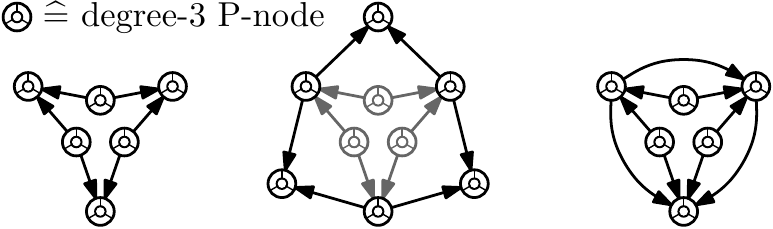}
  \caption{Consider the instance of {\sc Simultaneous PQ-Ordering} on
    the left, where every PQ-tree consists of a single P-node with
    degree~3.  The DAG in the center shows the result after expanding
    three times.  The so far processed part is shaded gray and for the
    remaining part we are in the same situation as before, hence
    iterated expansion would yield an infinite DAG.  To prevent
    infinite expansion we apply finalizing steps resulting in the DAG
    on the right.}
  \label{fig:SPQO-infinite-expansion}
\end{figure}

To sum up, we start with an instance $D$ of {\sc Simultaneous
  PQ-Ordering}.  As long as $D$ contains unprocessed critical triples
$(\mu, T_1, T_2)$ we apply expansion steps (or finalizing steps if
$T_1$ and $T_2$ are essentially the same) and mark $(\mu, T_1, T_2)$
as processed.  The resulting graph is called the \emph{expansion
  graph} of $D$ and is denoted by $D_{\ex}$.  Note that $D_{\ex}$ is
also an instance of {\sc Simultaneous PQ-Ordering}.  Before showing in
Lemma~\ref{lem:exp-graph-equiv} that $D$ and $D_{\ex}$ are equivalent,
we need to show that $D_{\ex}$ is well defined, that is, it is unique
and finite.  Lemma~\ref{lem:chain-of-p-nodes-terminates} essentially
states that the P-nodes become smaller at least every second expansion
step.  We will use this result in Lemma~\ref{lem:exp-graph-well-def}
to show finiteness.

\begin{lemma}
  \label{lem:chain-of-p-nodes-terminates}
  Let $D$ be an instance of {\sc Simultaneous PQ-Ordering} and let
  $D_{\ex}$ be its expansion graph.  Let further $T$ be a PQ-tree in
  $D_{\ex}$ containing a P-node $\mu$.  If $\mu$ is responsible for an
  expansion tree $T'$ containing a P-node $\mu'$ with $\deg(\mu') =
  \deg(\mu)$, then $\mu'$ itself is not responsible for an expansion
  tree $T''$ containing a P-node $\mu''$ with $\deg(\mu'') =
  \deg(\mu') = \deg(\mu)$.
\end{lemma}
\begin{proof}
  Since $T'$ is created by first projecting a child of $T$ to
  representatives of edges incident to $\mu$, it can contain at most
  $\deg(\mu)$ leaves.  Thus, if $T'$ contains a P-node $\mu'$ with
  $\deg(\mu') = \deg(\mu)$, it contains no other inner node.  Now
  assume that $\mu'$ is responsible for another expansion tree $T''$
  containing a P-node $\mu''$ with $\deg(\mu'') = \deg(\mu') =
  \deg(\mu)$ and let $(\mu', T_1, T_2)$ be the corresponding critical
  triple.  Again $T''$ consists only of the single P-node $\mu''$.
  Since $T_1$ and $T_2$ lie on a directed path from $T'$ to $T''$ they
  also need to consist of single P-nodes with $\deg(\mu')$ incident
  edges.  Thus, $T_1$ and $T_2$ consist both of a single P-node having
  the same degree and they fix the same, namely all, edges incident to
  $\mu'$.  Hence we would have applied a finalizing step instead of
  creating the expansion tree $T''$; a contradiction.
\end{proof}

\begin{lemma}
  \label{lem:exp-graph-well-def}
  The expansion graph $D_{\ex}$ of an instance $D = (N, A)$ of {\sc
    Simultaneous PQ-Ordering} is unique and finite.
\end{lemma}
\begin{proof}
  If we apply an expansion or a finalizing step due to a critical
  triple $(\mu, T_1, T_2)$, where $\mu$ is a P-node of the PQ-tree
  $T$, the result does only depend on the trees $T$, $T_1$ and $T_2$
  and the arcs $(T, T_1)$ and $(T, T_2)$.  By applying other expansion
  or finalizing steps, we of course do not change these trees or arcs,
  thus it does not matter in which order we expand and finalize a
  given DAG~$D$.  Hence, $D_{\ex}$ is unique and we can talk about
  \emph{the} expansion graph $D_{\ex}$ of an instance $D$ of {\sc
    Simultaneous PQ-Ordering}.

  To prove that $D_{\ex}$ is finite, we show that $\level(D_{\ex}) \le
  \level(D) + 4 \cdot (p_{\max} + 1)$, where $p_{\max}$ is the degree
  of the largest P-node in $D$.  To simplify the notation denote
  $p_{\max} + 1$ by $p_{\max}^+$.  Recall that the level of a node in
  $D$ was defined as the the shortest directed path from a sink to
  this node and $\level(D)$ is the largest level occurring in $D$.
  Note that all sources in $D_{\ex}$ are already contained in $D$,
  since every expansion tree has two incoming arcs.  Showing that the
  level of $D_{\ex}$ is finite is sufficient since there are only
  finitely many sources in $D_{\ex}$ and no node has infinite degree.
  Assume we have a PQ-tree $T_1$ in $D_{\ex}$ with $\level(T_1) >
  \level(D) + 4 \cdot p_{\max}^+$. Then $T_1$ is of course an
  expansion tree and there is a unique P-node $\mu_2$ that is
  responsible for $T_1$.  Denote the PQ-tree containing $\mu_2$ by
  $T_2$.  Since there is a directed path of length~2 from $T_2$ to
  $T_1$, we have $\level(T_2) \ge \level(T_1) - 2 > \level(D) + 4
  \cdot p_{\max}^+ - 2$.  Due to its level, $T_2$ itself needs to be
  an expansion tree and we can continue, obtaining a sequence $T_1,
  \dots, T_{2 \cdot p_{\max}^+}$ of expansion trees containing P-nodes
  $\mu_i$, such that $\mu_i$ is responsible for $T_{i-1}$.  Due to
  Lemma~\ref{lem:chain-of-p-nodes-terminates} the degree of $\mu_i$ is
  larger than the degree of $\mu_{i-2}$, hence $\deg(\mu_{2 \cdot
    p_{\max}^+}) \ge p_{\max}^+ > p_{\max}$, which is a contradiction
  to the assumption that the largest P-node in $D$ has degree
  $p_{\max}$.
\end{proof}

Now that we know that the expansion graph $D_{\ex}$ of a given
instance $D$ of {\sc Simultaneous PQ-Ordering} is well defined, we can
show what we already mentioned above, namely that $D$ and $D_{\ex}$
are equivalent.

\begin{lemma}
  \label{lem:exp-graph-equiv}
  An instance $D$ of {\sc Simultaneous PQ-Ordering} admits
  simultaneous PQ-orders if and only if its expansion graph $D_{\ex}$
  does.
\end{lemma}
\begin{proof}
  It is clear that $D$ is a subgraph of $D_{\ex}$.  Hence, if we have
  simultaneous orders for the expansion graph $D_{\ex}$, we of course
  also have simultaneous orders for the original instance $D$.

  It remains to show that we do not loose solutions by applying
  expansion or finalizing steps.  Assume we have simultaneous orders
  for the original instance $D$.  Since every expansion tree is a
  descendant of a PQ-tree in $D$, for which the order is already
  fixed, there is no choice left for the expansion trees.  Thus, we
  only need to show that for every expansion tree all parents induce
  the same order on its leaves and that this order is represented by
  the expansion tree.  We first show this for the expansion graph
  without the arcs inserted due to finalizing steps.  Afterwards, we
  show that adding these arcs preserves valid solutions.

  Consider an expansion tree $T(\mu, T_1, T_2)$ introduced due to the
  critical triple $(\mu, T_1, T_2)$ such that $T_1$, $T_2$ and the
  tree $T$ containing $\mu$ are not expansion trees.  By construction
  $T(\mu, T_1, T_2)$ represents the edges incident to $\mu$ fixed with
  respect to $T_1$ and $T_2$.  Since the orders chosen for $T_1$,
  $T_2$ and $T$ are valid simultaneous orders, $T_1$ and $T_2$ induce
  the same order for the leaves of $T(\mu, T_1, T_2)$.  Since $T(\mu,
  T_1, T_2)$ has no other incoming arcs, we do not need to consider
  other parents.  The induced order is of course represented by the
  projection of $T_1$ and $T_2$ to the commonly fixed edges incident
  to $\mu$, and hence it is of course also represented by their
  intersection $T(\mu, T_1, T_2)$.  For the case that $T$, $T_1$ or
  $T_2$ are expansion trees, we can assume by induction that the
  orders chosen for $T$, $T_1$ and $T_2$ are valid simultaneous
  orders, yielding the same result that $T_1$ and $T_2$ induce the
  same order represented by $T(\mu, T_1, T_2)$.  It remains to show,
  that the arcs introduced by a finalizing step respect the chosen
  orders.  Let $T(\mu, T_1, T_2)$ a critical triple such that $T_1$
  and $T_2$ consist of single P-nodes both fixing the same edges in
  $\mu$.  It is clear that the order chosen for $\mu$ induces the same
  order for $T_1$ and $T_2$ with respect to the canonical bijection
  $\varphi$ between the leaves of $T_1$ and $T_2$.  Hence, adding an
  arc $(T_1, T_2; \varphi)$ preserves simultaneous PQ-orders.
\end{proof}

For now, we know that we can consider the expansion graph instead of
the original instance to solve {\sc Simultaneous PQ-Ordering}.
Lemma~\ref{lem:simult-extension-of-orders} motivates that we can solve
the instance given by the expansion graph by simply choosing orders
bottom-up, if additionally the Q-constraints are satisfiable.
However, this only works for ``simple'' instances since satisfying
critical triples is no longer sufficient for a P-node that is fixed
with respect to more than two children.  And there is another problem,
namely that the expansion graph can become exponentially large.  In
the following section we will define precisely what ``simple'' means
and additionally address the second problem by showing that the
expansion graph has polynomial size for these instances.

\subsection{1-Critical and 2-Fixed Instances}
\label{sec:1-critical-inst}

The expansion graph was introduced to satisfy the critical triples
simply by choosing orders bottom-up, which can then be used to apply
Lemma~\ref{lem:simult-extension-of-orders}, if the additional
condition that every P-node is contained in at most one critical
triple is satisfied.  Let $D$ be an instance of {\sc Simultaneous
  PQ-Ordering} and let $D_{\ex}$ be its expansion graph.  We say that
$D$ is a \emph{1-critical instance}, if in its expansion graph
$D_{\ex}$ every P-node is contained in at most one critical triple.
We will first prove a lemma helping us, to deal with critical double
arcs.  Afterwards, we show how to solve 1-critical instances
efficiently.

\begin{lemma}
  \label{lem:unsatisfied-double-edges-are-sinks}
  Let $D$ be a 1-critical instance of {\sc Simultaneous PQ-Ordering}
  with expansion graph $D_{\ex}$.  Let further $(T, T'; \varphi_1)$
  and $(T, T'; \varphi_2)$ be a critical double arc.  Then $T'$ is a
  sink in $D_{\ex}$.
\end{lemma}
\begin{proof}
  Since $T'$ consists only of a single P-node, there is exactly one
  P-node $\mu$ in $T$ that is fixed with respect to $T'$.  Due to the
  double arc, $\mu$ is contained in the critical triple $(\mu, T',
  T')$.  The tree $T'$ is an expansion tree by construction, hence at
  the time $T'$ is created it has only the two incoming arcs $(T, T';
  \varphi_1)$ and $(T, T'; \varphi_2)$ and no outgoing arc.  Assume
  that we can introduce an outgoing arc to $T'$ by applying an
  expansion or finalizing step.  Then $T'$ needs to be contained in
  another critical triple than $(\mu, T', T')$ and since $T$ is its
  only parent and $\mu$ is the only P-node in $T$ fixed with respect
  to $T'$, this critical triple must also contain $\mu$.  But then
  $\mu$ is contained in more than one critical triple, which is a
  contradiction to the assumption that $D$ is 1-critical.
\end{proof}

\begin{lemma}
  \label{lem:solve-1-critical-instances}
  Let $D$ be a 1-critical instance of {\sc Simultaneous PQ-Ordering}
  with expansion graph $D_{\ex}$.  In time polynomial in $|D_{\ex}|$
  we can compute simultaneous PQ-orders or decide that no such orders
  exist.
\end{lemma}
\begin{proof}
  Due to Lemma~\ref{lem:exp-graph-equiv}, we can solve the instance
  $D_{\ex}$ of {\sc Simultaneous PQ-ordering} instead of $D$ itself.
  Of course we cannot find simultaneous PQ-orders for the PQ-trees in
  $D_{\ex}$ if any of these PQ-trees is the null tree.  Additionally,
  Lemma~\ref{lem:simult-extension-of-orders} states that the
  Q-constraints are necessary.  We can check in linear time whether
  there exists an assignment of {\sc true} and {\sc false} to the
  variables $x_\mu$, where $\mu$ is a Q-node, satisfying the
  Q-constraints by solving a linear size instance of {\sc
    2-Sat}~\cite{ComplexityofTimetable-Even.etal(76),
    linear-timealgorithmtesting-Aspvall.etal(79)}.  Hence, if
  $D_{\ex}$ contains the null tree or the Q-constraints are not
  satisfiable, we know that there are no simultaneous PQ-orders.
  Additionally, we need to deal with the critical double arcs.  Let
  $(T, T'; \varphi_1)$ together with $(T, T'; \varphi_2)$ be a
  critical double arc.  By construction, the target $T'$ consists of a
  single P-node fixing the same edges incident to a single P-node
  $\mu$ in $T$ with respect to both edges.  Thus, $\varphi_1$ and
  $\varphi_2$ can be seen as bijections between the leaves $L'$ of
  $T'$ and the fixed edges incident to $\mu$ and hence they define a
  permutation $\varphi$ on $L'$ with $\varphi = \varphi_2^{-1} \circ
  \varphi_1$.  To satisfy the critical triple $(\mu, T', T')$, we need
  to find an order $O'$ of $L'$ such that $\varphi_1(O') =
  \varphi_2(O')$.  This equation is equivalent to $\varphi_1 \circ
  \varphi (O') = \varphi_2 \circ \varphi(O')$, and hence also to
  $\varphi(O') = O'$.  Thus, the critical triple $(\mu, T', T')$ is
  satisfied if and only if $\varphi$ is order preserving with respect
  to $O'$.  Whether $\varphi$ is order preserving with respect to any
  order can be tested in $\mathcal O(|L'|)$ time by applying
  Lemma~\ref{lem:permutation-order-preserving}.  Now assume we have a
  variable assignment satisfying the Q-constraints, no PQ-tree is the
  null tree and every permutation $\varphi$ corresponding to a
  critical double arc is order preserving.  We show how to find
  simultaneous PQ-orders for all PQ-trees in $D_{\ex}$.

  Start with a sink $T$ in $D_{\ex}$.  If $T$ is the target of a
  critical double arc, it is a single P-node and its corresponding
  permutation $\varphi$ is order preserving by assumption and hence we
  can use Lemma~\ref{lem:permutation-order-preserving} to choose an
  order that is preserved by $\varphi$.  Otherwise, orientate every
  Q-node $\mu$ in $T$ as determined by the variable $x_\mu$ stemming
  from it.  Additionally, choose an arbitrary order for every P-node
  in~$T$.  Afterwards mark $T$ as processed.  We continue with a
  PQ-tree $T$ in $D_{\ex}$ for which all of its children $T_1, \dots,
  T_\ell$ are already processed, that is, we traverse $D_{\ex}$
  bottom-up.  Since $T_1, \dots, T_\ell$ are processed, orders $O_1,
  \dots, O_\ell$ for their leaves were already chosen.  Consider a
  P-node $\mu$ in $T$ contained in a critical triple $(\mu, T_i,
  T_j)$.  If there is the expansion tree $T(\mu, T_i, T_j)$, it
  guarantees that the edges incident to $\mu$ fixed with respect to
  $T_i$ and $T_j$ are ordered the same in $O_i$ and $O_j$ and hence
  the critical triple is satisfied.  If we had to apply a finalizing
  step due to the critical triple $(\mu, T_i, T_j)$, we have an arc
  from $T_i$ to $T_j$ (or in the other direction), again ensuring that
  $O_i$ and $O_j$ induce the same order on the fixed edges incident to
  $\mu$.  In the special case that $(\mu, T_i, T_j)$ corresponds to a
  critical double arc, we know due to
  Lemma~\ref{lem:unsatisfied-double-edges-are-sinks} that $T_i = T_j$
  is a sink.  Then the critical triple is also satisfied, since we
  chose an order that is preserved by the permutation $\varphi$
  corresponding to the critical double arc.  Thus, all critical
  triples containing P-nodes in~$T$ are satisfied.  Additionally, the
  Q-constraints are satisfied and since $D$ is 1-critical every P-node
  $\mu$ in $T$ is contained in at most one critical triple.  Hence, we
  can apply Lemma~\ref{lem:simult-extension-of-orders} to extend the
  orders $O_1, \dots, O_\ell$ simultaneously to an order $O$
  represented by~$T$.  This extension can clearly be computed in
  polynomial time and hence $D_{\ex}$ can be traversed bottom-up
  choosing an order for every PQ-tree in polynomial time in the size
  of $D_{\ex}$.
\end{proof}

As mentioned above, the expansion graph can be exponentially large for
instances that are not 1-critical, which can be seen as follows.
Assume a P-node $\mu$ in the PQ-tree $T$ is fixed with respect to
three children $T_1$, $T_2$ and $T_3$.  Then this P-node is
responsible for the three expansion trees $T(\mu, T_1, T_2)$, $T(\mu,
T_1, T_3)$ and $T(\mu, T_2, T_3)$.  So every layer can be three times
larger than the layer above, hence the expansion graph may be
exponentially large even if there are only linearly many layers.  But
if we can ensure that $\mu$ is fixed with respect to at most two
children of $T$, that is, it is contained in at most one critical
triple, it is responsible for only one expansion tree.  Of course, the
resulting expansion tree can itself contain several P-nodes that can
again be responsible for new expansion trees.  We first prove a
technical lemma followed by a lemma stating that the size of the
expansion graph remains quadratic in the size of $D$ for 1-critical
instances.

\begin{lemma}
  \label{lem:unequality-for-linear-hight}
  If $\mu$ is a P-node responsible for an expansion tree $T$
  containing the P-nodes $\mu_1, \dots, \mu_k$, the following
  inequality holds.
  \begin{equation*}
    \sum_{i = 1}^k \deg(\mu_i) \le \deg(\mu) + 2k - 2
  \end{equation*}
\end{lemma}
\begin{proof}
  Let $\eta_1, \dots, \eta_\ell$ be the Q-nodes contained in~$T$ and
  let $n_1$ be the number of leaves in~$T$.  Let further $n$ and $m$
  denote the number of vertices and edges in $T$, respectively.  We
  obtain the following equation by double counting.
  \begin{equation}
    \label{eq:1}
    n_1 + \sum_{i = 1}^k \deg(\mu_i) + \sum_{i = 1}^\ell \deg(\eta_i) = 2m
  \end{equation}
  Since $T$ is a tree, we can replace $m$ by $n-1$ and due to the fact
  that every node in $T$ is either a leaf, a P-node or a Q-node, we
  can replace $n$ further by $n_1 + k + \ell$.  With some additional
  rearrangement we obtain the following from Equation~\eqref{eq:1}.
  \begin{equation}
    \label{eq:2}
    \sum_{i = 1}^k \deg(\mu_i) = n_1 + 2k - 2 + 2\ell - \sum_{i = 1}^\ell \deg(\eta_i)
  \end{equation}
  The tree $T$ has at most $\deg(\mu)$ leaves since it is obtained by
  projecting some PQ-tree to representatives of the edges incident to
  $\mu$, yielding the inequality $n_1 \le \deg(\mu)$.  Additionally,
  we have the inequality $2\ell - \sum \deg(\eta_i) \le 0$ since
  $\deg(\eta_i) \ge 3$.  Plugging these two inequalities into
  Equation~\eqref{eq:2} yields the claim.
\end{proof}

\begin{lemma}
  \label{lem:exp-graph-poly-size}
  Let $D$ be a 1-critical instance of {\sc Simultaneous PQ-Ordering}.
  The size of its expansion graph $D_{\ex}$ is quadratic in $|D|$.
\end{lemma}
\begin{proof}
  We first show that the total size of all expansion trees is in
  $\mathcal O(|D|^2)$.  Afterwards, we show that the size of all arcs
  that are contained in $D_{\ex}$ but not in $D$ is linear in the
  total size of all expansion trees in $D_{\ex}$.

  Every expansion tree $T$ in $D_{\ex}$ has a P-node that is
  responsible for it.  If this P-node is itself contained in an
  expansion tree, we can again find another responsible P-node some
  layers above.  Thus, we finally find a P-node $\mu$ that was already
  contained in $D$, which is \emph{transitively responsible} for the
  expansion tree $T$.  Every PQ-tree for which $\mu$ is transitively
  responsible can have at most $\deg(\mu)$ leaves, thus its size is
  linear in $\deg(\mu)$ due to Lemma~\ref{lem:tree-without-deg-2}.
  Furthermore, we show that $\mu$ can only be transitively responsible
  for $\mathcal O(\deg(\mu))$ expansion trees, and thus for expansion
  trees of total size $\mathcal O(\deg(\mu)^2)$.  With this estimation
  it is clear that the size of all expansion trees is quadratic in the
  size of $D$.  To make it more precisely, denote the number of
  PQ-trees $\mu$ is transitively responsible for by $\resp(\mu)$.  We
  show by induction over $\deg(\mu)$ that $\resp(\mu) \le 3\deg(\mu) -
  8$.

  A P-node $\mu$ with $\deg(\mu) = 3$ can be responsible for at most
  one PQ-tree, thus $\resp(\mu) \le 3\deg(\mu) - 8$ is satisfied.  If
  $\mu$ has $\deg(\mu) > 3$ incident edges, it is directly responsible
  for at most one expansion tree $T$, since our instance is
  1-critical.  In the special case that $T$ consists of a single
  P-node $\mu'$ with $\deg(\mu') = \deg(\mu)$, the PQ-tree for which
  $\mu'$ is responsible cannot again contain a P-node of degree
  $\deg(\mu)$ due to Lemma~\ref{lem:chain-of-p-nodes-terminates}.
  Otherwise, $T$ contains $k$ P-nodes $\mu_1, \dots, \mu_k$ with
  $\deg(\mu_i) < \deg(\mu)$.  In the special case, $\resp(\mu) =
  \resp(\mu') + 1$ holds and we show the inequality $\resp(\mu) \le
  3\deg(\mu) - 8$ for both cases by showing $\resp(\mu) \le 3\deg(\mu)
  - 9$ for the second case.  In the second case, $\mu$ is transitively
  responsible for $T$ and all the PQ-trees $\mu_1, \dots, \mu_k$ are
  responsible for, yielding the following equation.
  \begin{equation*}
    \resp(\mu) = 1 + \sum_{i = 1}^k \resp(\mu_i)
  \end{equation*}
  Plugging in the induction hypothesis $\resp(\mu_i) \le 3\deg(\mu_i)
  - 8$ yields the following inequality.
  \begin{equation*}
    \resp(\mu) \le 1 + 3\sum_{i = 1}^k \deg(\mu_i) - 8k
  \end{equation*}
  If $k = 1$, this inequality directly yields the claim $\resp(\mu)
  \le 3\deg(\mu) - 9$ since $\deg(\mu_1) \le \deg(\mu) - 1$.
  Otherwise, we can use Lemma~\ref{lem:unequality-for-linear-hight} to
  obtain $\resp(\mu) \le 3\deg(\mu) - 5 - 2k$.  This again yields the
  claim $\resp(\mu) \le 3\deg(\mu) - 9$ since $k > 1$.  Finally, we
  have that the induction hypothesis holds for $\mu$, and hence every
  P-node is transitively responsible for $\mathcal O(\deg(\mu))$
  expansion trees of size $\mathcal O(\deg(\mu))$.

  For an arc that is contained in $D_{\ex}$ but not in $D$ consider
  the critical triple $(\mu, T_1, T_2)$ that is responsible for it.
  Since $\mu$ is not contained in another critical triple, it is only
  responsible for the arcs $(T_1, T(\mu, T_1, T_2))$ and $(T_2, T(\mu,
  T_1, T_2))$ or $(T_1, T_2)$ in the case of a finalizing step.  The
  size of these arcs is in $\mathcal O(\deg(\mu))$ since the expansion
  tree contains at most $\deg(\mu)$ leaves and, if the finalizing step
  is applied, $T_1$ and $T_2$ are single P-nodes of degree at most
  $\deg(\mu)$.  Hence, the size of newly created arcs in $D_{\ex}$ is
  linear in the size of all PQ-trees in $D_{\ex}$, which concludes the
  proof.
\end{proof}

Putting Lemma~\ref{lem:solve-1-critical-instances} and
Lemma~\ref{lem:exp-graph-poly-size} together directly yields the
following theorem.  For a detailed runtime analysis see
Section~\ref{sec:impl-deta}, showing that quadratic time is
sufficient, which is not as obvious as it seems to be.

\begin{theorem}
  \label{thm:solve-1-crit-inst-in-poly-time}
  {\sc Simultaneous PQ-Ordering} can be solved in polynomial time for
  1-critical instances.
\end{theorem}

Actually, Theorem~\ref{thm:solve-1-crit-inst-in-poly-time} tells us
how to solve 1-critical instances, which was the main goal of this
section.  However, the characterization of the 1-critical instances is
not really satisfying, since we need to know the expansion graph,
which may be exponentially large, to check whether an instance is
1-critical or not.  For our applications we can ensure that all
instances are 1-critical and hence do not need to test it
algorithmically.  But to prove for an application that all instances
are 1-critical, it would be much nicer to have conditions for
1-criticality of an instance that are defined for the instance itself
and not for some other structure derived from it.  In the remaining
part of this section we will provide sufficient conditions for an
instance to be 1-critical that do not rely on the expansion graph.

Let $D = (N, A)$ be an instance of {\sc Simultaneous PQ-Ordering}.
Let further $T$ be a PQ-tree with a parent $T'$ and let $\mu$ be a
P-node in $T$.  Recall that there is exactly one P-node~$\mu'$ in $T'$
it stems from, that is, $\mu'$ is fixed with respect to $\mu$ and no
other P-node in $T'$ is fixed with respect to $\mu$.  Note that there
may be several P-nodes in $T$ stemming from $\mu'$.  Consider a P-node
$\mu$ in the PQ-tree $T \in N$ such that $T$ is a source in $D$.  We
define the \emph{fixedness} $\fixed(\mu)$ of~$\mu$ to be the number of
children fixing it.  Now let~$\mu$ be a P-node of some internal
PQ-tree~$T$ of~$D$ with parents~$T_1,\dots,T_\ell$.  Each of the
trees~$T_i$ contains exactly one P-node~$\mu_i$ that is fixed
by~$\mu$.  Additionally, let~$k'$ be the number of children
fixing~$\mu$.  We set~$\fixed(\mu) = k' + \sum (\fixed(\mu_i) - 1)$.
We say that a P-node~$\mu$ is \emph{$k$-fixed}, if~$\fixed(\mu) \le k$
and an instance~$D$ is~\emph{$k$-fixed} for some integer~$k$ if all
its P-nodes are $k$-fixed.  The motivation for this definition is that
a P-node with fixedness $k$ in $D$ is fixed with respect to at most
$k$ children in the expansion graph $D_{\ex}$.  We obtain the
following theorem providing sufficient conditions for $D$ to be a
1-critical instance.

\begin{theorem}
  \label{thm:1-critical-instances-sufficient-cond}
  Every 2-fixed instance of {\sc Simultaneous PQ-Ordering} is
  1-critical.
\end{theorem}
\begin{proof}
  Let~$D$ be a 2-fixed instance of {\sc Simultaneous PQ-Ordering} and
  let $D_{\ex}$ be its expansion graph.  We need to show for every
  P-node $\mu$ in $D_{\ex}$ that it is contained in at most one
  critical triple, that is, it is fixed with respect to at most two
  children.  We will show that separately for the cases where the tree
  $T$ containing $\mu$ is already contained in $D$ and where $T$ is an
  expansion tree.

  Assume that $T$ is already contained in $D$.  It is clear that $\mu$
  is fixed with respect to at most two children in $D$, since it is at
  most 2-fixed, but it may happen that $T$ has additional children in
  $D_{\ex}$.  We will show by induction over the depth of the node $T$
  in $D_{\ex}$ that $\mu$ has at most $\fixed(\mu)$ children fixing it
  in $D_{\ex}$.  Recall that the depth of a node in a DAG is defined
  as the length of the longest directed path from a source to this
  node.  For sources in $D$ it is clear that the number of children
  fixing a P-node does not increase by expanding $D$, which shows the
  base case.  For the general case let $T_1, \dots, T_\ell$ be the
  parents of $T$ and let $\mu_1, \dots, \mu_\ell$ be the corresponding
  P-nodes $\mu$ stems from.  Let further $\mu$ be fixed with respect
  to $k'$ children of $T$ in $D$.  By the definition of fixedness we
  have $\fixed(\mu) = k' + \sum (\fixed(\mu_i) - 1)$.  Note that
  $\fixed(\mu_i) \ge 1$ for every $i = 1, \dots, \ell$ since $\mu_i$
  is at least fixed with respect to $T$ and note further, that $T_i$
  has by induction at most $\fixed(\mu_i)$ children fixing $\mu_i$.
  Thus, $\mu_i$ can be contained in at most $\fixed(\mu_i) - 1$
  critical triples also containing $T$, which means, that $\mu_i$ can
  be responsible for at most $\fixed(\mu_i) - 1$ children of $T$ in
  $D_{\ex}$.  Hence, $T$ can have in $D_{\ex}$ at most $k' + \sum
  (\fixed(\mu_i) - 1) = \fixed(\mu)$ children fixing $\mu$.  By the
  assumption that $\fixed(\mu) \le 2$ we obtain that $\mu$ is
  contained in at most one critical triple in $D_{\ex}$.

  Now consider the case where $T$ is an expansion tree with P-node
  $\mu$.  At the time $T$ is created, it has two incoming and no
  outgoing arcs, denote the parents by $T_1$ and $T_2$, and the
  P-nodes $\mu$ stems from by $\mu_1$ and $\mu_2$, respectively.
  Again we show by induction over the depth of $T$ in $D_{\ex}$ that
  $T$ has at most two children fixing $\mu$.  In the base case, $T_1$
  and $T_2$ are both already contained in $D$.  As shown above,
  $\mu_1$ and $\mu_2$ can each be contained in at most one critical
  triple, hence expansion can introduce at most two children fixing
  $\mu$.  In the general case, a parent $T_i$ for $i = 1, 2$ is either
  contained in $D$ or an expansion graph.  In the first case it again
  can introduce at most one child fixing $\mu$, in the second case we
  can apply the induction hypothesis with the same result.  Note that
  in a finalizing step for one of the trees a new incoming arc is
  created instead of an outgoing arc.  But this incoming arc can
  itself of course be responsible for at most one outgoing arc, hence
  the number of children fixing a P-node cannot become larger than
  two.  Finally, we have that every P-node in every PQ-tree in
  $D_{\ex}$ is fixed with respect to at most two children, hence $D$
  is 1-critical.
\end{proof}

Theorem~\ref{thm:solve-1-crit-inst-in-poly-time} and
Theorem~\ref{thm:1-critical-instances-sufficient-cond} together
provide a framework for solving problems that can be formulated as
instances of {\sc Simultaneous PQ-Ordering}.  We can use
Theorem~\ref{thm:1-critical-instances-sufficient-cond} to prove that
the instances our application produces are 1-critical, whereas
Theorem~\ref{thm:solve-1-crit-inst-in-poly-time} tells us that we can
solve these instances in polynomial time.  Note that since the
Q-constraints are expressed as a {\sc 2-Sat} formula, it is also not
difficult to completely fix the orientations of some Q-nodes.

\subsection{Implementation Details}
\label{sec:impl-deta}

To solve an instance of {\sc Simultaneous PQ-Ordering}, we first normalize the
instance, then compute the expansion graph and finally choose orders bottom-up.
As shown in Lemma~\ref{lem:exp-graph-poly-size} the size of the expansion graph
is quadratic in the size of $D$.  All other steps that need to be applied are
simple, such as projection, intersection or the extension of an order.  All
these steps run in linear time, but unfortunately linear in the size of the
parent.  For example, in the normalization step the projection of a tree $T$ to
the leaves of its child $T'$ must be computed, consuming linear time in~$|T|$.
Since $T$ can be a large PQ-tree with many small children we need quadratic
time.  A similar problem arises when computing an expansion tree due to a
critical triple $(\mu, T_1, T_2)$.  To compute $T(\mu, T_1, T_2)$ the trees
$T_1$ and $T_2$ need to be projected to representatives of the commonly fixed
edges incident to $\mu$, consuming $\mathcal O(|T_1| + |T_2|)$ time.  Since the
resulting expansion tree $T(\mu, T_1, T_2)$ can be arbitrarily small, these
costs cannot be expressed in terms of $|T(\mu, T_1, T_2)|$.  But since $T_1$ and
$T_2$ can have linearly many expansion trees as children we potentially need
quadratic time for each PQ-tree in $D_{\ex}$ to compute the expansion graph,
yielding an $\mathcal O(|D|^4)$ time algorithm.  Another problem is the
extension of orders bottom-up.  If a PQ-tree $T$ has one child $T'$ with chosen
order, it is easy to extend this order to $T$ in $|T|$ time.  However, $T$ can
have linearly many children, yielding an algorithm consuming quadratic time per
PQ-tree and thus overall again $\mathcal O(|D|^4)$ time.  However, if
additionally the projection $\proj T {L'} $ of $T$ to the leaves $L'$ of $T'$ is
known, the order chosen for $T'$ can be extended in $\mathcal O(|T'|)$ time to
$\proj T {L'}$.  Furthermore, the extension of orders from several projections
of $T$ to $T$ can be done in time linear in the size of all projections, if some
additional projection information are stored.  In this section we show how to
compute the normalization in quadratic time, which is straightforward.
Afterwards, we give a more detailed estimation for the size of the expansion
graph of 1-critical instances.  Then, we show that computing the expansion graph
for 1-critical instances actually runs in quadratic time.  Furthermore, we show
for the normalization and the expansion that for every arc the projection of the
parent to the leaves of the child together with additional projection
information can be computed and stored without consuming additional time.  This
information can then be used to choose orders bottom-up in linear time in the
size of the expansion graph.  Altogether, this yields a quadratic time algorithm
to solve 1-critical instances of {\sc Simultaneous PQ-Ordering}.

In the remaining part of this section let $D = (N, A)$ be a 1-critical instance
of {\sc Simultaneous PQ-Ordering} with the expansion graph $D_{\ex} = (N_{\ex},
A_{\ex})$.  Let further $|D|$, $|N|$, $|A|$, $|D_{\ex}|$, $|N_{\ex}|$ and
$|A_{\ex}|$ denote the size of $D$, $N$, $A$, $D_{\ex}$, $N_{\ex}$ and
$A_{\ex}$, respectively.  Recall that the size of a node is linear in the size
of the contained PQ-tree and the size of an arc is linear in the size of its
target, which is due to the injective map that needs to be stored for every arc.
Furthermore, let $p_{\max}$ be the degree of the largest P-node in $D$ and let
$\#N$ denote the number of nodes in $D$.

\paragraph{Normalization.}

As mentioned above, we want to compute and store some additional information
besides computing the normalization.  In detail, let $(T, T')$ be an arc and let
$L'$ be the leaves of~$T'$.  For every node in the projection $\proj T {L'}$ of
$T$ to the leaves of $T'$ there is a node in $T$ it stems from and for every
edge incident to a P-node in the projection there is an edge incident to the
corresponding P-node in $T$ it stems from.  We say that the arc $(T, T')$ has
\emph{additional projection information}, if $\proj T {L'}$ with a pointer from
every node and edge to the node and edge in $T$ it stems from is known.  Note
that the arc $(T, T')$ does not become asymptotically larger due to additional
projection information.  In the following, being a normalized instance of {\sc
  Simultaneous PQ-Ordering} includes that every arc has additional projection
information.  The following lemma is not really surprising.

\begin{lemma}
  \label{lem:time-normalization}
  An instance $D = (N, A)$ of {\sc Simultaneous PQ-Ordering} can be normalized
  in $\mathcal O(\#N \cdot |N|)$ time.
\end{lemma}
\begin{proof}
  To normalize an instance $D$ of {\sc Simultaneous PQ-Ordering} we need to
  project $T$ to the leaves of $T'$ and intersect the result with $T'$ for every
  arc $(T, T')$ in $D$.  The projection can be done in $\mathcal O(|T|)$ time,
  while the intersection consumes $\mathcal O(|T'|)$ time.  Note that the
  additional projection information can be simply stored directly after
  computing the projection.  Since $T$ may have $\#N$ children all these
  projections consume $\mathcal O(\#N \cdot |T|)$ time.  Summing over all
  PQ-trees yields $\mathcal O(\#N \cdot |N|)$ for the normalization of $D$.
\end{proof}

\paragraph{Size of the Expansion Graph.}

In Lemma~\ref{lem:exp-graph-poly-size} we already showed that the expansion
graph of a 1-critical instance has quadratic size.  However, this can be done
more precisely.

\begin{lemma}
  \label{lem:time-exp-size}
  Let $D$ be a 1-critical instance of {\sc Simultaneous PQ-Ordering} with the
  expansion graph $D_{\ex}$.  It holds $|D_{\ex}| \in \mathcal O(p_{\max} \cdot
  |N| + |A|)$, where $p_{\max}$ is the degree of the largest P-node in $D$.
\end{lemma}
\begin{proof}
  The proof of Lemma~\ref{lem:exp-graph-poly-size} shows that every P-node $\mu$
  can be transitively responsible for at most $3\deg(\mu) - 8$ expansion trees
  where each of these expansion trees has size $\mathcal O(\deg(\mu))$.  Thus,
  $\mu$ is responsible for expansion trees of total size $\mathcal
  O(\deg(\mu)^2)$.  To compute the total size of all expansion trees we need to
  sum over all P-nodes $\mu_1, \dots, \mu_\ell$ that are already contained in
  $D$.  The following estimations show the claimed size of $\mathcal O(p_{\max}
  \cdot |N|)$.
  \begin{equation*}
    \sum_{i = 1}^\ell \deg(\mu_i)^2 \le p_{\max} \cdot \sum_{i = 1}^\ell\deg(\mu_i) 
    \le p_{\max} \cdot |N|
  \end{equation*}

  As mentioned in the proof of Lemma~\ref{lem:exp-graph-poly-size} the size of
  all newly created arcs in $D_{\ex}$ is linear in the size of all nodes in
  $D_{\ex}$.  Thus we obtain $|D_{\ex}| \in \mathcal O(p_{\max} \cdot |N| +
  |A|)$ for the whole expansion graph.
\end{proof}

\paragraph{Computing the Expansion Graph.}

When computing the expansion tree $T(\mu, T_1, T_2)$ due to the
critical triple $(\mu, T_1, T_2)$ we need to project $T_1$ and $T_2$
to the representatives of the commonly fixed edges incident to $\mu$.
Let $T$ denote the tree containing $\mu$ and let $L_1$ and $L_2$ be
the leaves of $T_1$ and $T_2$, respectively.  First, we need to find
the commonly fixed edges and a representative for each.  Assume that
the projections $\proj T {L_1}$ and $\proj T {L_2}$ are stored as
ensured by the normalization.  Then for every edge incident to $\mu$
it can be easily tested in constant time, if it is contained in both
projections, consuming $\mathcal O(\deg(\mu))$ time overall.  With a
simple traversal of $\proj T {L_i}$ (for $i = 1, 2$) representatives
of these commonly fixed edges can be found in $\mathcal O(|T_i|)$ time
and the projection of $T_i$ to these representatives can also be done
in $\mathcal O(|T_i|)$ time.  The intersection of the two projections
yields $T(\mu, T_1, T_2)$ in $\mathcal O(|T(\mu, T_1, T_2)|)$ time,
which can be neglected.  For the two newly created arcs $(T_1, T(\mu,
T_1, T_2))$ and $(T_2, T(\mu, T_1, T_2))$ we again need to ensure that
the additional projection information are stored.  However, this
projection was already computed and can simple be stored without
additional running time.  Hence the total running time for computing
the expansion tree $T(\mu, T_1, T_2)$ is in $\mathcal O(\deg(\mu) +
|T_1| + |T_2|)$.  Thus, a superficial analysis yields quadratic
running time in the size of the expansion graph.  However, we can do
better, as shown in the following lemma.

\begin{lemma}
  \label{lem:time-comp-exp-graph}
  The expansion graph $D_{\ex}$ of a 1-critical instance $D = (N, A)$
  of {\sc Simultaneous PQ-Ordering} can be computed in $\mathcal
  O(|N|^2)$ time.
\end{lemma}

\begin{proof}
  As mentioned above, computing the expansion tree $T(\mu, T_1, T_2)$
  consumes $\mathcal O (\deg(\mu) + |T_1| + |T_2|)$ time.  We consider
  this time as cost and show how to assign it to different parts of
  $D$ defining them to be responsible for this cost.  The cost
  $\mathcal O(\deg(\mu))$ can be simply assigned to $\mu$.  Since
  every P-node $\mu$ is contained in at most one critical triple this
  can happen at most once, yielding linear cost in total.  Assume
  without loss of generality that $|T_1| \ge |T_2|$.  In this case we
  only need to assign the cost $\mathcal O(|T_1|)$.  To do that, we
  consider three cases.

  \begin{figure}[tb]
    \centering \subcaptionbox{ \label{fig:ipl-details-a}}
    {\includegraphics[page=1]{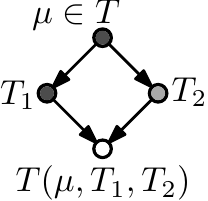}}\hspace{3em}
    \subcaptionbox{ \label{fig:ipl-details-b}}
    {\includegraphics[page=2]{fig/implementation-details}}\hspace{3em}
    \subcaptionbox{ \label{fig:ipl-details-c}}
    {\includegraphics[page=3]{fig/implementation-details}}
    \caption{Nodes in the original graph are shaded dark gray and
      expansion trees white.  Light gray is used where it does not
      matter.  (\subref{fig:ipl-details-a}) The case where $T_1$ is
      contained in the original graph.  (\subref{fig:ipl-details-b})
      The case where $T_1$ is an expansion graph but $T$ containing
      $\mu$ is not.  (\subref{fig:ipl-details-c}) The case where
      neither $T_1$ nor $T$ are expansion graphs.}
    \label{fig:implementation-details}
  \end{figure}

  {\bf If $\boldsymbol{T_1 \in N}$}, that is, $T_1$ is not an
  expansion tree, then we assign the cost $\mathcal O(|T_1|)$ to
  $T_1$.  This can happen at most as many times as $T_1$ occurs in a
  critical triple.  In each of these critical triples there
  necessarily is a P-node that is contained in a PQ-tree in a parent
  of~$T_1$.  There can be $\mathcal O(|N|)$ of these P-nodes and since
  every P-node is contained in at most one critical triple the total
  cost assigned to $T_1$ is in $\mathcal O(|N| \cdot |T_1|)$.  Note
  that no expansion tree is responsible for any cost, thus by summing
  over all PQ-trees in $D_{\ex}$ we obtain that the total cost is in
  $\mathcal O(|N|^2)$.  Figure~\ref{fig:ipl-details-a} illustrates
  this case.

  {\bf If $\boldsymbol{T_1 \not\in N}$ but $\boldsymbol{\mu \in T \in
      N}$}, that is, $T_1$ is an expansion tree, but the P-node $\mu$
  is contained in the original graph $D$.  Then $T_1$ has exactly two
  parents, like every other expansion tree, and of course one of them
  is the tree $T$ containing the P-node $\mu$.  Furthermore, there is
  a P-node $\mu_1$ responsible for $T_1$; let $T_1'$ be the PQ-tree
  containing $\mu_1$.  Thus $T_1$ was created due to a critical triple
  containing $\mu_1$ and $T$, and $T_1'$ containing $\mu_1$ needs to
  be a parent of $T$ as depicted in Figure~\ref{fig:ipl-details-b}.
  In this case we assign the cost $\mathcal O(|T_1|)$ to $T_1'$ or
  more precisely to $\mu_1$.  Since $T$ was already contained in the
  original graph, we also have $T_1' \in N$, thus again, only PQ-trees
  from the original graphs are responsible for any costs.  Since $T_1$
  is obtained by projecting $T$ and its other parent to
  representatives of edges incident to $\mu_1$ we have that $|T_1| \in
  \mathcal O(\deg(\mu_1))$.  Due to the fact that $\mu_1$ is contained
  in at most one critical triple it is overall responsible for
  $\mathcal O(\deg(\mu_1))$ cost and hence we obtain only linear cost
  by summing over all P-nodes in all PQ-trees in $D$.

  {\bf If $\boldsymbol{T_1 \not\in N}$ and $\boldsymbol{\mu \in T \not
      \in N}$}, that is, $T_1$ is an expansion tree and $\mu$ is
  contained in an expansion tree.  In other words, we are somehow
  ``far away'' from the original graph.  With the same argument as
  before, we can find a P-node $\mu'$ in a PQ-tree $T'$ that is
  responsible for the PQ-tree $T$ containing $\mu$ and this PQ-tree
  needs to be a parent of the PQ-tree $T_1'$; see
  Figure~\ref{fig:ipl-details-c}.  If $T'$ again is an expansion tree,
  we can find a P-node responsible for it and so on, until we reach a
  P-node $\mu''$ in the PQ-tree $T''$ that is transitively responsible
  for $T$ and $T'$, such that $T''$ is already contained in the graph
  $D$.  Then we assign the cost $\mathcal O(|T_1|)$ to $T''$ or more
  precisely to $\mu''$.  Since $T_1$ is a child of $T$ its size needs
  to be linear in $|T|$.  Furthermore, since $\mu''$ is transitively
  responsible for $T$, we have $|T| \in \mathcal O(\deg(\mu''))$.
  Thus we assign cost linear to $\deg(\mu'')$ to $\mu''$.  As shown
  for Lemma~\ref{lem:exp-graph-poly-size} $\mu''$ can be transitively
  responsible for at most $3\deg(\mu'') - 8$ expansion trees, thus it
  is overall responsible for $\mathcal O(\deg(\mu'')^2)$ cost.  Note
  that again only PQ-trees in $D$ are responsible for any costs.  Thus
  by summing over all P-nodes in all PQ-trees we obtain $\mathcal
  O(p_{\max} \cdot |N|)$.

  To sum up, the costs from the first case are dominating, hence we
  obtain a running time of $\mathcal O(|N|^2)$ for computing the
  expansion graph $D_{\ex}$ of a 1-critical instance $D = (N, A)$ of
  {\sc Simultaneous PQ-Ordering}.
\end{proof}

\paragraph{Extending Orders.}

As shown in Lemma~\ref{lem:solve-1-critical-instances}, {\sc
  Simultaneous PQ-Ordering} can be solved for 1-critical instances in
time polynomial in the size of the expansion graph.  There are three
things to do, first the Q-constraints need to be satisfied, which can
be checked in linear time, second the critical double arcs need to be
satisfied, which again can be done in linear time if possible, and
finally orders for the edges around P-nodes need to be chosen
bottom-up.  This is not obviously possible in linear time.  However,
the additional projection information that is stored for every arc
makes it possible, which is shown in the following lemma.

\begin{lemma}
  \label{lem:time-extending-orders}
  Let $D$ be a 1-critical instance of {\sc Simultaneous PQ-Ordering} with
  expansion graph~$D_{\ex}$.  In $\mathcal O(|D_{\ex}|)$ time we can compute
  simultaneous PQ-orders or decide that no such orders exist.
\end{lemma}
\begin{proof}
  The major work for this lemma was already done in the proof of
  Lemma~\ref{lem:solve-1-critical-instances}.  It remains to show how orders for
  the P-nodes can be chosen bottom-up in the expansion graph in linear time.

  Consider a PQ-tree $T$ in the expansion graph $D_{\ex}$ having the PQ-trees
  $T_1, \dots, T_\ell$ as children.  Assume further that orders $O_1, \dots,
  O_\ell$ are already chosen for the children.  The obvious approach to extend
  these orders simultaneously to an order represented by $T$ would take
  $\mathcal O(\ell \cdot |T|)$ time, yielding a worst case quadratic running
  time per PQ-tree in the expansion tree.  However, it can also be done in
  $\mathcal O(|T| + |T_1| + \dots + |T_\ell|)$ time, which can be seen as
  follows.  Let $T_i$ be one of the children of $T$ and let $T_i'$ be the
  projection of $T$ to the leaves of $T_i$, which was stored for the arc $(T,
  T_i)$ while normalizing and expanding.  Since $T_i'$ has as many leaves as
  $T_i$, we can apply the order $O_i$ to $T_i'$ in $\mathcal O(|T_i|)$ time,
  inducing an order of incident edges around every P-node of $T_i'$.  Now let
  $\mu_i$ be a P-node of $T_i'$ and let $\mu$ be the P-node in $T$ it stems
  from.  Recall that we can find $\mu$ in constant time and furthermore for an
  edge incident to $\mu_i$ we can find the edge incident to $\mu$ it stems
  from in constant time.  Thus, we can simply take the order of incident edges
  around $\mu_i$ and replace each edge by the edge incident to $\mu$ it stems
  from.  This order is then stored for~$\mu$.  Note that $\mu$ may store up to
  two orders in this way since it is fixed with respect to at most two
  children.  It is clear that this can be done in $\mathcal O(\deg(\mu_i))$
  time, thus processing all nodes in $T_i$ takes $\mathcal O(|T_i|)$ time.
  Now assume we have processed all children of $T$.  Then for the free P-nodes
  in $T$ there is nothing stored, for a P-node $\mu$ fixed with respect to one
  child there is one order given for a subset of edges incident to $\mu$ and
  for the P-nodes fixed with respect to two children there are two such
  orders.  In the first case, we can simply choose an arbitrary order for the
  edges incident to $\mu$, taking $\mathcal O(\deg(\mu))$ time.  In the second
  case, the free edges are added in an arbitrary way to the already ordered
  edges, which can again be done in $\mathcal O(\deg(\mu))$ time.  If we have
  two orders, these orders need to be merged, which can clearly be done in
  linear time.  Afterwards, the free edges can be added in an arbitrary way.
  This again consumes $\mathcal O(\deg(\mu))$ time.  Hence, we need for each
  node in $T$ linear time in its degree and hence $\mathcal O(|T|)$ for the
  whole tree.  Altogether we obtain the claimed $\mathcal O(|T| + |T_1| +
  \dots + |T_\ell|)$ running time for extending the orders $O_1, \dots,
  O_\ell$ to an order $O$ represented by $T$.  Recall, that $|T_i|$ is linear
  in the size of the arc $(T, T_i)$.  Thus, extending orders bottom-up in the
  expansion graph $D_{\ex} = (N_{\ex}, A_{\ex})$ takes $\mathcal O(|N_{\ex}| +
  |A_{\ex}|) = \mathcal O(|D_{\ex}|)$ time.
\end{proof}

\paragraph{Overall Running Time.}

For applications producing instances of {\sc Simultaneous PQ-Ordering} it may be
possible that reconsidering the runtime analysis containing normalization, size
and computation time of the expansion graph and order extension yields a better
running time then $\mathcal O(|N|^2)$.  However, for the general case we obtain
the following theorem by putting Lemma~\ref{lem:time-normalization},
Lemma~\ref{lem:time-exp-size}, Lemma~\ref{lem:time-comp-exp-graph} and
Lemma~\ref{lem:time-extending-orders} together.  Note that the running time is
dominated by the computation of the expansion graph.

\begin{theorem}
  \label{thm:time-solve-in-quadr-time}
  {\sc Simultaneous PQ-Ordering} can be solved in $\mathcal O(|N|^2)$ time for a
  1-critical instance $D = (N, A)$.
\end{theorem}

\subsection{Simultaneous PQ-Ordering with Reversing Arcs}
\label{sec:simult-pq-order-reversed}

As mentioned in Section~\ref{sec:spqr-trees} we can express all
embeddings of a biconnected planar graph in terms of PQ-trees by
considering the embedding tree $T(v)$ describing all possible orders
of incident edges around $v$, if we additionally ensure that Q-nodes
stemming from the same R-node in the SPQR-tree $\mathcal T$ are
oriented the same and pairs of P-nodes stemming from the same P-node
in $\mathcal T$ are ordered oppositely.  Forcing edges to be ordered
the same can be easily achieved with an instance of {\sc Simultaneous
  PQ-Ordering} by inserting a common child.  However, we want to
enforce edges around P-nodes to be ordered oppositely and not the
same.  Note that this cannot be achieved by simply choosing an
appropriate injective mapping from the leaves of the child to the
leaves of the parent, since it depends on the order if such a map
reverses it.

To solve this problem we introduce {\sc Simultaneous PQ-Ordering with
  Reversing Arcs}, which is an extension of the problem {\sc
  Simultaneous PQ-Ordering}.  Again, we have a DAG $D = (N, A)$ with
nodes $N = \{T_1, \dots, T_k\}$, such that every node $T_i$ is a
PQ-tree and every arc consists of a source $T_i$, a target $T_j$ and
an injective map $\varphi : L_j \rightarrow L_i$, where $L_i$ and
$L_j$ are the leaves of $T_i$ and $T_j$, respectively.  In addition to
that, every arc can be a \emph{reversing arc}.  Reversing arcs are
denoted by $(T_i, -T_j; \varphi)$, whereas normal arcs are denoted by
$(T_i, T_j; \varphi)$ as before.  {\sc Simultaneous PQ-Ordering with
  Reversing Arcs} asks whether there exist orders $O_1, \dots, O_k$
such that every normal arc $(T_i, T_j; \varphi) \in A$ implies that
$\varphi(O_j)$ is a suborder of $O_i$, whereas every reversing arc
$(T_i, -T_j; \varphi) \in A$ implies that the reversal of
$\varphi(O_j)$ is a suborder of $O_i$.  As for {\sc Simultaneous
  PQ-Ordering}, we define an instance of {\sc Simultaneous PQ-Ordering
  with Reversing Arcs} to be \emph{normalized}, if a normal arc $(T_i,
T_j; \varphi)$ implies that $\mathcal L_i$ contains an order $O_i$
extending $\varphi(O_j)$ for every order $O_j \in \mathcal L_i$ and a
reversing arc $(T_i, -T_j; \varphi)$ implies that $\mathcal L_i$
contains an order $O_i$ extending the reversal of $\varphi(O_j)$ for
every order $O_j \in \mathcal L_j$, where $\mathcal L_i$ and $\mathcal
L_j$ are the sets of orders represented by $T_i$ and $T_j$,
respectively.  Since $\mathcal L_i$ is represented by a PQ-tree, it is
closed with respect to reversing orders.  Thus, if $\mathcal L_i$
contains an order extending $\varphi(O_j)$, it also contains an order
extending the revers order of $\varphi(O_j)$.  Hence, we can normalize
an instance of {\sc Simultaneous PQ-Ordering with Reversing Arcs} in
the same way we normalize an instance of {\sc Simultaneous
  PQ-Ordering} by ignoring that some of the arcs are reversing.

In the following we show how to adapt the solution for {\sc
  Simultaneous PQ-Ordering} presented in the previous sections to
solve {\sc Simultaneous PQ-Ordering with Reversing Arcs}.  To give a
rough overview, the definitions of the Q-constraints and the critical
triples can be modified in a straight-forward manner, such that
Lemma~\ref{lem:simult-extension-of-orders}, stating that satisfying
the Q-constraints and the critical triples is necessary and sufficient
to be able to extend orders chosen for several PQ-trees to an order of
a common parent, is still true.  By declaring some of the created arcs
to be reversing, the definitions of expansion and finalizing step can
be easily adapted such that the resulting expansion trees and the
newly created arcs ensure that the responsible critical triples are
satisfied.  Thus, again the only critical triples that are not
automatically satisfied by choosing orders bottom-up correspond to
critical double arcs.
Lemmas~\ref{lem:chain-of-p-nodes-terminates},~\ref{lem:exp-graph-well-def}
and~\ref{lem:exp-graph-equiv} showing that the expansion graph is well
defined and equivalent to the original instance work in exactly the
same way.  For the definition of 1-critical instances there is no need
to change anything.
Lemma~\ref{lem:unsatisfied-double-edges-are-sinks} stating that
critical double arcs have a sink as target works as before.  In
Lemma~\ref{lem:solve-1-critical-instances} we showed how to solve
1-critical instances by testing whether the Q-constraints are
satisfiable and whether we can choose orders for the critical double
arcs satisfying the corresponding critical triple.  If this was the
case, we simply chose orders bottom-up.  Testing the Q-constraints can
now be done in the same way.  For the critical double arcs we can do
the same as before if both arcs are normal or both are reversing.  If
one of them is normal and the other is reversing, we need to check if
the corresponding permutation is order reversing instead of order
preserving, hence we use Lemma~\ref{lem:permutation-order-reverting}
instead of Lemma~\ref{lem:permutation-order-preserving}.  Afterwards,
it is again ensured that every critical triple is satisfied, hence we
can choose orders bottom-up as before.
Lemma~\ref{lem:exp-graph-poly-size} stating that the expansion graph
has quadratic size for 1-critical instances works as before, since the
only change in the definition of the expansion graph is that some arcs
are reversing arcs instead of normal arcs, which of course does not
change the size of the graph.  Finally, we can put
Lemma~\ref{lem:solve-1-critical-instances} and
Lemma~\ref{lem:exp-graph-poly-size} together yielding that {\sc
  Simultaneous PQ-Ordering with Reversing Arcs} can be solved in
polynomial time for 1-critical instances as stated before in
Theorem~\ref{thm:solve-1-crit-inst-in-poly-time} for {\sc Simultaneous
  PQ-Ordering}.
Theorem~\ref{thm:1-critical-instances-sufficient-cond} providing an
easy criterion that an instance is 1-critical works exactly the same
as before.

Let us start with the Q-constraints in more detail.  Let $\mu$ be a
Q-node in $T$ that is fixed with respect to the child $T'$ of $T$ and
let $\rep(\mu)$ be its representative in $T'$.  To ensure that $\mu$
is ordered as determined by $\rep(\mu)$, we introduced either the
constraint $x_\mu = x_{\rep(\mu)}$ or $x_\mu \not= x_{\rep(\mu)}$.
Now if the arc $(T, T')$ is reversing, we simply negate this
constraint, ensuring that $\mu$ is orientated oppositely to the
orientation determined by $\rep(\mu)$.  Let $\mu$ be a P-node in the
PQ-tree $T$ that is fixed with respect to two children $T_1$ and $T_2$
of $T$.  Then $\mu$, $T_1$ and $T_2$ together form again a critical
triple.  If both arcs $(T, T_1)$ and $(T, T_2)$ are normal arcs, we
denote this critical triple by $(\mu, T_1, T_2)$ as before.  If $(T,
-T_i)$ is a reversing arc, we symbolise that by a minus sign in the
critical triple, for example if we have the arcs $(T, T_1)$ and $(T,
-T_2)$, we denote the critical triple by $(\mu, T_1, -T_2)$.  Assume
we have orders $O_1$ and $O_2$ represented by $T_1$ and $T_2$,
respectively.  In the case that both arcs are normal or both are
reversing, we say that the critical triple is satisfied, if the edges
incident to $\mu$ fixed with respect to $T_1$ and $T_2$ are ordered
the same in both orders $O_1$ and $O_2$, which is the same definition
as before.  In the case that one of the arcs is normal and the other
is reversing, we define a critical triple to be satisfied if the order
$O_1$ induces the opposite order than $O_2$ for the commonly fixed
edges incident to $\mu$.  With these straight-forwardly adapted
definitions it is clear that the proof of
Lemma~\ref{lem:simult-extension-of-orders} works exactly as before.
To improve readability we cite this lemma here.

\begin{relemma}{lem:simult-extension-of-orders}
  \lemSimultExtensionOfOrdersText
\end{relemma}

This lemma implies that we can choose orders bottom-up, if we ensure
that the Q-constraints and the critical triples are satisfied, which
leads us to the definition of the expansion graph.  If we have a
critical triple $(\mu, (-)T_1, (-)T_2)$, in general we apply an
expansion step as before, that is, we project $T_1$ and $T_2$ to
representatives of the commonly fixed edges incident to $\mu$ and
intersect the result to obtain the expansion tree $T(\mu, (-)T_1,
(-)T_2)$.  Additionally, we add arcs from $T_1$ and $T_2$ to the
expansion tree.  The only thing we need to change is that the arc from
$T_i$ (for $i = 1, 2)$ to $T(\mu, (-)T_1, (-)T_2)$ is reversing if the
arc $(T, -T_i)$ is reversing.  Consider for example the critical
triple $(\mu, -T_1, T_2)$.  Then we have the reversing arcs $(T,
-T_1)$ and $(T_1, -T(\mu, -T_1, T_2))$ and the normal arcs $(T, T_2)$
and $(T_2, T(\mu, -T_1, T_2))$.  If we choose an order for the leaves
of $T(\mu, -T_1, T_2)$ representing the common fixed edges incident to
$\mu$, this order is reversed when it is extended to an order $O_1$
represented by $T_1$ and it remains the same by extension to an order
$O_2$ represented by~$T_2$.  Hence, the edges incident to $\mu$ fixed
with respect to $T_1$ and $T_2$ are ordered oppositely in $O_1$ and
$O_2$ implying that the critical triple $(\mu, -T_1, T_2)$ is
satisfied.  In other words by extending an order represented by
$T(\mu, -T_1, T_2)$ to an order of $T$ containing $\mu$ it is reversed
twice over the path containing $T_1$ yielding the same order as an
extension over the path containing $T_2$ not reversing it at all.  The
other three configurations work analogously.  The finalizing step can
be handled similarly.  If for a critical triple $(\mu, (-)T_1,
(-)T_2)$ both PQ-trees $T_1$ and $T_2$ consist of a single P-node
fixing the same edges incident to $\mu$, we obtain a bijection
$\varphi$ between the leaves of $T_1$ and the leaves of~$T_2$.  As
before, we create an arc from $T_2$ to $T_1$ with the map $\varphi$.
This new arc is a normal arc if both arcs $(T, (-)T_1)$ and $(T,
(-)T_2)$ are normal or if both are reversing.  If one is reversing and
one is normal, the new arc $(T_1, -T_2; \varphi)$ is reversing.
Again, this new arc ensures that the critical triple $(\mu, (-)T_1,
(-)T_2)$ is satisfied, if we choose orders bottom-up.  Note that we
need to consider the special case where we have a critical triple
$(\mu, (-)T', (-)T')$ due to a double arc.  As before we apply
expansion steps as if the children were different, ensuring that the
critical triple is satisfied.  Again, a finalizing step would
introduce a self loop, thus we simply prune expansion here (if $T'$ is
an expansion tree, otherwise we apply one more expansion step),
introducing an unsatisfied double arc.  The only difference to the
unsatisfied double arcs we had before is that the arcs may be
reversing.

For an instance $D$ of {\sc Simultaneous PQ-Ordering with Reversing
  Arcs}, we obtain the expansion graph $D_{\ex}$ by iteratively
applying expansion and finalizing steps.  Denote the expansion graph
that we would obtain from $D$ if we assume that all arcs are normal
by~$D_{\ex}'$.  It is clear that the only difference between $D_{\ex}$
and $D_{\ex}'$ is that some arcs in $D_{\ex}$ are reversing arcs.
Hence, everything we proved for the structure of the expansion graph
of an instance of {\sc Simultaneous PQ-Ordering} still holds if we
allow reversing arcs.  Particularly, we have that the expansion graph
is well defined (Lemma~\ref{lem:chain-of-p-nodes-terminates} and
Lemma~\ref{lem:exp-graph-well-def}), that the target of every
unsatisfied double arc is a sink if $D$ is 1-critical
(Lemma~\ref{lem:unsatisfied-double-edges-are-sinks}), that $|D_{\ex}|$
is polynomial in $|D|$ if $D$ is 1-critical
(Lemma~\ref{lem:exp-graph-poly-size}) and that $D$ is 1-critical if it
is at most 2-fixed
(Theorem~\ref{thm:1-critical-instances-sufficient-cond}).
Furthermore, all the implementation details provided in
Section~\ref{sec:impl-deta} still work.  Note that we say that an
instance $D$ is 1-critical if every P-node in every PQ-tree in
$D_{\ex}$ is contained in at most one critical triple, which is
exactly the same definition as before.

It remains to show, that the instances $D$ and $D_{\ex}$ are still
equivalent (Lemma~\ref{lem:exp-graph-equiv}) and that we can solve
$D_{\ex}$ by checking the Q-constraints, dealing with the unsatisfied
double arcs and finally choosing orders bottom-up, if $D$ is
1-critical (Lemma~\ref{lem:solve-1-critical-instances}).  In the proof
of Lemma~\ref{lem:exp-graph-equiv} we had to show that simultaneous
PQ-orders for all PQ-trees in $D$ induce simultaneous PQ-orders for
$D_{\ex}$.  That can be done analogously for the case where we allow
reversing arcs.  Most parts of the proof for
Lemma~\ref{lem:solve-1-critical-instances} can be adapted straight
forwardly since Lemma~\ref{lem:simult-extension-of-orders} still holds
if we allow reversing arcs.  The only difference is that the arcs in
an unsatisfied double arc can be reversing.  Consider an unsatisfied
double arc $(T, (-)T'; \varphi_1)$ and $(T, (-)T'; \varphi_2)$
together with the corresponding permutation $\varphi$ on the leaves of
$T'$.  If both arcs are normal or both are reversing, we need to check
if $\varphi$ is order preserving and choose an order that is preserved
by $\varphi$, which can be done due to
Lemma~\ref{lem:permutation-order-preserving}.  If, however, one of the
arcs is normal and the other is reversing, we need to check if
$\varphi$ is order reversing and then choose an order that is
reversed.  This is something we have not done before, but it can be
easily done by applying Lemma~\ref{lem:permutation-order-reverting}
instead of Lemma~\ref{lem:permutation-order-preserving}.  Finally,
Lemma~\ref{lem:solve-1-critical-instances} also works if we allow
reversing arcs and hence we obtain the following theorem analogously
to Theorem~\ref{thm:time-solve-in-quadr-time}

\begin{theorem}
  {\sc Simultaneous PQ-Ordering with Reversing Arcs} can be solved in
  $\mathcal O(|N|^2)$ time for a 1-critical instances $D = (N, A)$.
\end{theorem}

Now that we know that 1-critical instances of {\sc Simultaneous
  PQ-Ordering with Reversing Arcs} can be solved essentially in the
same way as 1-critical instances of {\sc Simultaneous PQ-Ordering} we
do not longer distinguish between these two problems.  Thus, if we
create 1-critical instances of {\sc Simultaneous PQ-Ordering} in our
applications, we allow them to contain reversing arcs.

\section{Applications}
\label{sec:applications}

As mentioned in Section~\ref{sec:spqr-trees} and again in
Section~\ref{sec:simult-pq-order-reversed} to motivate why reversing arcs are
necessary, we want to express all combinatorial embeddings of a biconnected
planar graph in terms of PQ-trees or more precisely in terms of an instance of
{\sc Simultaneous PQ-Ordering}.  A detailed description of this instance is
given in Section~\ref{sec:pq-embedd-repr}.  This representation is then used
to solve {\sc Partially PQ-Constrained Planarity} for biconnected graphs
(Section~\ref{sec:partially-pq-constr-planarity}) and {\sc Simultaneous
  Embedding with Fixed Edges} for biconnected graphs with a connected
intersection (Section~\ref{sec:simult-embedd-biconn-conn}).  Furthermore, we
show in Section~\ref{sec:simult-interv-graphs} and
Section~\ref{sec:extend-part-interv} how {\sc Simultaneous PQ-Ordering} can be
used to recognize simultaneous interval graphs and extend partial interval
representations in linear time.

\subsection{PQ-Embedding Representation}
\label{sec:pq-embedd-repr}

Let $G = (V, E)$ be a planar biconnected graph and let $\mathcal T$ be its
SPQR-tree.  We want to define an instance $D(G) = (N, A)$ of {\sc Simultaneous
  PQ-Ordering} called the \emph{PQ-embedding representation} containing the
embedding trees representing the circular order of edges around every vertex as
defined in Section~\ref{sec:spqr-trees}, such that it is ensured that every set
of simultaneous PQ-orders corresponds to an embedding of $G$ and vice versa.
For every R-node $\eta$ in $\mathcal T$, we define the PQ-tree $Q(\eta)$
consisting of a single Q-node with three edges and for every P-node $\mu$ in
$\mathcal T$ with $k$ virtual edges in $\skel(\mu)$ we define the PQ-tree
$P(\mu)$ consisting of a single P-node of degree~$k$.  The trees $Q(\eta)$ and
$P(\mu)$ will ensure that embedding trees of different vertices sharing R- or
P-nodes in the SPQR-tree are ordered consistently, thus we will call them the
\emph{consistency trees}.  The node set $N$ of the PQ-embedding representation
contains the consistency trees $Q(\eta)$ and $P(\mu)$ and the embedding trees
$T(v)$ for $v \in V$.  If we consider an R-node $\eta$ in the SPQR-tree
$\mathcal T$, then there are several Q-nodes in different embedding trees
stemming from it and we need to ensure that all these Q-nodes are oriented the
same or in other words we need to ensure that they are all oriented the same as
$Q(\eta)$, which can be done by simply adding arcs from the embedding trees to
$Q(\eta)$ with suitable injective maps.  Similarly, the skeleton of every P-node
$\mu$ in $\mathcal T$ contains two vertices $v_1$ and $v_2$.  Thus, the
embedding trees $T(v_1)$ and $T(v_2)$ contain P-nodes $\mu_1$ and $\mu_2$
stemming from $\mu$ and every incident edge corresponds to a virtual edge in
$\skel(\mu)$.  We need to ensure that the order of incident edges around $\mu_1$
is the reversal of the order of edges around $\mu_2$, or in other words, we need
to ensure that the order for $\mu_1$ is the same and the order for $\mu_2$ is
the opposite to any order chosen for $P(\mu)$, which can be ensured by a normal
arc $(T(v_1), P(\mu))$ and a reversing arc $(T(v_2), -P(\mu))$.  If we solve the
PQ-embedding representation $D(G)$ as instance of {\sc Simultaneous PQ-Ordering}
we would choose orders bottom-up.  Thus, we would first choose orders for the
trees $P(\mu)$ and $Q(\mu)$, which corresponds to choosing orders for the
P-nodes and orientations for the R-nodes in the SPQR-tree.  For the embedding
trees there is no choice left, since all nodes are fixed by some children, which
is not surprising since the planar embedding is already chosen.  Hence,
extending the chosen orders to orders of the embedding trees can be seen as
computing the circular orders of edges around every vertex for given embeddings
of the skeletons of every node in $\mathcal T$.
Figure~\ref{fig:PQ-embedding-rep-from-SPQR} depicts the PQ-embedding
representation for the example we had before in
Figure~\ref{fig:PQ-trees-with-common-nodes-from-SPQR}.  Note that the size of
the PQ-embedding representation $D(G)$ is obviously linear in the size of the
SPQR-tree $\mathcal T$ of $G$, and thus linear in the size of the planar graph
$G$ itself.

\begin{figure}[tb]
  \centering
  \includegraphics[page=2]{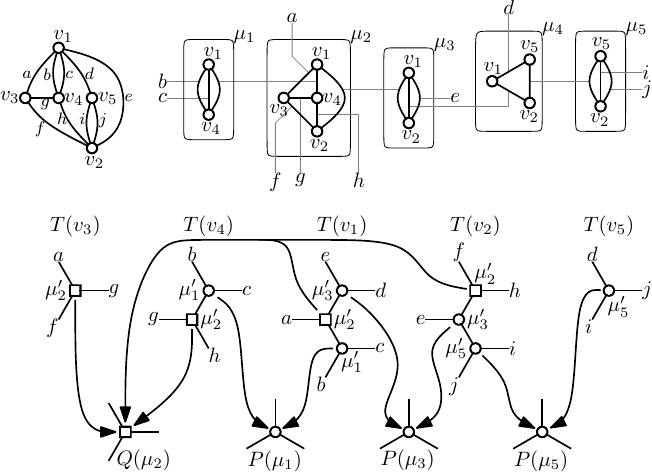}
  \caption{A biconnected planar graph and its SPQR-tree on the top and the
    corresponding PQ-embedding representation on the bottom.  The injective maps
    on the edges are not explicitly depicted, but the starting points of the
    arcs suggests which maps are suitable.}
  \label{fig:PQ-embedding-rep-from-SPQR}
\end{figure}

The PQ-embedding representation is obviously less elegant than the SPQR-tree,
also representing all embeddings of a biconnected planar graph.  At least for a
human, the planar embeddings of a graph are easy to understand by looking at the
SPQR-tree, whereas the PQ-embedding representation does not really help.
However, with the PQ-embedding representation it is easier to formulate
constraints concerning the order of incident edges around a vertex, since these
orders are explicitly expressed by the embedding trees.

\subsection{Partially PQ-Constrained Planarity}
\label{sec:partially-pq-constr-planarity}

Let $G = (V, E)$ be a planar graph and let $C = \{T'(v_1), \dots,
T'(v_n)\}$ be a set of PQ-trees, such that for every vertex $v_i \in
V$ the leaves of $T(v_i)$ are a subset $E'(v_i) \subseteq E(v_i)$ of
edges incident to $v_i$.  We call $T'(v_i)$ the \emph{constraint tree}
of the vertex $v_i$.  The problem {\sc Partially PQ-Constrained
  Planarity} asks whether a planar embedding of $G$ exists, such that
the order of incident edges $E(v_i)$ around every vertex $v_i$ induces
an order on $E'(v_i)$ that is represented by the constraint tree
$T'(v_i)$.

Given an instance $(G, C)$ of {\sc Partially PQ-Constrained
  Planarity}, it is straightforward to formulate it as an instance of
{\sc Simultaneous PQ-Ordering} if $G$ is biconnected.  Simply take the
PQ-embedding representation $D(G)$ of $G$ and add the constraint trees
together with an arc $(T(v), T'(v); \id)$ from the embedding tree to
the corresponding constraint tree.  Denote the resulting instance of
{\sc Simultaneous PQ-Ordering} by $D(G, C)$.
Figure~\ref{fig:planar-emb-with-PQ-const} depicts an example instance
of {\sc Partially PQ-Constrained Planarity} formulated as instance of
{\sc Simultaneous PQ-Ordering}.  Note that we can leave the orders of
edges around a vertex unconstrained by choosing the empty PQ-tree as
its constraint tree.  To obtain the following theorem, we need to show
that $(G, C)$ and $D(G, C)$ are equivalent, which is quite obvious,
and that $D(G, C)$ is an at most 2-fixed instance of {\sc Simultaneous
  PQ-Ordering}.

\begin{figure}[tb]
  \centering
  \includegraphics[page=3]{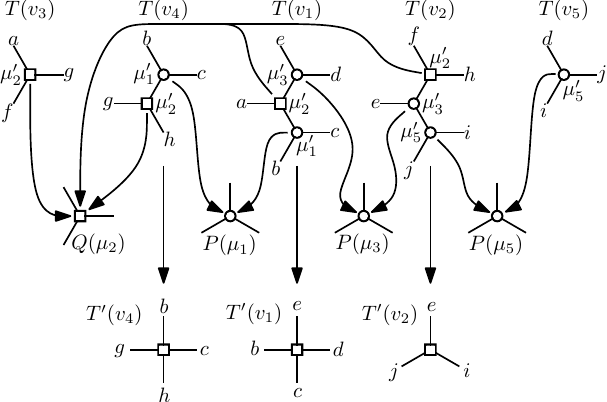}
  \caption{The PQ-embedding representation from
    Figure~\ref{fig:PQ-embedding-rep-from-SPQR} together with the
    constraint trees provided by an instance of {\sc Partially
      PQ-Constrained Planarity}.}
  \label{fig:planar-emb-with-PQ-const}
\end{figure}

\begin{theorem}
  \label{thm:plan-emb-with-pq-const-biconn}
  {\sc Partially PQ-Constrained Planarity} can be solved in quadratic
  time for biconnected graphs.
\end{theorem}
\begin{proof}
  Consider $(G, C)$ to be an instance of {\sc Partially PQ-Constrained
    Planarity} where $G$ is a biconnected planar graph and $C$ the set
  of constraint trees.  Let further $D(G, C)$ be the corresponding
  instance of {\sc Simultaneous PQ-Ordering}.  Since $D(G, C)$
  contains the PQ-embedding representation $D(G)$, every solution of
  $D(G, C)$ yields a planar embedding of $G$.  Additionally, this
  planar embedding respects the constraint trees since the order of
  edges around every vertex is an extension of an order of the leaves
  in the corresponding constraint tree.  On the other hand, it is
  clear that a planar embedding of $G$ respecting the constraint trees
  yields simultaneous orders for all trees in $D(G, C)$.  Since the
  size of $D(G, C)$ is linear in the size of $(G, C)$, we can solve
  $(G, C)$ in quadratic time using
  Theorem~\ref{thm:time-solve-in-quadr-time}, if $D(G, C)$ is
  1-critical.  We will show that the instance $D(G, C)$ is at most
  2-fixed, and hence, due to
  Theorem~\ref{thm:1-critical-instances-sufficient-cond} also
  1-critical.

  To compute the fixedness of every P-node in every PQ-tree in $D(G,
  C)$, we distinguish between three kinds of trees, the embedding
  trees, the consistency trees and the constraint trees.  If we
  consider a P-node $\mu$ in an embedding tree $T(v)$, this P-node is
  fixed with respect to exactly one consistency tree, namely the tree
  that represents the P-node in the SPQR-tree $\mu$ stems from.  In
  addition to the consistency trees, $T(v)$ has the constraint tree
  $T'(v)$ as child, thus $\mu$ can be fixed with respect to $T'(v)$.
  Since $T(v)$ has no parents and no other children, $\mu$ is at most
  2-fixed, that is $\fixed(\mu) \le 2$.  Consider a P-node $\mu'$ in a
  constraint tree $T'(v)$.  Since $T'(v)$ has no children and its only
  parent is $T(v)$ containing the P-node $\mu$ that is fixed by
  $\mu'$, we have by the definition of fixedness that $\fixed(\mu') =
  \fixed(\mu) - 1$.  Since $\mu$ is a P-node in an embedding tree we
  obtain $\fixed(\mu') \le 1$.  We have two kinds of consistency
  trees, some stem from P- and some from R-nodes in the SPQR-tree.  We
  need to consider only trees $P(\mu)$ stemming form P-nodes since the
  consistency trees stemming from R-nodes only contain a single
  Q-node.  Denote the single P-node in $P(\mu)$ also by $\mu$ and let
  $\mu_1$ and $\mu_2$ be the two P-nodes in the embedding trees
  $T(v_1)$ and $T(v_2)$ that are fixed with respect to $P(\mu)$.
  Since $P(\mu)$ has no child and only these two parents, we obtain
  $\fixed(\mu) = (\fixed(\mu_1) - 1) + (\fixed(\mu_2) - 1)$.  Since
  $\mu_1$ and $\mu_2$ are P-nodes in embedding trees this yields
  $\fixed(\mu) \le 2$.  Hence, all P-nodes in all PQ-trees in $D(G,
  C)$ are at most 2-fixed, thus $D(G, C)$ itself is 2-fixed.  Finally,
  we can apply Theorem~\ref{thm:1-critical-instances-sufficient-cond}
  yielding that $D(G, C)$ is 1-critical and thus can be solved
  quadratic time, due to Theorem~\ref{thm:time-solve-in-quadr-time}.
\end{proof}

Since $D(G, C)$ is a special instance of {\sc Simultaneous
  PQ-Ordering}, which seems to be quite simple, it is worth to make a
more detailed runtime analysis, yielding the following theorem.

\begin{theorem}
  \label{thm:plan-emb-with-pq-const-biconn-lin-time}
  {\sc Partially PQ-Constrained Planarity} can be solved in linear
  time for biconnected graphs.
\end{theorem}
\begin{proof}
  As figured out in Section~\ref{sec:impl-deta} about the
  implementation details, there are four major parts influencing the
  running time.  First, a given instance needs to be normalized
  consuming quadratic time (Lemma~\ref{lem:time-normalization}), the
  expansion graph has quadratic size in worst case
  (Lemma~\ref{lem:time-exp-size}) and its computation consumes
  quadratic time (Lemma~\ref{lem:time-comp-exp-graph}) and finally
  choosing borders bottom-up needs linear time in the size of the
  expansion graph (Lemma~\ref{lem:time-extending-orders}).

  In an instance $D(G, C)$ of {\sc Simultaneous PQ-Ordering} stemming
  from an instance $(G, C)$ of {\sc Partially PQ-Constrained
    Planarity} there are two kinds of arcs.  First, arcs from
  embedding trees to consistency trees, and second, arcs from
  embedding trees to constraint trees.  When normalizing an arc from
  an embedding tree to a consistency tree there is nothing to do,
  since there is a bijection between the consistency tree and an inner
  node of the embedding tree.  The arcs from embedding trees to
  constraint trees can be normalized as usual consuming only linear
  time, since each embedding tree has only one consistency tree as
  child.  Hence, normalization can be done in linear time.  When
  computing the expansion graph, the fixedness of the nodes is
  important.  As seen in the proof of
  Theorem~\ref{thm:plan-emb-with-pq-const-biconn}, the P-nodes in
  embedding and consistency trees are at most 2-fixed, whereas the
  P-nodes in constraint trees are at most 1-fixed.  Note that every
  critical triple $(\mu, T_1, T_2)$ in $D(G, C)$ is of the kind that
  $\mu$ is contained in an embedding tree, $T_1$ is a constraint tree
  and $T_2$ is a consistency tree.  Thus, the expansion tree $T(\mu,
  T_1, T_2)$ created due to such a triple has two parents where one of
  them is at most 1-fixed and the other at most 2-fixed.  Hence, by
  the definition of fixedness, $T(\mu, T_1, T_2)$ itself is at most
  1-fixed.  After creating these expansion trees, all newly created
  critical triple must contain a P-node $\mu$ in a consistency tree
  and two expansion trees.  By creating expansion trees for these
  critical triples no new critical triple are created and hence the
  expansion stops.  It is clear that the resulting expansion graph has
  only linear size and can be computed in linear time.  Choosing
  orders bottom-up takes linear time in the size of the expansion
  graph, as before.  Hence we obtain the claimed linear running time.
\end{proof}

\subsection{Simultaneous Embedding with Fixed Edges}
\label{sec:simult-embedd-biconn-conn}

Let $\1G = (\1V, \1E)$ and $\2G = (\2V, \2E)$ be two planar graphs
sharing a common subgraph $G = (V, E)$ with $V = \1V \cap \2V$ and $E
= \1E \cap \2E$.  {\sc Simultaneous Embedding with Fixed Edges} asks,
whether there exist planar drawings of $\1G$ and $\2G$ such that their
intersection $G$ is drawn the same in both.  Jünger and Schulz show
that this is equivalent to the question whether combinatorial
embeddings of $\1G$ and $\2G$ inducing the same combinatorial
embedding for their intersection $G$
exist~\cite[Theorem~4]{IntersectionGraphsin-Juenger.Schulz(09)}.

Assume that $\1G$ and $\2G$ are biconnected and $G$ is connected.
Then the order of incident edges around every vertex determines the
combinatorial embedding, which is not the case for disconnected
graphs.  Thus, we can reformulate the problem as follows.  Can we find
planar embeddings of $\1G$ and $\2G$ inducing for every common vertex
$v \in V$ the same order of common incident edges $E(v)$ around $v$?
Since both graphs are biconnected, they both have a PQ-embedding
representation and it is straightforward to formulate an instance
$(\1G, \2G)$ of {\sc SEFE} as an instance $D(\1G, \2G)$ of {\sc
  Simultaneous PQ-Ordering}.  The instance $D(\1G, \2G)$ contains the
PQ-embedding representations $D(\1G)$ and $D(\2G)$ of $\1G$ and $\2G$,
respectively.  Every common vertex $v \in V$ occurs as $\1v$ in $\1V$
and as $\2v$ in~$\2V$, thus we have the two embedding trees $T(\1v)$
and $T(\2v)$.  By projecting these two embedding trees to the common
edges incident to $v$ and intersecting the result, we obtain a new
tree $T(v)$ called the \emph{common embedding tree} of $v$.  If we add
the arcs $(T(\1v), T(v))$ and $(T(\2v), T(v))$ to the instance $D(\1G,
\2G)$ of {\sc Simultaneous PQ-Ordering}, we ensure that the common
edges incident to $v$ are ordered the same in both graphs.  Note that
this representation is quite similar to the representation of an
instance of {\sc Partially PQ-Constrained Planarity}.  Every common
embedding tree can be seen as a constraint tree for both graphs
simultaneously.  To obtain the following theorem, we need to show that
the instances $(\1G, \2G)$ of {\sc SEFE} and the instance $D(\1G,
\2G)$ of {\sc Simultaneous PQ-Ordering} are equivalent and that
$D(\1G, \2G)$ is at most 2-fixed.

\begin{theorem}
  \label{thm:simult-embedd-biconn-conn}
  {\sc Simultaneous Embedding with Fixed Edges} can be solved in
  quadratic time, if both graphs are biconnected and the common graph
  is connected.
\end{theorem}
\begin{proof}
  Let $(\1G, \2G)$ be an instance of {\sc SEFE} with the common graph
  $G$ such that $\1G$ and $\2G$ are biconnected and $G$ is connected.
  Let further $D(\1G, \2G)$ be the corresponding instance of {\sc
    Simultaneous PQ-Ordering} as defined above.  Since $D(\1G, \2G)$
  contains the PQ-embedding representations $D(\1G)$ and $D(\2G)$,
  every solution of $D(\1G, \2G)$ yields planar embeddings of $\1G$
  and $\2G$.  Furthermore, the common edges incident to a common
  vertex $v \in V$ are ordered the same in the two embedding trees
  $T(\1v)$ and $T(\2v)$ since both orders extend the same order of
  common edges represented by the common embedding tree~$T(v)$.  Thus,
  the embeddings for $\1G$ and $\2G$ induced by a solution of $D(\1G,
  \2G)$ induce the same embedding on the common graph and hence are a
  solution of $(\1G, \2G)$.  On the other hand, if we have a {\sc
    SEFE} of $\1G$ and $\2G$, these embeddings induce orders for the
  leaves of all PQ-trees in $D(\1G, \2G)$ and since the common edges
  around every common vertex are ordered the same in both embeddings,
  all constraints given by arcs in $D(\1G, \2G)$ are satisfied.

  To compute the fixedness of every P-node in every PQ-tree in $D(\1G,
  \2G)$ we distinguish between three kinds of trees, the embedding
  trees, the consistency trees and the common embedding trees.  The
  proof that $\fixed(\mu) \le 2$ for every P-node $\mu$ in every
  embedding and consistency tree works as in the proof of
  Theorem~\ref{thm:plan-emb-with-pq-const-biconn}.  For a P-node $\mu$
  in a common embedding tree $T(v)$ we have two P-nodes $\1\mu$ and
  $\2\mu$ in the parents $T(\1v)$ and $T(\2v)$ of $T(v)$ it stems
  from.  Since $T(v)$ has no other parents and no children, we obtain
  $\fixed(\mu) = (\fixed(\1\mu) - 1) + (\fixed(\2\mu) - 1)$ by the
  definition of fixedness.  Since $\1\mu$ and $\2\mu$ are P-nodes in
  embedding trees, we know that their fixedness is at most~2.  Thus,
  we have $\fixed(\mu) \le 2$.  Hence, all P-nodes in all PQ-trees in
  $D(\1G, \2G)$ are at most 2-fixed, thus $D(\1G, \2G)$ itself is
  2-fixed.
\end{proof}

\subsection{Simultaneous Interval Graphs}
\label{sec:simult-interv-graphs}

A graph $G$ is an \emph{interval graph}, if each vertex $v$ can be
represented as an interval $I(v) \subset \mathbb R$ such that two
vertices $u$ and $v$ are adjacent if and only if their intervals
intersect, that is, $I(u) \cap I(v) \not= \emptyset$.  Such a
representation is called \emph{interval representation} of $G$; see
Figure~\ref{fig:interval-graphs-single} for two examples.  Two graphs
$\1G$ and $\2G$ sharing a common subgraph are \emph{simultaneous
  interval graphs} if $\1G$ and $\2G$ have interval representations
such that the common vertices are represented by the same intervals in
both representations; see Figure~\ref{fig:interval-graphs-simult} for
an example.  The problem to decide whether $\1G$ and $\2G$ are
simultaneous interval graphs is called {\sc Simultaneous Interval
  Representation} having the pair $(\1G, \2G)$ as input.

\begin{figure}[tb]
  \centering \subcaptionbox{ \label{fig:interval-graphs-single}}
  {\includegraphics[page=1]{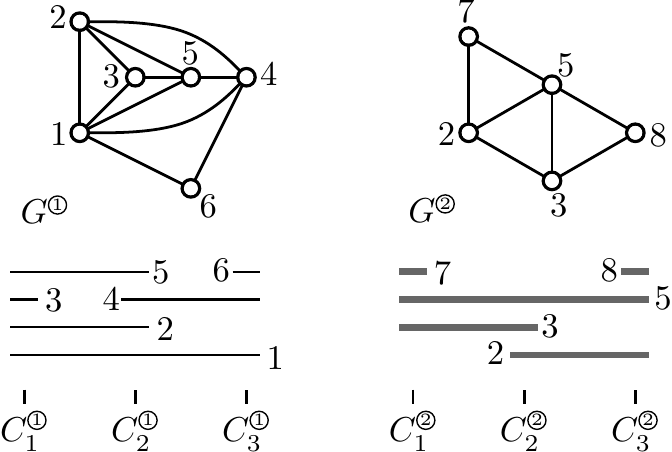}}\hspace{3em}
  \subcaptionbox{ \label{fig:interval-graphs-simult}}
  {\includegraphics[page=2]{fig/interval-graphs}}
  \caption{(\subref{fig:interval-graphs-single}) Two interval graphs
    $\1G$ and $\2G$ with interval representations.  The maximal
    cliques are $\1{C_1},
    \1{C_2}, \1{C_3}$ and $\2{C_1}, \2{C_2}, \2{C_3}$, respectively. \\
    (\subref{fig:interval-graphs-simult}) Interval representations of
    $\1G$ and $\2G$ such that common vertices are represented by the
    same interval in both representations, in other words, a
    simultaneous interval representation of $\1G$ and $\2G$.}
  \label{fig:interval-graphs}
\end{figure}

The first algorithm recognizing interval graphs in linear time was
given by Booth and
Lueker~\cite{TestingConsecutiveOnes-Booth.Lueker(76)} and was based on
a characterization by Fulkerson and
Gross~\cite{IncidenceMatricesand-Fulkerson.Gross(65)}.  This
characterization says that $G$ is an interval graph if and only if
there is a linear order of all its maximal cliques such that for each
vertex $v$ all cliques containing $v$ appear consecutively.  It is
easy to see that an interval graph can have only linearly many maximal
cliques thus it is clear how to recognize interval graphs in linear
time by using PQ-trees.  The problem {\sc Simultaneous Interval
  Representation} was first considered by Jampani and
Lubiw~\cite{SimultaneousIntervalGraphs-Jampani.Lubiw(10)} who show
how to solve it in $\mathcal O(n^2 \log n)$ time.

In Theorem~\ref{thm:interval-graph-char} we give a proof of the
characterization by Fulkerson and Gross that can then be extended to a
characterization of simultaneous interval graphs in
Theorem~\ref{thm:sim-interval-graph-char}.  With this characterization
it is straightforward to formulate an instance of {\sc Simultaneous
  PQ-Ordering} that can be used to test whether a pair of graphs are
simultaneous interval graphs in linear time, improving the so far
known result.  The following definition simplifies the notation.  Let
$C_1, \cdots, C_\ell$ be sets (for example maximal cliques) and let
$v$ be an element contained in some of these sets.  We say that a
linear order of these sets is \emph{$v$-consecutive} if the sets
containing $v$ appear consecutively.

\begin{theorem}[Fulkerson and
  Gross~\cite{IncidenceMatricesand-Fulkerson.Gross(65)}]
  \label{thm:interval-graph-char}
  A graph $G$ is an interval graph if and only if there is a linear
  order of all maximal cliques of $G$ that is $v$-consecutive with
  respect to every vertex $v$.
\end{theorem}
\begin{proof}
  Assume $G$ is an interval graph with a fixed interval
  representation.  Let $C = \{v_1, \dots, v_k\}$ be a maximal clique
  in $G$.  It is clear that there must be a position $x$ such that $x$
  is contained in the intervals $I(v_1), \dots, I(v_k)$.  Additionally
  $x$ is not contained in any interval represented by another vertex
  since the clique $C$ is maximal.  By fixing such positions $x_1,
  \dots, x_\ell$ for each of the maximal cliques $C_1, \dots, C_\ell$
  in $G$, we define a linear order on all maximal cliques.  Assume
  this order is not $v$-consecutive for some vertex $v$.  Then there
  are cliques $C_i, C_j, C_k$ with $x_i < x_j < x_k$ such that $v \in
  C_i, C_k$ but $v \notin C_j$.  However, since $v$ is in $C_i$ and
  $C_k$ its interval $I(v)$ needs to contain $x_i$ and $x_k$, and
  hence also $x_j$, which is a contradiction to the construction of
  the position $x_j$.  Hence the defined linear order of all maximal
  cliques is $v$-consecutive with respect to every vertex $v$.

  Now assume $O = C_1 \dots C_\ell$ is a linear order of all maximal
  cliques of $G$ that is $v$-consecutive for every vertex $v$.  Let
  $v$ be a vertex and let $C_i$ and $C_j$ be the leftmost and
  rightmost cliques containing $v$, respectively.  Then define $I(v) =
  [i, j]$ to be the interval representing $v$.  With this
  representation, we obtain all edges contained in the the maximal
  cliques $C_1, \dots, C_\ell$ at the natural numbers $1, \dots,
  \ell$, since for each clique $C_i = \{v_1, \dots, v_k\}$ the
  position $i$ is contained in all the intervals $I(v_1), \dots,
  I(v_k)$.  Furthermore, there is no vertex $u \notin C_i$ such that
  $I(u)$ also contains $i$, because such a vertex would need to be
  contained in a clique on the left and in a clique on the right to
  $C_i$, which is a contradiction since the order $O$ is
  $u$-consecutive.  Thus, at the integer positions $1, \dots, \ell$
  all edges in $G$ are represented and no edges not in $G$.
  Furthermore, all intervals $I(v)$ containing a non integer position
  $1 < x < \ell$ contain also $\lceil x \rceil$ and $\lfloor x
  \rfloor$, yielding that no edge is defined due the position $x$
  which is not already defined due to an integer position.  Hence,
  this definition of intervals is an interval representation of $G$
  showing that $G$ is an interval graph.
\end{proof}

We can extend this characterization of interval graphs to a
characterization of simultaneous interval graphs by using the same
arguments as follows.

\begin{theorem}
  \label{thm:sim-interval-graph-char}
  Two graphs $\1G$ and $\2G$ are simultaneous interval graphs if and
  only if there are linear orders of the maximal cliques of $\1G$ and
  $\2G$ that are $v$-consecutive with respect to every vertex $v$ in
  $\1G$ and $\2G$, respectively, such that they can be extended to an
  order of the union of maximal cliques that is $v$-consecutive with
  respect to every common vertex $v$.
\end{theorem}
\begin{proof}
  Assume $\1G$ and $\2G$ are simultaneous interval graphs and let for
  every vertex $v$ be $I(v)$ the interval representing $v$.  Assume
  $\1 {\mathcal C} = \{\1{C_1}, \dots, \1{C_k}\}$ and $\2 {\mathcal C}
  = \{\2{C_1}, \dots, \2{C_\ell}\}$ are the maximal cliques in $\1G$
  and $\2G$ respectively.  When considering $\1G$ for itself, we again
  obtain for every maximal clique $\1C = \{v_1, \dots, v_r\}$ a
  position $x$ such that $x$ is contained in $I(v_i)$ for every $v_i
  \in \1C$ but in no other interval representing a vertex in $\1G$.
  The same can be done for the maximal cliques of $\2G$, yielding a
  linear order $O$ of all maximal cliques $\mathcal C = \1 {\mathcal
    C} \cup \2 {\mathcal C}$.  It is clear that the projection of this
  order to the cliques in $\1G$ is $v$-consecutive for every vertex
  $v$ in $\1G$ due to Theorem~\ref{thm:interval-graph-char} and the
  same holds for $\2G$.  It remains to show that $O$ is
  $v$-consecutive for each common vertex $v$.  Assume $O$ is not
  $v$-consecutive for some common vertex $v$.  Then there need to be
  three cliques $C_i$, $C_j$ and~$C_k$, no matter if they are maximal
  cliques in $\1G$ or in $\2G$, with positions $x_i$, $x_j$ and $x_k$
  such that $x_i < x_j < x_k$ and $v \in C_i, C_k$ but $v \notin C_j$.
  However, since the interval $I(v)$ contains $x_i$ and~$x_k$ it also
  contains $x_j$, which is a contradiction to the construction of the
  position $x_j$ for the clique $C_j$ since $v$ is a common vertex.
  Note that this is the same argument as used in the proof of
  Theorem~\ref{thm:interval-graph-char}.

  Conversely, we need to show how to construct an interval
  representation from a given linear order of all maximal cliques.
  Assume we have a linear order $O$ of all maximal cliques satisfying
  the conditions of the theorem.  Rename the cliques such that $C_1
  \dots C_{k + \ell}$ is this order, neglecting for a moment from
  which graph the cliques stem.  Let $v$ be a vertex in $\1G$ or $\2G$
  and let $C_i$ and $C_j$ be the leftmost and rightmost clique in $O$
  containing $v$.  Then we define the interval $I(v)$ to be $[i, j]$,
  as in the case of a single graph.  Our claim is that this yields a
  simultaneous interval representation of $\1G$ and $\2G$.  Again, it
  is easy to see that a non integer position $x$ is only contained in
  intervals also containing $\lceil x \rceil$ and $\lfloor x \rfloor$.
  Thus we only need to consider the positions $1, \dots, k + \ell$,
  let $i$ be such an integral position.  Assume without loss of
  generality that $C_i = \{v_1, \dots, v_r\}$ is a clique of $\1G$.
  Then~$i$ is contained in all the intervals $I(v_1), \dots, I(v_r)$
  by definition.  The position $i$ may be additionally contained in
  the interval $I(u)$ for a vertex that is exclusively contained in
  $\2G$ but this does not create an edge between vertices in $\1G$.
  However, there is no vertex $u \notin C_i$ contained in $\1G$ such
  that $i$ is contained in $I(u)$ since this would violate the
  $u$-consecutiveness either of the whole order or of the projection
  to the cliques in $\1G$.  Since the same argument works for cliques
  in $\2G$, all edges in maximal cliques of $\1G$ and $\2G$ are
  represented by the defined interval representation and at the
  integer positions no edges not contained are represented.  Hence,
  this definition of intervals is a simultaneous interval
  representation of $\1G$ and $\2G$.
\end{proof}

With this characterization it is straightforward to formulate the
problem of recognizing simultaneous interval graphs as an instance of
{\sc Simultaneous PQ-Ordering}.  Furthermore, the resulting instance
is so simple that it can be solved in linear time.  Since we want to
represent linear orders instead of circular orders we need to use
rooted PQ-trees instead of unrooted ones.  This can be achieved as
mentioned in the preliminaries about PQ-trees
(Section~\ref{sec:pq-trees}).  Consider an instance of {\sc
  Simultaneous PQ-Ordering} having rooted PQ-trees as nodes.  By
introducing for every PQ-tree a new leaf $\ell$, the special leaf, on
top of the root, unrooting the PQ-tree and setting $\varphi(\ell) =
\ell$ for every arc $(T, T'; \varphi)$ we obtain an equivalent
instance of {\sc Simultaneous PQ-Ordering} having unrooted PQ-trees as
nodes.  Thus, solving {\sc Simultaneous PQ-Ordering} for the case that
the PQ-trees are rooted reduces to the case where the PQ-trees are
unrooted.  Note that the other direction does not work as simple,
since we cannot necessarily find a single leaf $\ell$ contained in
every PQ-tree.  The PQ-trees mentioned in the remaining part of this
section are assumed to be rooted, representing linear orders.

\begin{theorem}
  \label{thm:sim-interval-graph}
  {\sc Simultaneous Interval Representation} can be solved in linear
  time.
\end{theorem}
\begin{proof}
  Let $\1 {\mathcal C} = \{\1 {C_1}, \dots, \1 {C_k}\}$ and $\2
  {\mathcal C} = \{\2 {C_1}, \dots, \2 {C_\ell}\}$ be the maximal
  cliques of $\1G$ and $\2G$ respectively and let $\mathcal C = \1
  {\mathcal C} \cup \2 {\mathcal C}$ be the set of all maximal
  cliques.  We define three PQ-trees $T$, $\1T$ and $\2T$ having
  $\mathcal C$, $\1 {\mathcal C}$ and $\2 {\mathcal C}$ as leaves,
  respectively.  The tree $T$ is defined such that it represents all
  linear orders of $\mathcal C$ that are $v$-consecutive with respect
  to all common vertices $v$.  The trees $\1T$ and $\2T$ are defined
  to represent all linear orders of $\1 {\mathcal C}$ and $\2
  {\mathcal C}$ that are $v$-consecutive with respect to all vertices
  $v$ in $\1G$ and $\2G$, respectively.  Note that $\1T$ and $\2T$ are
  the PQ-trees that would be used to test whether $\1G$ and $\2G$
  themselves are interval graphs.  By the characterization in
  Theorem~\ref{thm:sim-interval-graph-char} it is clear that $\1G$ and
  $\2G$ are simultaneous interval graphs if and only if we can find an
  order represented by $T$ extending orders represented by $\1T$ and
  $\2T$.  Hence $\1G$ and $\2G$ are simultaneous interval graphs if
  and only if the instance $D$ of {\sc Simultaneous PQ-Ordering}
  consisting of the nodes $T$, $\1T$ and $\2T$ and the arcs $(T, \1T)$
  and $(T, \2T)$ has a solution.  This can be checked in quadratic
  time using Theorem~\ref{thm:time-solve-in-quadr-time} since $D$ is
  obviously 1-critical.  Furthermore, normalization can of course be
  done in linear time and the expansion tree of linear size can be
  computed in linear time since expansion stops after a single
  expansion step.  Hence the instance $D$ of {\sc Simultaneous
    PQ-Ordering} can be solved in linear time, which concludes the
  proof.
\end{proof}

\subsection{Extending Partial Interval Representations}
\label{sec:extend-part-interv}

Let $G$ be a graph, $H=(V,E)$ be a subgraph of~$G$ and let $I$ be an
interval representation of~$H$.  The problem {\sc Partial Interval
  Graph Extension} asks, whether there exists an interval graph
representation~$I'$ of~$G$ such that for all $v \in V$ we have that
$I'(v) = I(v)$.  We call an instance~$(G,H,I)$ of {\sc Partial
  Interval Graph Extension} a \emph{partial interval graph}.

Klav\'ik et al.~\cite{ExtendingPartialRepresentations-Klavik.etal(11)}
show that {\sc Partial Interval Graph Extension} can be solved in
time~$O(nm)$, where $n = |V(G)|$ and~$m=|E(G)|$.  We show that {\sc
  Partial Interval Graph Extension} can be reduced in~$O(n+m)$ time to
an instance of {\sc Simultaneous Interval Representation}.  It then
follows from Theorem~\ref{thm:sim-interval-graph} that the partial
interval graph extension problem can be solved in~$O(n+m)$ time.

Without loss of generality, we assume that the endpoints of all
intervals~$I(v), v \in V(H)$ are distinct.  For $v \in V(H)$
let~$\ell(v)$ and~$r(v)$ denote the left and right endpoint of~$I(v)$,
respectively.  Further let~$S(I)$ denote the sequence of these
endpoints in increasing order of coordinate.  We call this order the
\emph{signature} of~$I$.  We say that two interval representations~$I$
and~$I'$ of the same graph~$H$ are \emph{equivalent} if they have the
same signature.  Klav\'ik et
al.~\cite{ExtendingPartialRepresentations-Klavik.etal(11)} show that
{\sc Partial Interval Graph Extension} for a partial interval
graph~$(G,H,I)$ is equivalent to deciding whether there exists an
interval representation~$I'$ of~$G$ whose restriction to~$H$ is
equivalent to~$I$.  In the following we construct an interval graph
$G'$ containing $H$ as an induced subgraph such that every interval
representation of $G'$ induces an interval representation of $H$ that
is equivalent to $I$.

Let~$p_1,\dots,p_{2n}$ denote the interval endpoints of~$I$ in
increasing order.  We now add several intervals to the representation.
Namely, for each point~$p_i$ we put three intervals of length~$\eps$.
The interval~$\ell_i$ is to the left of~$p_i$, interval~$r_i$ is to
the right of~$p_i$ and~$m_i$ contains~$p_i$ and intersects
both~$\ell_i$ and~$r_i$.  We choose $\eps$ small enough so that no two
intervals of distinct points~$p_i$ and~$p_j$ intersect.  We call these
intervals \emph{markers}.  Finally, we add~$2n-1$ \emph{connectors},
where the connector~$c_i$, for~$i=1,\dots,2n-1$ lies strictly
between~$p_i$ and~$p_{i+1}$, and intersects~$r_i$ and~$\ell_{i+1}$;
see Figure~\ref{fig:partial-interv-extension} for an example.  Now
consider the graph~$G'$ given by this interval representation
containing~$H$ as induced subgraph and the new vertices $L_i, M_i,
R_i$ and~$C_i$ corresponding to the intervals $\ell_i, m_i, r_i$
and~$c_i$.  Then $(G, G')$ defines an instance of {\sc Simultaneous
  Interval Representation} corresponding to the instance $(G, H, I)$
of {\sc Partial Interval Graph Extension} and we obtain the following
theorem by showing their equivalence.

\begin{figure}[tb]
  \centering
  \includegraphics[page=3]{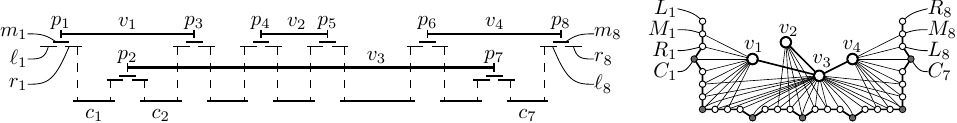}
  \caption{An example graph $H$ containing the vertices $v_1, \dots,
    v_4$ with prescribed interval representation $I$ together with the
    markers $\ell_i, m_i, r_i$ and the connectors $c_i$ on the left.
    The resulting graph $G'$ with the new vertices $L_i, M_i, R_i$ and
    $C_i$ on the right.}
  \label{fig:partial-interv-extension}
\end{figure}

\begin{theorem}
  The problem \textsc{Partial Interval Graph Extension} can be solved
  in linear time.
\end{theorem}
\begin{proof}
  Let $(G, H, I)$ be an instance of {\sc Partial Interval Graph
    Extension} and let $(G, G')$ be the corresponding instance of {\sc
    Simultaneous Interval Representation} as defined above.  We need
  to show that these two instances are equivalent and that~$(G, G')$
  has size linear in the size of~$(G, H, I)$.

  Obviously~$G'$ contains $H$ as an induced subgraph.  We claim that
  in any interval representation~$I'$ of~$G'$ the
  subrepresentation~$I'|_H$ is equivalent to~$I$.  First, note that
  the sequence $L_1, M_1, R_1, C_1, \dots, C_{2n-1}, L_{2n}, M_{2n},
  R_{2n}$ is an induced path in~$G'$.  Hence, in every representation
  of~$G'$ the starting points of their intervals occur either in this
  or in the reverse order.  In particular, the marker
  intervals~$I'(M_i)$ are pairwise disjoint and sorted.  Let~$v_i$
  denote the vertex whose interval has~$p_i$ as an endpoint.
  Since~$M_i$ is adjacent to~$L_i$ and~$R_i$, exactly one of which is
  adjacent~$v_i$, it follows that~$I'(M_i)$ contains an endpoint of
  $I'(v_i)$.  Since this holds for each marker~$M_i$, the claim
  follows.

  With this result the equivalence of the instance $(G, H, I)$ and
  $(G, G')$ is easy to see.  If~$(G,H,I)$ admits an interval
  representation of~$G$, then the above construction shows how to
  construct a corresponding simultaneous representation of~$(G,G')$.
  On the other hand, if~$G$ and~$G'$ admit a simultaneous interval
  representation, then the endpoints of the intervals corresponding to
  vertices of~$H$ must occur in the same order as in~$I$, and hence
  the interval representation of~$G$ extends~$I$.
  
  It remains to show that~$G'$ has size linear in the size of~$H$.  To
  this end, we revisit the construction of~$G'$ from~$H$.  Let~$H'$ be
  the subgraph of~$G'$ obtained by removing the vertices corresponding
  to connectors.  We first show that the size of~$H'$ is linear in the
  size of~$H$.

  Clearly, $H'$ contains exactly six additional vertices for each
  vertex of~$H$ (three for each endpoint of an interval representing a
  vertex of~$H$), and thus $|V(H')| = 7n$.  Now consider the edges
  of~$H'$.  We denote by~$I(p)$ the set of vertices whose intervals
  contain~$p$ in the \emph{interior}.  Let again~$p_1,\dots,p_{2n}$
  denote the endpoints of the intervals in the interval
  representation~$I$ of~$H$.  Recall that for each such endpoint we
  add three vertices, which are represented by the intervals~$\ell_i,
  m_i$ and~$r_i$, respectively.  Note that the endpoints~$p_{i-1}$
  and~$p_{i+1}$ (if they exist) lie to the left of~$\ell_i$ and to the
  right of~$r_i$, respectively, and hence do not intersect with these
  intervals.  The neighbors of~$L_i, M_i$ and~$R_i$ belonging to~$H$
  are contained in~$I(p_i) \cup \{v_i\}$.  This implies that the
  degree of~$L_i$,~$M_i$ and~$R_i$ is linear in the degree of~$v_i$
  in~$H$, and hence the total number of edges in~$H'$ is linear in
  $|E(H)|$.

  For the step from~$H'$ to~$G'$, we add the connectors.  Consider the
  $i$th connector~$C_i$, which is adjacent to~$R_i$ and~$L_{i+1}$.
  Since no other intervals start or end in between, the vertex
  corresponding to the connector~$C_i$ is adjacent to the same
  vertices as~$R_i$ and~$L_{i+1}$.  Thus, the size of~$G'$ is linear
  in the size of~$H'$ and the claim follows.  Moreover, it is clear
  that assuming the intervals of~$I$ are given in sorted order,
  then~$G'$ can be constructed from~$G$ in~$O(n+m)$ time.
\end{proof}

\subsection{Generalization to Non-Biconnected Graphs}
\label{sec:thoughts-gener}

The reason why our solutions for {\sc Partially PQ-Constrained
  Planarity} and {\sc SEFE} are restricted to the case where the
graphs are biconnected is that the set of possible orders of edges
around a cutvertex may not be PQ-representable.  However, this is not
really necessary.  Assume we have a representation of all embeddings
of a planar graph as instance of {\sc Simultaneous PQ-Ordering} with
the following two properties.  First, this instance contains a PQ-tree
$T(v)$ for every vertex $v$ having the edges incident to $v$ as
leaves.  Second, this instance remains 1-critical even if we introduce
an additional child to $T(v)$.  If this is the case, {\sc Simultaneous
  PQ-Ordering} can be solved by introducing the constraint tree
$T'(v)$ as child of $T(v)$.  Similarly, in the setting of {\sc SEFE}
common edges around a vertex $v$ can be enforced to be ordered the
same by introducing a common embedding tree $T(v)$ having the common
edges incident to $v$ as leaves as child of the trees $T(\1v)$ and
$T(\2v)$, where $T(\1v)$ and $T(\2v)$ have the edges incident to $v$
in $\1G$ and $\2G$ as leaves, respectively.  We show that all
embeddings can be represented by such an instance for the special case
that every cutvertex is contained in only two blocks.  Furthermore,
this extends to the case where each block containing the cutvertex $v$
consists of a single edge except for up to two blocks.

Consider a cutvertex $v$ that is contained in two blocks $B_1(v)$ and
$B_2(v)$ and let $E_1(v)$ and $E_2(v)$ be the edges incident to $v$
contained in $B_1(v)$ and $B_2(v)$, respectively.  As before, the
orders of $E_1(v)$ around $v$ that can occur in a planar drawing can
be represented by a PQ-tree $T_1(v)$ with $E_1(v)$ as leaves; call
$T_1(v)$ the \emph{block embedding tree} with respect to $B_1$.  Let
$T_2(v)$ be the block embedding tree of $v$ with respect to the second
block $B_2$.  It is clear that in a planar drawing of the whole graph
the edges $E_1(v)$ (and with it also $E_2(v)$) appear consecutively
around $v$.  This condition can be formulated independently from the
PQ-trees $T_1(v)$ and $T_2(v)$ by an other PQ-tree $T(v)$ consisting
of two P-nodes $\mu_1$ and $\mu_2$ with the edge $\{\mu_1, \mu_2\}$
and leaves $E_1(v)$ and $E_2(v)$ attached to $\mu_1$ and $\mu_2$,
respectively.  It is clear that the instance of {\sc Simultaneous
  PQ-Ordering} consisting of the PQ-tree $T(v)$ with $T_1(v)$ and
$T_2(v)$ as children represents all possible circular orders of edges
around $v$ in the sense that in every planar embedding the order of
edges around $v$ induces a solution of this instance and vice versa;
Figure~\ref{fig:ext-pq-emb-rep} depicts this instance of {\sc
  Simultaneous PQ-Ordering}.  We call the PQ-tree $T(v)$ the
\emph{combined embedding tree} of $v$.  Of course the order of edges
around each vertex cannot be chosen independently, but since the block
embedding trees are embedding trees of biconnected components the P-
and Q-nodes stem from P- and R-nodes in the SPQR-tree and we can again
ensure consistency by introducing the consistency trees for each
block.  This yields an extension of the PQ-embedding representation to
the case that $G$ may contain cutvertices that are contained in two
blocks.  Note that the embedding tree of a vertex that is not a
cutvertex can be seen as combined and block embedding tree at the same
time.

\begin{figure}[tb]
  \centering
  \includegraphics[page=1]{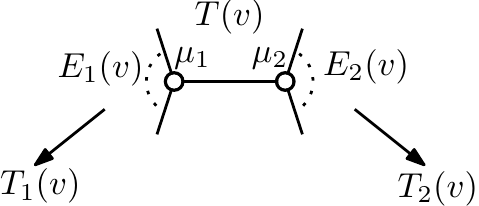}
  \caption{Representation of the possible orders of edges around a
    cutvertex $v$ for the special case that $v$ is contained in two
    blocks in terms of an instance of {\sc Simultaneous PQ-Ordering}.}
  \label{fig:ext-pq-emb-rep}
\end{figure}

It is easy to see that this representation satisfies the conditions
mentioned above.  First, the combined embedding tree $T(v)$ has the
edges incident to $v$ as leaves.  Second, if an additional child is
introduced to every combined embedding tree the instance remains
2-fixed, which can be seen as follows.  The combined embedding tree
has three children, the two block embedding trees and the additional
child.  However, each P-node in $T(v)$ is fixed with respect to only
one of the block embedding trees, thus it is 2-fixed.  Every P-node in
the block embedding trees is fixed with respect to one child, the
corresponding consistency tree, thus it is 2-fixed since it has the
combined embedding tree as parent.  The P-nodes in consistency trees
are also 2-fixed, since they have two 2-fixed parents.  Hence we
obtain a 2-fixed instance, if we use this extended PQ-embedding
representation to formulate {\sc Partially PQ-Constrained Planarity}
or {\sc SEFE} as instance of {\sc Simultaneous PQ-Ordering}.
Furthermore, the runtime analysis yielding linear time for {\sc
  Partially PQ-Constrained Planarity} in
Theorem~\ref{thm:plan-emb-with-pq-const-biconn-lin-time} works
analogous.

Assume now that all blocks containing $v$ consist of a single edge
except for up to two blocks $B_1$ and $B_2$.  A block consisting of a
single edge is identified with this edge and called \emph{bridge}.  It
is clear that each bridge can be attached arbitrarily to an embedding
of $B_1 + B_2$.  Hence we can modify the above defined extension of
the PQ-embedding representation by introducing a single P-node
containing all edges incident to $v$ as parent of the combined
embedding tree.  The analysis from above works analogously yielding
the following two theorems.

\begin{theorem}
  {\sc Partially PQ-Constrained Planarity} can be solved in linear
  time, if each vertex is contained in up to two blocks not consisting
  of a single edge.
\end{theorem}

\begin{theorem}
  {\sc Simultaneous Embedding with Fixed Edges} can be solved in
  quadratic time, if in both graphs every vertex is contained in at
  most two blocks not consisting of a single edge and the common graph
  is connected.
\end{theorem}
Note that this special case always applies if the cutvertices have
degree at most~5.  In particular, for \textsc{SEFE} we obtain the
following corollary.

\begin{cor}
  {\sc Simultaneous Embedding with Fixed Edges} can be solved in
  quadratic time for maxdeg-5 graphs whose intersection is connected.
\end{cor}

\section{Conclusion}
\label{cha:conclusion}

In this work we introduced a new problem called {\sc Simultaneous
  PQ-Ordering}.  It has as input a set of PQ-trees with a child-parent
relation (a DAG with PQ-trees as nodes) and asks, whether for every
PQ-tree a circular order can be chosen such that it is an extension of
the orders of all its children.  This was motivated by the possibility
to represent the possible circular orders of edges around every vertex
of a biconnected planar graph by a PQ-tree.  Unfortunately, {\sc
  Simultaneous PQ-Ordering} turned out to be $\mathcal {NP}$-complete
in general.  However, we were able to find an algorithm solving {\sc
  Simultaneous PQ-Ordering} in polynomial time for ``simple''
instances, the 1-critical instances.  To achieve this result we showed
that satisfying the Q-constraints and the critical triples is
sufficient to extend orders of several children simultaneously to a
parent, if each P-node is contained in at most one critical triple.
We were able to ensure that a critical triples are satisfied
automatically when choosing orders bottom-up by inserting new
PQ-trees, the expansion trees.  Creating the expansion trees
iteratively for every critical triple led to the expansion graph that
turned out to have polynomial size for 1-critical instances.  Hence,
we are able to solve a 1-critical instance of {\sc Simultaneous
  PQ-Ordering} in polynomial time, essentially by choosing orders
bottom-up in the expansion graph.  We have shown how this framework
can be applied to solve {\sc Partially PQ-Constrained Planarity} for
biconnected graphs and {\sc Simultaneous Embedding with Fixed Edges}
for biconnected graphs with a connected intersection in polynomial
time (linear and quadratic, respectively), which were both not known
to be efficiently solvable before.  Furthermore, we have shown how to
solves {\sc Simultaneous Interval Representation} and {\sc Partial
  Interval Graph Extension} in linear time, which improves over the
best known algorithms with running times $\mathcal O(n^2 \log n)$ and
$\mathcal O(nm)$ algorithm, respectively.  We stress that all these
results can be obtained in a straightforward way from the main result
of this work, the algorithm for {\sc Simultaneous PQ-Ordering} for
2-fixed instances.

\paragraph{Open problems.} However, several questions remain open for
the applications as well as for problems related to {\sc Simultaneous
  PQ-Ordering}.  Since the set of possible orders of edges around a
cutvertex in a planar drawing is not necessarily PQ-representable our
solutions for {\sc Partially PQ-Constrained Planarity} and {\sc SEFE}
cannot handle graphs containing cutvertices, except for the special
cases discussed in Section~\ref{sec:thoughts-gener}.  Another
limitation in the case of {\sc SEFE} is that the common graph needs to
be connected.  Since {\sc Simultaneous PQ-Ordering} focuses on very
local conditions for every vertex, it is difficult to formulate
conditions concerning the relative positions of different connected
components in terms of such conditions, at least if we need to ensure
that the resulting instances are 1-critical.  It seems worthwhile to
investigate the problem of ensuring consistent relative positions in
the absence of other embedding constraints, e.g., if the intersection
of the two graphs is a set of disjoint cycles.  For {\sc Simultaneous
  Interval Representation} the complexity is still open for the case
where more than two graphs are allowed.

One approach to address these problems is to extend the results on
{\sc Simultaneous PQ-Ordering} to instances that are not 1-critical or
generalize it in the sense that structures different to PQ-trees are
used as nodes in the DAG.  Questions forming the basis of such an
approach could be of the following kind.  Given three PQ-trees having
some leaves in common, can we find an order for each of the trees such
that the three resulting orders can be extended to a common order?
Note that testing this for three fixed orders can be done efficiently.
Does it make the problem easier if we consider rooted PQ-trees
representing linear orders?  Can we find other structures representing
sets of orders that are closed with respect to intersection?  Can the
possible orders of edges around cutvertices be represented by such
structures?

\bibliographystyle{amsalpha} \bibliography{SimPQOrd}

\end{document}